\documentclass[letter,11pt]{article}
\usepackage[margin=1in]{geometry}
\usepackage[utf8]{inputenc}
\usepackage{times}

\usepackage[subtle]{savetrees}

\usepackage[utf8]{inputenc} 
\usepackage[T1]{fontenc}    
\usepackage{hyperref}       
\usepackage{url}            
\usepackage{booktabs}       
\usepackage{amsfonts}       
\usepackage{nicefrac}       
\usepackage{microtype}      

\usepackage{amsthm,amsmath,amssymb} 
\usepackage{xspace}
\usepackage[vlined,linesnumbered,ruled]{algorithm2e}
 
\numberwithin{equation}{section}

\usepackage{enumerate}
\usepackage{listings}
\usepackage{tikz}
\usepackage{csquotes}
\usepackage{cite}
\usepackage{subfigure}
\usepackage{thmtools,thm-restate} 
\usepackage{paralist} 
\usepackage{pgfplots}
\pgfplotsset{compat=newest}
\usepgfplotslibrary{groupplots}

\usepackage{graphics,graphicx,xcolor,color,url}    
\usetikzlibrary{arrows,shapes,snakes,automata,backgrounds,petri}
\tikzset{
  cross/.style={cross out, draw=black, minimum size=2*(#1-\pgflinewidth), inner sep=0pt, outer sep=0pt},
  cross/.default={3pt},
  vertex/.style={circle, draw, fill=black, inner sep=0pt, minimum width=4pt},
}

\newcommand{\drawpred}[2]{
  \draw (#1, #2) node[cross,red] {};
}
\newcommand{\drawreal}[2]{
  \draw[black!40!green] (#1, #2) circle (3pt);
}
\newcommand{\interval}[4]{
  \draw (#2, #4) node[anchor=east]{#1} -- (#3, #4);
  \draw (#2, #4-0.1) -- (#2, #4+0.1);
  \draw (#3, #4-0.1) -- (#3, #4+0.1);
}
\newcommand{\intervalp}[5]{
  \interval{#1}{#2}{#3}{#4}
  \drawpred{#5}{#4}
}
\newcommand{\intervalr}[5]{
  \interval{#1}{#2}{#3}{#4}
  \drawreal{#5}{#4}
}
\newcommand{\intervalpr}[6]{
  \interval{#1}{#2}{#3}{#4}
  \drawpred{#5}{#4}
  \drawreal{#6}{#4}
}

\SetKw{KwLet}{let}
\SetKw{KwAnd}{and}

\theoremstyle{theorem}
\newtheorem{theorem}{Theorem}[section]
\newtheorem{lemma}[theorem]{Lemma}
\newtheorem{coro}[theorem]{Corollary}

\newtheorem{claim}[theorem]{Claim}
\newtheorem{definition}[theorem]{Definition}
\newtheorem{obs}[theorem]{Observation}

\usepackage{thm-restate}
\newtheorem{lem}[theorem]{Lemma}

\newcommand{\pred}[1]{\overline{#1}}

\newcommand{\ALG}{\mathrm{ALG}}
\newcommand{\OPT}{\mathrm{OPT}}
\newcommand{\opt}{\mathrm{opt}}

\newcommand{\w}{\ensuremath{\overline{w}}} 

\newcommand{\ZZ}{\mathbb{Z}}

\newcommand{\RR}{\mathbb{R}}

\newcommand{\eps}{\ensuremath{\varepsilon}\xspace} 
\newcommand{\sym}{\ensuremath{\Delta}\xspace} 
\DeclareMathOperator{\EX}{\mathbb{E}}
\newcommand{\ud}{\ensuremath{D}}
\newcommand{\hs}{\ensuremath{\mathcal{H}}}

\graphicspath{ {figs/} }
\title{
Learning-Augmented Query Policies
\thanks{Research partially supported by EPSRC grants EP/S033483/1 and EP/T01461X/1, and by the German Science Foundation (DFG) under contract ME~3825/1.}
}

\author{%
  Thomas Erlebach\thanks{Durham University, Department of Computer Science, \texttt{thomas.erlebach@durham.ac.uk}} \and Murilo S.\ de Lima\thanks{K\'{o}pavogur, Iceland, \texttt{mslima@ic.unicamp.br}. Work done while employed at University of Leicester.} \and Nicole Megow\thanks{University of Bremen, Faculty of Mathematics and Computer Science, \texttt{\{nicole.megow,jschloet\}@uni-bremen.de}} \and Jens Schl{\"o}ter\footnotemark[3]
}

\begin{document}

\maketitle

\begin{abstract}

We study how to utilize (possibly machine-learned) predictions in a model for computing under uncertainty in which an algorithm can query unknown data. The goal is to minimize the number of queries needed to solve the problem. We consider fundamental problems such as finding the minima of intersecting sets of elements or sorting them (these problems can also be phrased as (hyper)graph orientation problems), as well as the minimum spanning tree problem. We discuss different measures for the prediction accuracy and design algorithms with performance guarantees that improve with the accuracy of predictions and that are robust with respect to very poor prediction quality. These measures are intuitive and might be of general interest for inputs involving uncertainty intervals. We show that our predictions are PAC learnable. We also provide new structural insights for the minimum spanning tree problem that might be useful in the context of explorable uncertainty regardless of predictions. Our results prove that untrusted predictions can circumvent known lower bounds in the model of explorable uncertainty. We complement our results by experiments that empirically confirm the performance of our algorithms. 
	 \end{abstract} 

\section{Introduction}
Dealing with uncertainty
is a common challenge
in many real-world settings. The research area
of \emph{explorable uncertainty}
\cite{kahan91queries,erlebach15querysurvey}
considers such scenarios assuming that, for
any uncertain input element, a \emph{query} can be used to obtain
the exact value of that element. The uncertain input value
is often represented by an interval that
contains the exact
value, and a query returns that exact value.
Queries are costly, and hence the goal is to make as few queries as possible
until sufficient information has been obtained to solve the given
problem. The major challenge is to balance the resulting exploration-exploitation tradeoff. 
For all problems that we consider, there exist input instances that are impossible to solve without querying the entire input. Therefore, instead of aiming to derive absolute bounds on the number of queries required for an input of size~$n$ in the worst case, we use 
competitive analysis to compare the number of queries
made by an algorithm with the
minimum number of queries among all feasible solutions, i.e., we aim for query-competitive algorithms.

In this query model, we consider very fundamental problems  that underlie numerous
applications: sorting, computing the minimum element, and computing
a minimum spanning tree in a graph with uncertain edge weights.
These problems
are well understood in the setting of explorable uncertainty: The best
known deterministic algorithms are $2$-competitive and no deterministic algorithm can be better \cite{erlebach08steiner_uncertainty,megow17mst,halldorsson19sortingqueries,kahan91queries}.
For the sorting and minimum problems,
we consider the setting where we want to solve the problem
for a number of different, possibly overlapping subsets of a
given ground set of uncertain elements. Such settings can be
motivated e.g.\ by distributed database caches~\cite{olston2000queries}
where one wants to answer database queries using cached local data
and a minimum number of queries to the master database.
Interestingly, these problems can also be phrased as (hyper)graph orientation problems, where the goal is to orient each (hyper)edge towards its minimum-weight vertex~\cite{BampisDEdLMS21}.
The minimum spanning tree (MST) problem is one of the most fundamental combinatorial
problems.
It is among the most widely studied problems in computing with
explorable uncertainty and has been a cornerstone in the
development of algorithmic approaches as well as lower
bound techniques \cite{erlebach08steiner_uncertainty,erlebach15querysurvey,megow17mst}.

Instead of assuming that no information about an uncertain
value is available except for the interval in which it is contained,
we study
a setting 
where predictions for the uncertain values are available.
For example, machine learning (ML) methods could be used to predict the value
of an interval.
Given the tremendous progress in artificial intelligence and ML in recent decades, we can expect that those predictions are of good accuracy but there is no guarantee and the predictions might be completely wrong.
This lack of provable performance guarantees for ML
often causes concerns regarding how confident one can be that
an ML algorithm will perform sufficiently well in all circumstances.
In many settings, e.g., safety-critical applications,
having provable performance guarantees is highly desirable or
even obligatory.
It is a very natural question
whether the availability of such (ML) predictions can be exploited by query 
algorithms for computing with explorable uncertainty.
Ideally, an algorithm should perform very well if
predictions are accurate, but even if they are arbitrarily
wrong, the algorithm should not perform worse than an algorithm
that handles the problem without access to predictions.
To emphasize
that the predictions can be wrong, we refer to them
as \emph{untrusted predictions}.

We note that the availability of both uncertainty intervals and
untrusted predictions is natural in many scenarios. For example,
a known past location and maximum movement speed of a mobile node
yield an uncertainty interval that is guaranteed to contain the
current location, while a ML method may predict
the node's precise location based on past movement data.
Similarly, in a distributed database system where the master database
updates a value in the local database only if the new value is outside
a fixed interval around the previously stored value, we have an
interval that is guaranteed to contain the current value while
a ML method may predict the precise current value
based on past time series data.

We study for the first time the combination of
explorable uncertainty and untrusted predictions. Our work draws inspiration
from recent work that has considered untrusted predictions in the context
of online algorithms, where the input is revealed to an algorithm incrementally
and the algorithm must make decisions without knowledge of future inputs.
We adopt the notions of $\alpha$-consistency
and $\beta$-robustness~\cite{lykouris2018competitive,purohit2018improving}: 
An algorithm is $\alpha$-consistent if it
is $\alpha$-competitive when the predictions are correct, and it is
$\beta$-robust if it is $\beta$-competitive no matter how wrong the
predictions are. Furthermore, we are interested in a smooth transition
between the case with correct predictions and the case with arbitrarily
incorrect predictions: We aim for performance guarantees
that degrade gracefully with the \emph{amount} of prediction error.
This raises interesting questions regarding appropriate ways of
measuring prediction errors, and we explore several such measures.

Our results show that, in the setting of explorable uncertainty,
it is in fact possible to exploit ML predictions of the uncertain
values in such a way that the performance of an algorithm is improved
when the predictions are good, while at the same time a strong bound
on the worst-case performance can be guaranteed even when the predictions
are arbitrarily bad. Following this approach, ML can thus be embedded within a
system to improve the typical performance while still maintaining
provable worst-case guarantees. In this way users can be shielded from
occasional failures of the ML algorithms.
Therefore, our approach contributes methods and analysis techniques
that can help to address the challenge of building trustable AI systems.

\paragraph{Main results} We show how to
utilize (possibly machine-learned) predictions in a query-based model for computing under uncertainty and prove worst-case performance guarantees. Our major contribution is twofold: 
\begin{inparaitem}
\item[(a)] we prove that untrusted predictions can circumvent lower bounds in the context of explorable uncertainty, and 
\item[(b)] we provide new structural insights for previously studied, fundamental problems that might be useful regardless of predictions.
\end{inparaitem}
Finally, we conduct experiments that show that our algorithms are practical and that confirm the theoretical improvement.

For the problems of {sorting or identifying minima in intersecting sets}
and finding an MST,
we give algorithms that are $1.5$-consistent
and $2$-robust, and show that this is the best possible consistency when aiming for optimal robustness. 
It is worth noting that our algorithms achieve the
improved competitive ratio of $1.5$ in case of accurate predictions
while maintaining the worst-case ratio of~$2$. This is in contrast
to work on other online problems with predictions where the exploitation
of predictions usually incurs an increase in the worst-case
ratio (see, e.g, \cite{purohit2018improving,AntoniadisGKK20}).
We also give a parameterized {robustness-consistency} tradeoff.
Our major focus lies on a more fine-grained performance analysis
with guarantees that improve with the accuracy of the predictions.
We compare three different measures $k_{\#}, k_h, k_M$ for the prediction accuracy. The number of inaccurate predictions $k_{\#}$ is
too crude to allow for a performance improving on the lower bounds of~$2$ for the setting without predictions~\cite{erlebach08steiner_uncertainty,kahan91queries}.
We propose two measures that take structural insights about uncertainty intervals into account, the hop distance $k_h$ and the mandatory query distance $k_M$. The latter can be proven to be more restrictive,~i.e.,~$k_M\leq k_h$. We give proper definitions later. While the hop distance $k_h$ is very intuitive and possibly of general interest, the mandatory query distance $k_M$ is tailored to problems with explorable uncertainty.
We show that the predictions are PAC-learnable w.r.t.~both error measures, $k_h$ and~$k_M$.

For the problem of identifying the minima in intersecting sets, we provide an algorithm with competitive ratio $\min\{(1+\frac{1}{\gamma-1})(1+ \frac{k_M}{\opt}), \gamma \}$, for any integral $\gamma\geq 2$. Here, $\opt$ is the minimum number of queries in an offline solution; precise definitions are given below. This is best possible for $k_M{=0}$ {and large $k_M$}. With respect to the hop distance, we achieve the stronger bound $\min\{(1+\frac{1}{\gamma})(1+ \frac{k_h}{\opt}), \gamma\}$, for any integral $\gamma \geq 2$, which is also best possible for $k_h = 0$ and large $k_h$. It is not difficult to see that the sorting problem can be reduced to the minimum problem by creating a set for each pair of elements that are in the same set of the sorting instance, so these bounds also apply to the sorting problem. For the special case of sorting a single set, we obtain an algorithm with competitive ratio $\min\{1+k/\opt, 2\}$ for any of the considered accuracy measures, which is best possible.
Finding the MST under uncertainty is the combinatorially most challenging problem. As our main result, we give an algorithm with competitive ratio $\min\{ 1+ \frac{1}{\gamma} +  (5 + \frac{1}{\gamma}) \cdot \frac{k_h}{\opt}, \max\{3,\gamma + \frac{1}{\opt}\}\}$, for any integral $\gamma \geq 2$.
All our algorithms have polynomial running time except the algorithms for the minimum problem, which may involve solving an NP-hard vertex cover problem. We justify this complexity by showing that even the offline variant of the minimum problem is NP-hard.

Omitted proofs are provided in the appendix.

\paragraph{Further related work}
There is a long history of research on the tradeoff between exploration and exploitation when coping with uncertainty in the input data. Stochastic models are often assumed,
e.g., in work on multi-armed bandits~\cite{Thompson33,BubeckC12,GittinsGW11-book}, Weitzman's Pandora's box problem \cite{Weitzman1979}, and more recently query-variants of combinatorial optimization problems; {see, e.g.,
\cite{singla2018price,gupta2019markovian}, and specific problems such as stochastic knapsack~\cite{DeanGV08,Ma18}, orienteering~\cite{GuptaKNR15,BansalN15}, matching~\cite{ChenIKMR09,BansalGLMNR12,BlumDHPSS20,BehnezhadFHR19,AssadiKL19}, and probing problems \cite{AdamczykSW16,GuptaN13,GuptaNS16}.}

In our work, we assume no knowledge of stochastic information and aim for robust algorithms that perform well even in a worst case. This line of research on (adversarial) explorable uncertainty has been initiated by Kahan~\cite{kahan91queries} in the context of selection problems.
In particular, he showed for the problem of identifying all maximum elements of a set of uncertain
values that querying the intervals in order of non-increasing right endpoints
requires at most one more query {than} the optimal query set.
Subsequent work addressed finding the~$k$-th smallest value in a set of uncertainty intervals~\cite{gupta16queryselection,feder03medianqueries},
caching problems
\cite{olston2000queries}, computing a function value~\cite{khanna01queries}, sorting~\cite{halldorsson19sortingqueries}, and
combinatorial optimization problems, such as shortest path~\cite{feder07pathsqueires}, the knapsack problem~\cite{goerigk15knapsackqueries}, scheduling problems~\cite{DurrEMM20,arantes18schedulingqueries,albersE2020}, the MST problem and matroids~\cite{erlebach08steiner_uncertainty,erlebach14mstverification,megow17mst,focke17mstexp,MerinoS19}. 

Most related to our work are previous results on the MST problem and sorting with explorable uncertainty.
For the MST problem with uncertain edge weights represented by open intervals, a $2$-competitive
deterministic algorithm was presented and shown to be best possible~\cite{erlebach08steiner_uncertainty}.
The algorithm is based on the concept of \emph{witness sets}, i.e., sets of uncertain elements
with the property that any feasible query set must query at least one element of the set. The algorithm
from~\cite{erlebach08steiner_uncertainty} repeatedly
identifies a witness set of size~$2$ that corresponds to two candidates for the maximum-weight edge
in a cycle of the given graph, and queries both its elements. It is also known that randomization admits 
an improved competitive ratio of $1.707$ for the MST problem with uncertainty~\cite{megow17mst}. {Both, a deterministic $2$-competitive algorithm and a randomized $1.707$-competitive algorithm, are known for the more general problem of finding the minimum base in a matroid~\cite{erlebach16cheapestset,megow17mst}, even for the case with non-uniform query costs~\cite{megow17mst}.}
For sorting a single set, a $2$-competitive algorithm exists (even
with arbitrary query costs) and is
best possible~\cite{halldorsson19sortingqueries}. In the case of uniform query costs, the algorithm simply queries witness
sets of size~$2$; in the case of arbitrary costs, it first queries a minimum-weight vertex cover of the interval
graph corresponding to the instance and then executes any remaining queries that are still necessary.
For uniform query cost, the competitive ratio can be improved to~$1.5$ using randomization~\cite{halldorsson19sortingqueries}.

Our work is the first to consider explorable uncertainty in the recently proposed framework of online algorithms using (machine-learned) predictions \cite{medina2017revenue,purohit2018improving,lykouris2018competitive}. After work on revenue optimization~\cite{medina2017revenue} and online caching~\cite{lykouris2018competitive}, Kumar et al.~\cite{purohit2018improving} studied online algorithms with respect to consistency and robustness in the context of classical online problems, ski-rental and non-clairvoyant scheduling. They also studied the performance as a function of the prediction error. This work initiated a vast growing line of research. Studied problems include rent-or-buy problems \cite{purohit2018improving,GollapudiP19,WeiZ20}, revenue optimization~\cite{medina2017revenue},
	scheduling and bin packing~\cite{purohit2018improving,angelopoulos2019online,moseley2020online,Mitzenmacher20,AzarLT2021,LattanziLMV20}, caching and metrical task systems~\cite{lykouris2018competitive,rohatgi2019near,AntoniadisCEPS2020,Wei20}, matching~\cite{KumarPSSV19,AntoniadisGKK20} and secretary problems~\cite{DuettingLLV21,AntoniadisGKK20}. Very recently and in a similar spirit as our work, Lu et al.~\cite{LuRSZ2021} studied a generalized sorting problem with additional predictions. Their model strictly differs from ours, as they focus on bounds for the absolute number of pair-wise comparisons whereas we aim for query-competitive algorithms.
Overall, learning-augmented online optimization is a highly topical concept which has not yet been studied in the explorable uncertainty model.

There is a significant body of work on computing in models where information
about a hidden object can be accessed only via queries. The 
hidden object can for example be a function, a matrix, or a graph.
In the graph context, property testing~\cite{Goldreich2017} has been studied extensively
since the early 1990s. A typical problem is to decide
whether a given graph has a certain property or is ``far'' from having that
property using a small (sublinear or even constant) number of queries that look
up entries of the adjacency matrix of the graph.
Many other types of queries have
been studied (see e.g.~\cite{Mazzawi2010,Beame2018,Rubinstein2018,Chen2020,Assadi2020} and many more).
Such work has often considered graph reconstruction problems or
parameter estimation problems (e.g., estimating the number of edges).
The bounds on the number of queries made by an algorithm that have been shown
in these problems are usually absolute, i.e., given as a function of the input size,
but independent of the input graph itself, and the resulting correctness guarantees
are often probabilistic.

In contrast to much of the work on algorithms with query access to a hidden
object, we evaluate our algorithms in an instance-dependent manner: For each
input, we compare the number of queries made by an algorithm with the best
possible number of queries \emph{for that input}, using competitive analysis.
In computing with uncertainty, there are typically inputs where even an
optimal query set has to query essentially the whole input in order to be
able to solve the problem, hence absolute bounds on the number of queries
depending only on the size of the input would often be trivial. The goal
is hence to devise algorithms that use, on each input, a number of queries
that is not much larger than the optimal query set for that input.

\section{{Definitions, accuracy of predictions, and lower bounds}}
\label{sec:prelim}
\paragraph{Problem definitions} In the \emph{minimum problem under uncertainty}, we are given a set $\mathcal{I}$ of $n$ uncertainty intervals with a predicted value $\pred{w}_i\in I_i$ for each $I_i \in \mathcal{I}$, and a family $\mathcal{S}$ of $m$ subsets of $\mathcal{I}$.
The \emph{true} value of interval $I_i$ is denoted by $w_i$ and can be revealed by a \emph{query}.
The task is to identify for each $S \in \mathcal{S}$ the element with the minimum true value. Note that this may not require to determine the actual value of this element.

The {\em sorting problem under uncertainty} is closely related to the minimum problem. For the same input, the task is to sort, for each set $S \in \mathcal{S}$, the intervals in non-decreasing order of their true values.
	
In the {\em minimum spanning tree (MST) problem under uncertainty}, we are given a graph $G=(V,E)$, with uncertainty intervals $I_e$
and predicted values $\pred{w}_e \in I_e$ for the weight of each edge $e \in E$. A minimum spanning tree (MST) is a tree that connects all vertices of $G$ at a minimum total edge weight.
The task is to find an MST with respect to the true values of the edge weights.

In all three problems, the goal is to solve the task using a minimum number of queries. 
Note that the exact value of a solution, i.e., the minimum value or the weight of the MST, does not need to be determined. 
We further assume that each uncertainty interval is either trivial or open, i.e., $I_i = (L_i,U_i)$ or $I_i = \{w_i\}$,  
as otherwise a simple lower bound of $n$ on the competitive ratio exists for the minimum and MST problems~\cite{gupta16queryselection}.

Further, we study \emph{adaptive
strategies} that make queries sequentially and utilize the results of previous steps to decide upon the next query.
We impose no time/space complexity constraints on the algorithms, as we are interested in understanding the competitive ratio of the problems.
We assume the algorithms never query intervals that are trivial or that were previously queried.
A set $W$ of queries is called a \emph{witness set} \cite{bruce05uncertainty,erlebach08steiner_uncertainty} if every feasible solution (i.e., every set of queries that solves the problem) contains at least one query in~$W$.

\paragraph{Competitive analysis}
We employ competitive analysis and compare the outcome of our algorithms with the best offline solution, i.e., the minimum number of queries needed to verify a solution when all values are known in advance. We call the offline variants of our problems \emph{verification problems}.
{By $\OPT$ we denote an arbitrary optimal query set for the verification problem, and by $\opt$ its cardinality.
For an algorithm for the online problem, we denote by $\ALG$ the set of queries it makes and by $|\ALG|$ the cardinality of that set.}
An algorithm is $\rho$-{\em competitive} if it executes, for any problem instance, at most $\rho \cdot {\opt}$ queries.
Further, we quantify the performance of our algorithms depending on the quality of predictions. For the extreme cases, we say that an algorithm is $\alpha$-\emph{consistent} if it is $\alpha$-competitive if the predictions are correct, and $\beta$-\emph{robust} if it is $\beta$-competitive if the predictions are inaccurate. 

Clearly, an algorithm that assumes the predicted values to be correct and solves the verification problem is $1$-consistent. However, such an algorithm may have an arbitrarily bad performance if the predictions are incorrect. Similarly, the known deterministic $2$-competitive algorithms for the online problems without predictions~\cite{erlebach08steiner_uncertainty, kahan91queries} are $2$-robust and $2$-consistent. The known lower bounds of $2$ rule out any robustness factor less than $2$ for our model. They build on the following simple example with two intersecting intervals $I_a,I_b$. No matter which interval a deterministic algorithm queries first, say $I_a$, the realized value could be $w_a\in I_a\cap I_b$, which requires a second query. If the adversary chooses $w_b\notin I_a\cap I_b$, querying just $I_b$ would have been sufficient to identify the minimal interval.

We give a bound on the best~achievable~tradeoff between consistency and robustness. Later, we will  provide algorithms with matching performance guarantees.

\begin{restatable}{theorem}{ThmLBTradeoffWithoutError}
\label{theo_minimum_combined_lb}
  Let $\beta \geq 2$ be a fixed integer.
  For the minimum (even in a single set), sorting and MST problems under uncertainty, there is no deterministic $\beta$-robust algorithm that is $\alpha$-consistent for $\alpha < 1 + \frac{1}{\beta}$. {And vice versa, no deterministic $\alpha$-consistent algorithm, with $\alpha>1$, is $\beta$-robust for $\beta < \max\{\frac{1}{\alpha-1},2\}$.}
\end{restatable}

\paragraph{Accuracy of predictions}
We aim for a more fine-grained performance analysis giving guarantees that depend on the quality of predictions.
A very natural, simple error measure is the number of inaccurate predictions $k_\#=|\{I_i \in \mathcal{I}\,|\, w_i \not= \w_i\}|$. However, we show that for $k_{\#} \geq 1$ the competitive ratio cannot be better than the known lower bounds of~$2$.
The reason for the weakness of this measure is that it completely ignores 
the interleaving  structure of intervals. (Similarly, using an $L_1$ error metric such as
$\sum_{I_i\in \mathcal{I}}|w_i - \w_i|$ would not be meaningful because only
the \emph{order} of the values and the interval endpoints matters for our problems.)
To address this, we consider {two refined measures for the predictor quality}.

 \emph{Hop distance.} This metric captures naturally the relation between a predicted and a true value in relation to other intervals. 
For a non-trivial interval $I_j=(L_j,U_j)$,
we say that the value of interval $I_i$ \emph{passes over} $L_j$ if one of $w_i,\w_i$ is
$\le L_j$ and the other is $>L_j$. Similarly,
the value of $I_i$ \emph{passes over} $U_j$ if one of $w_i,\w_i$ is
$< U_j$ and the other is $\ge U_j$. Intuitively, the value of $I_i$ passes over one endpoint of~$I_j$
if it enters or leaves $I_j$, and it passes over both endpoints of $I_j$ if it jumps over $I_j$
when going from predicted to true values.
For a trivial interval $I_j=\{w_j\}$, we say that the value of $I_i$
\emph{jumps over}~$I_j$ if one of $w_i,\w_i$ is strictly smaller than $w_j$ and the other is strictly
larger than $w_j$.
{To avoid counting values passing over endpoints of irrelevant intervals,} 
let $A_i$ be the set of intervals that potentially interact with $I_i$.
For the minimum and the sorting problem, we let $A_i$ be the union of all sets that contain $I_i$. When
considering the MST problem, we consider the maximal biconnected component containing $I_i$, i.e., $A_i$
is the union of all intervals on cycles containing~$I_i$.
Now define $h_i=h_i(A_i)$ to be the number of non-trivial intervals $I_j\in A_i$ such that the value of~$I_i$ passes
over $L_j$ plus the number of non-trivial intervals $I_j\in A_i$ such that the value of~$I_i$ passes
over $U_j$, plus the number of trivial intervals $I_j=\{w_j\}$ in $A_i$ such that
the value of $I_i$ jumps over~$I_j$.
The hop distance of a given instance is then $k_h = \sum_{i=1}^n h_i$; 
see also Figure~\ref{fig_hop_ex}.
Note that $k_{\#} = 0$ implies $k_h = 0$, so Theorem~\ref{theo_minimum_combined_lb} implies that no algorithm can simultaneously have competitive ratio better than $1 + \frac{1}{\beta}$ if $k_h = 0$ and $\beta$ for arbitrary~$k_h$.

\begin{figure}[t]
	\centering
\subfigure[]{\label{fig_hop_ex}
	\begin{tikzpicture}[line width = 0.3mm, scale = 0.9, transform shape]	
	\intervalpr{$I_1$}{0}{4}{1.5}{1}{2.75}	
	\intervalpr{$I_2$}{1.5}{6}{2.25}{4.5}{2}	
	\intervalpr{$I_3$}{2.5}{6}{3}{4.5}{5.5}							
	\intervalpr{$I_4$}{3}{6}{3.75}{3.25}{3.75}		
	
	\begin{scope}[on background layer]
	\draw[line width = 0.2mm,,lightgray] (0,1) -- (0,4);
	\draw[line width = 0.2mm,,lightgray] (1.5,1) -- (1.5,4);
	\draw[line width = 0.2mm,,lightgray] (2.5,1) -- (2.5,4);
	\draw[line width = 0.2mm,,lightgray] (3,1) -- (3,4);
	\draw[line width = 0.2mm,,lightgray] (6,1) -- (6,4);
	\draw[line width = 0.2mm,,lightgray] (4,1) -- (4,4);
	\end{scope}
	
	\draw[decoration= {brace, amplitude = 5 pt, aspect = 0.5}, decorate] (2.75,1.3) -- (1,1.3);
	\node[] (l1) at (1.85,1) {\scalebox{0.75}{$h_1=2$}};
	
	\draw[decoration= {brace, amplitude = 5 pt, aspect = 0.5}, decorate] (4.5,2.1) -- (2,2.1);
	\node[] (l1) at (3.2,1.8) {\scalebox{0.75}{$h_2=3$}};
	
	\draw[decoration= {brace, amplitude = 2.5 pt, aspect = 0.5}, decorate] (5.5,2.85) -- (4.5,2.85);
	\node[] (l1) at (4.95,2.6) {\scalebox{0.75}{$h_3=0$}};
	
	\draw[decoration= {brace, amplitude = 1 pt, aspect = 0.5}, decorate] (3.75,3.6) -- (3.25,3.6);
	\node[] (l1) at (3.5,3.4) {\scalebox{0.75}{$h_4=0$}};
	\end{tikzpicture}
}
\hspace*{2cm}
\subfigure[]{\label{fig_man_ex}
	\begin{tikzpicture}[line width = 0.3mm, scale = 0.9, transform shape]
		\intervalpr{$I_1$}{0}{4}{1.5}{1}{1}		
		\intervalpr{$I_2$}{1.5}{6}{2.25}{3.25}{5.5}		
		\intervalpr{$I_3$}{2.5}{6}{3}{3.25}{5.5}								
		\intervalpr{$I_4$}{3}{6}{3.75}{3.25}{5.5}				
		
		\node[] (l1) at (1.85,0.95) {};	
	\end{tikzpicture}
}
\caption{Examples for the minimum problem with a single set $S=\{I_1,I_2,I_3,I_4\}$. Circles illustrate true values and crosses illustrate the predicted values. \subref{fig_hop_ex} Predictions and true values with a total hop distance of $k_h = 5$.
\subref{fig_man_ex} Predictions and true values with a mandatory query distance of $k_M = 1$.
}
\label{fig_error_ex}
\end{figure}
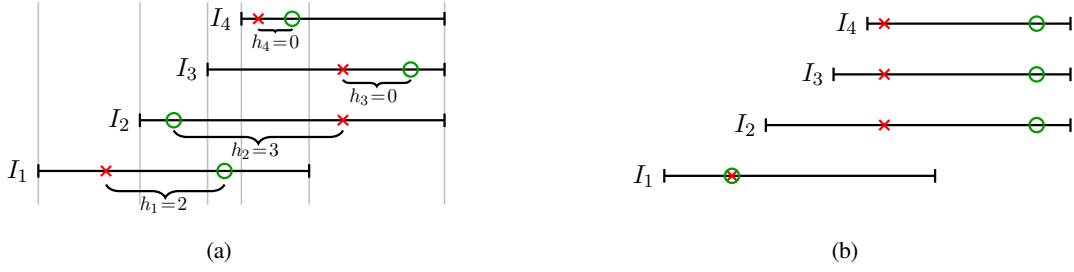

\emph{Mandatory query distance.} While the hop distance takes structural information regarding the interval structure into account, it does not distinguish whether a ``hop'' affects a feasible solution. We introduce a third and strongest measure for the prediction accuracy based on the following definition. 

\begin{definition}[mandatory]
	Given a problem instance with uncertainty intervals, an interval
	is {\em mandatory} if it is in each feasible query set of the verification problem. An interval is {\em prediction mandatory} if it is in each feasible query set assuming that the predictions $\w$ are accurate.
\end{definition}

Let $\mathcal{I}_P$ be the set of prediction mandatory elements, 
and let $\mathcal{I}_R$ be the set of really mandatory elements. 
The 
\emph{mandatory query distance} is the size of the symmetric difference of $\mathcal{I}_P$ and $\mathcal{I}_R$, i.e., $k_M = |\mathcal{I}_P \sym \mathcal{I}_R| = |(\mathcal{I}_P \cup \mathcal{I}_R) \setminus (\mathcal{I}_P \cap \mathcal{I}_R)| = |(\mathcal{I}_P \setminus \mathcal{I}_R) \cup (\mathcal{I}_R \setminus \mathcal{I}_P)|$.
Figure~\ref{fig_man_ex} shows an example with $k_M=1$.
Considering the true values in the example, both $\{I_1\}$ and $\{I_2,I_3,I_4\}$ are feasible solutions. Thus, no element is part of every feasible solution and $\mathcal{I}_R = \emptyset$. 
Assuming correct predicted values, $I_1$ has to be queried even if all other intervals have already been queried and, therefore, $\mathcal{I}_P = \{I_1\}$.
It follows $k_M = |\mathcal{I}_P \sym \mathcal{I}_R| = 1$.

Obviously, $k_M$ is a problem-specific error measure as, in a given set of uncertainty intervals, different intervals may be mandatory for different problems.
We can relate $k_M$ to $k_h$ in the following theorem.
 
\begin{restatable}{theorem}{HopDistanceMandatoryDistance}
	\label{Theo_hop_distance_mandatory_distance}
	For any instance of the minimum, sorting and MST problems under uncertainty, the hop distance is at least as large as the mandatory query distance, $k_M \leq k_h$.
\end{restatable}

We provide a lower bound on the competitive ratio that is stronger than Theorem~\ref{theo_minimum_combined_lb}, and later we give matching algorithms for the minimum and sorting problems.
The choice of $\gamma \geq 2$ is due to the lower bound of $2$ in the robustness for all problems we consider.

\begin{restatable}{theorem}{ThmLBMandQueryDist}
\label{theo_lb_sym_diff}
  Let $\gamma \geq 2$ be a fixed rational value. 
  If a deterministic algorithm for the minimum, sorting or MST problem is
  $\gamma$-robust,
  then it cannot have competitive ratio better than $1 + \frac{1}{\gamma-1}$ for $k_M = 0$.
  Furthermore, if an algorithm has competitive ratio $1 + \frac{1}{\gamma-1}$ for $k_M = 0$, then it cannot be better than
	$\gamma$-robust.
\end{restatable}
 
We conclude the definition and discussion of measures for the prediction accuracy with a simple lower bound on the competitive ratio regardless of any desired robustness.

\begin{restatable}{theorem}{ThmLBAllErrorMeasures}
\label{thm_lb_error_measure}
	Any deterministic algorithm for minimum, sorting or MST under uncertainty has a competitive ratio $\rho\geq \min\{1+\frac{k}{\opt},2\}$, for any error measure $k \in \{k_{\#}, k_M, k_h\}$, even for {disjoint} sets.
\end{restatable}

\paragraph{Learnability of predictions} 
We argue that our assumption of having access to machine-learned predictions is realistic. We do so by proving that the predictions are PAC-learnable w.r.t.~$k_h$ and $k_M$. See Appendix~\ref{sec:learnability} for full proofs.
To show PAC-learnability, we assume that the realization $w$ of true values for $\mathcal{I}$ is i.i.d.~drawn from an unknown distribution $\ud$, and
we can i.i.d.~sample realizations from $\ud$ to obtain a training set.
Let $\hs$ denote the set of all possible predictions $\pred{w}$, that is, vectors $\pred{w}$ with $\pred{w}_i \in I_i$ for each $I_i \in \mathcal{I}$.
Let $k_h(w,\pred{w})$ denote the hop distance of the predictions $\pred{w}$ for the realization with the true values $w$.
Since $w$ is drawn from $\ud$, the value $k_h(w,\pred{w})$ is a random variable. 
Analogously, we consider $k_M(w,\pred{w})$ with regard to the mandatory query distance.
Our goal is to learn predictions $\pred{w}$ that (approximately) minimize the expected error $\EX_{w \sim \ud}[k_h(w,\pred{w})]$ respectively $\EX_{w \sim \ud}[k_M(w,\pred{w})]$.

For both error measures, we employ the \emph{empirical risk minimization (ERM)} algorithm. ERM first i.i.d.~samples a training set~$S=\{w^1,\ldots,w^m\}$ of~$m$ true value vectors from~$\ud$.
Then, it returns the~$\pred{w} \in \hs$ that minimizes the \emph{empirical error}~$k_S(\pred{w}) = \frac{1}{m} \sum_{j=1}^{m} k(w^j,\pred{w})$ with $k \in \{k_h,k_m\}$.
We show a polynomial sample complexity $m$ by carefully reducing $\hs$. The challenging part
is to execute ERM in polynomial time.

\begin{restatable}{theorem}{TheoLearningHop}
	\label{theo_learnability_hop}
	For any~$\eps, \delta \in (0,1)$ and $k \in \{k_h,k_M\}$, there exists a learning algorithm that, using a training set of size~$m$, returns predictions $\pred{w} \in \hs$, such that
	$\EX_{w \sim \ud}[k(w,\pred{w})] \le \EX_{w \sim \ud}[k(w,\pred{w}^*)] + \eps$ holds with probability at least~$(1-\delta)$, where~$\pred{w}^* = \arg\min_{\pred{w}' \in \hs} \EX_{w \sim \ud}[k(w,\pred{w}')]$.
	The sample complexity is $m \in \mathcal{O}\left(\frac{(\log(n) - \log(\delta/n))\cdot (2n)^2}{(\eps/n)^2}\right)$ for $k=k_h$, and $m \in \mathcal{O}\left(\frac{(n \cdot \log(n) - \log(\delta))\cdot n^2}{\eps^2}\right)$ otherwise.
	For $k=k_h$, the running time is polynomial in $m$ and $n$.
	Otherwise, the running time is exponential in $n$.
\end{restatable}

Since learning w.r.t.~$\hs$ and~$k_M$ by using Theorem~\ref{theo_learnability_hop} requires exponential running time, we present an alternative approach.
Instead of showing the learnability of the predicted values, we prove that the set $\mathcal{I}_P$ that leads to the smallest expected error can be learned.
Note that access to only $\mathcal{I}_P$ is sufficient to execute the algorithm for the minimum problem that achieves a competitive ratio depending on $k_M$.
To be more specific, let $\mathcal{P}$ be the power set of $\mathcal{I}$, let $\mathcal{I}_w$ denote the set of mandatory elements for the realization with true values~$w$, and let $k_M(\mathcal{I}_w, P)$ with $P \in \mathcal{P}$ denote the mandatory query distance under the assumption that $\mathcal{I}_P = P$ and $\mathcal{I}_R = \mathcal{I}_w$.
Since $w$ is drawn from $\ud$, the value $\EX_{w \sim \ud}[k_M(\mathcal{I}_w,P)]$ is a random variable.
The mandatory query distance is a problem-specific metric as the characterization of (prediction) mandatory intervals may differ for different problems. We show efficient learnability w.r.t.~$\mathcal{P} $ and~$k_{M}$ for all problems for which $k_M(w,\w)$ can be computed efficiently. This is the case for all problems considered in this paper.
As another argument for learning w.r.t.~$\mathcal{P}$ and $k_M$, we show in Appendix~\ref{sec:learnability} that, for any distribution $D$, the best prediction in $\mathcal{P}$ never has a larger expected error than the best prediction in $\hs$. The inverse is not true.

\begin{restatable}{theorem}{TheoLearningMan}
	\label{theo_learnability_mandatory}
	For all problems for which $k_M$ can be determined in polynomial time and for
	any~$\eps, \delta \in (0,1)$, there exists a learning algorithm that, using a training set of size~$m  \in \mathcal{O}\left(\frac{(n - \log(\delta))\cdot n^2}{\eps^2}\right),$ returns a predicted set of mandatory intervals $P \in \mathcal{P}$ in time polynomial in~$n$ and~$m$, such that
	$\EX_{w \sim \ud}[k_M(\mathcal{I}_w,P)] \le \EX_{w \sim \ud}[k_M(\mathcal{I}_w,P^*)] + \eps$ holds with probability at least~$(1-\delta)$, where~$P^* = \arg\min_{P' \in \mathcal{P}} \EX_{w \sim \ud}[k_M(\mathcal{I}_w,P')]$.
\end{restatable}

\section{Overview of Techniques}
\label{sec:methods}
Given predicted values for the uncertainty intervals, it is tempting to simply run an optimal offline algorithm~(verification algorithm) under the assumption that the predictions are correct, and to then perform all the queries computed by that verification algorithm.
This is obviously optimal with respect to consistency, but might give arbitrarily bad solutions in the case when the predictions are faulty. 
Instead of blindly following the offline algorithm, we need a strategy to be robust against prediction errors. 
Therefore, we carefully combine structural properties of the problem with {the additional knowledge of untrusted} predictions. 

The crucial structure and unifying concept in all problems under consideration are {\em witness sets}. Witness sets are the key to any comparison with an optimal solution. A ``classical'' witness set is a set of elements for which we can guarantee that any feasible solution must query at least {\em one} of these elements.
Depending on the particular problem, witness sets can be structurally very different (e.g., simply pairs of overlapping intervals for one problem, or sets resulting from a complex consideration of cycles in a certain order for another problem). Nevertheless, in the classical setting without access to predictions, all our problems admit $2$-competitive online algorithms that rely essentially on identifying and querying disjoint witness sets of size two.
We refer to witness sets of size two also as {\em witness pairs}.
While completely relying on querying witness pairs ensures $2$-robustness, it does not lead to any improvements in terms of consistency.
In order to obtain an improved consistency while maintaining optimal or near-optimal robustness, we need to carefully balance the usage of an offline algorithm with the usage of witness sets.

\paragraph{The general framework}
On a high level,
our general algorithmic framework
follows the structure of the offline algorithm: 
In a first stage, it queries elements that are mandatory under the assumption that the predictions are correct.
In a second stage, when no elements are prediction mandatory any more, the algorithm has to decide for each witness pair which of the two elements to query.
If the predictions are correct, the second phase reduces to a vertex cover type problem. 
The following theorem summarizes some of our main results. For minimum and sorting, we show the result by proving the more general Theorem~\ref{thm_minimum:hop} in Section~\ref{sec:minimum} and Appendix~\ref{app:minimum} 
and giving a reduction from the sorting problem to the minimum problem in Section~\ref{sec:sorting}.	For the MST problem, we show the result directly
in Section~\ref{sec:mst} and Appendix~\ref{appx:mst}. 

\begin{theorem}
\label{thm:main1.5-2}
 {For each of the problems, minimum, sorting and MST under uncertainty, there is an algorithm that is $1.5$-consistent and $2$-robust.}
\end{theorem}

During the first phase, just querying prediction mandatory elements can be arbitrarily bad in terms of robustness (cf.\ Theorem~\ref{theo_minimum_combined_lb}).
Thus, given a prediction mandatory element that we wish to query, we query further elements in such a way that they form witness sets. Note that it is not sufficient to augment the set to a ``classical'' size-$2$ witness set (which already might be non-trivial to do for some problems), as this would not yield a performance guarantee better than~$2$ even if the predictions are correct.
Instead, we identify a set of elements for which we can guarantee that at least $2$ out of $3$ of them must be queried by any feasible solution. 
Since we cannot always find such elements based on structural properties alone (otherwise, there would be a $1.5$-competitive algorithm for the problem without predictions), we identify such sets under the assumption that the predictions are correct.
Even after identifying such elements, the algorithm needs to query them in a careful order: If the predictions are wrong, we lose the guarantee on the elements and querying all of them might violate the $2$-robustness.
The first phase of our framework repeatedly identifies such elements and queries them in a careful order while adjusting for potential errors, until no unqueried prediction mandatory elements remain.
Identifying the three elements with the guarantee mentioned above is a major contribution which might be of independent interest regardless of predictions, especially in the context of the MST problem. 
While existing algorithms for MST under uncertainty~\cite{erlebach08steiner_uncertainty,megow17mst} essentially follow the algorithms of Kruskal or Prim and only identify witness sets in the cycle or cut that is currently under consideration, we derive criteria to identify additional witness sets outside the current cycle/cut.
Note that the identification of witness sets and the characterization of (prediction) mandatory intervals are problem specific.
The first phase of our framework differs for the different problems only in the identification of witness sets and (prediction) mandatory elements, and the order in which we query them.

In the second phase, there are no more prediction mandatory elements.
Therefore, the algorithm cannot identify any more ``safe'' queries and, for each witness pair, has to decide which element to query. 
This phase boils down to finding a minimum vertex cover in an auxiliary graph representing
the structure of the witness sets.
In particular, the second phase for the minimum problem and the sorting problem is relatively straightforward and consists of finding and querying a minimum vertex cover.
If the predictions are correct, querying the vertex cover solves the remaining problem with an optimal number of queries. 
Otherwise, additional queries might be necessary, but we can show that this does not violate $2$-robustness.

In the MST problem, wrong predictions can change the vertex cover instance dynamically, and therefore it must be solved in a very careful and adaptive way, requiring substantial additional work. 
Nevertheless, the general framework is the same for all problems.

We remark that the computation of the minimum vertex cover in the second phase is the \emph{only} operation of our algorithms that potentially requires exponential running time. 
For the problems of finding an MST and sorting a single set,
we have to compute a minimum vertex cover in  bipartite and interval graphs, respectively. Both problems are well-known to be polynomial-time solvable.

\paragraph{Parameterized error-sensitive guarantees} We show how to refine our general framework to achieve error-sensitive and parameterized performance guarantees.
 {We parameterize the first phase as follows: we repeatedly query $\gamma-2$ prediction mandatory elements in addition to the set of three elements of the general framework. The second phase remains unchanged. For both, the minimum problem and sorting, we can bound the error-sensitivity by charging each queried prediction mandatory element that turns out to be not mandatory and each additional query in the second phase to a distinct prediction error in the $k_h$-metric.}
\begin{restatable}{theorem}{thmminimumhop}
	\label{thm_minimum:hop}
	There is an algorithm for the minimum problem under uncertainty that, given an integer parameter $\gamma \geq 2$, achieves a competitive ratio of $\min\{ (1 + \frac{1}{\gamma})(1 + \frac{k_h}{\opt}), \gamma\}$.
	Furthermore, if $\gamma = 2$, then the competitive ratio is $\min\{1.5 + k_h/\opt, 2\}$.
\end{restatable}

The more adaptive and dynamic nature of the MST problem complicates the handling of Phase $2$.
We therefore employ again a more adaptive strategy to cover for potential prediction errors in Phase $2$.
By using a charging/counting scheme that builds on K\"onig-Egerv\'ary's famous theorem on the duality of minimum vertex covers and maximum matchings in bipartite graphs, we achieve the {following result.}
\begin{restatable}{theorem}{MstTheoremTwo}
	\label{theorem_mst2}
	There is an algorithm for the MST problem under uncertainty with competitive ratio $\min\{ 1+ \frac{1}{\gamma} +  \left(5 + \frac{1}{\gamma}\right) \cdot \frac{k_h}{\opt},\, \max\{3,\gamma + \frac{1}{\opt}\}\}$, for any $\gamma\in\ZZ$ with $\gamma\geq 2$, {where $k_h$ is the hop distance of an instance}.
\end{restatable}

Further, we {refine}  our framework to obtain guarantees dependent on {the mandatory query distance} $k_M$ for the minimum and sorting problems. We observe that we can hope to obtain only a slightly worse consistency for the same robustness in comparison to the guarantees dependent on $k_h$ (cf.\ Theorem~\ref{theo_lb_sym_diff}). 
Based on this observation, we adjust the framework to be \enquote{less careful} in the first phase and repeatedly query $\gamma-1$ prediction mandatory elements and one additional element that forms a witness set with one of the queried prediction mandatory elements.
For sorting and minimum, we show that this adjustment with an unchanged second phase leads to the competitive ratio stated in the following theorem.
\begin{restatable}{theorem}{ThmMinAlpha}
	\label{thm:min-alpha}
	There is an algorithm for the minimum problem under uncertainty that, given an integer parameter $\gamma \geq 2$, achieves a competitive ratio of $\min\{ (1+\frac{1}{\gamma-1}) \cdot (1 + \frac{k_M}{\opt}), \gamma\}$.
\end{restatable}
The error-sensitivity can again be shown by charging each queried prediction mandatory element that turns out to not be mandatory and each additional query in the second phase to a distinct error in the $k_M$-metric. 
In contrast to the $k_h$-dependent error-sensitivity, only elements that are prediction mandatory based on the initially given information and that turn out not to be mandatory contribute to $k_M$.
Thus, we fix the set of prediction mandatory elements at the beginning of the algorithm and show that querying only those elements in the first phase is sufficient.
We remark that the integrality requirement in Theorem~\ref{thm:min-alpha} can be removed by randomization at the cost of a slightly worse guarantee.

 {Finally, we observe that for sorting a single set, a substantially better algorithm is possible: a $1$-consistent  $2$-robust algorithm with a competitive ratio that linearly degrades depending on the prediction error. Note that for a }
$1$-consistency, an algorithm \emph{must} follow the offline algorithm and cannot afford additional queries unless the predictions are wrong. 
To simultaneously guarantee $2$-robustness and error dependency, the algorithm has to perform queries in a very carefully selected order, both for the prediction mandatory elements in Phase 1 and the vertex cover in Phase 2.
By employing such a strategy, we achieve the following theorem.
\begin{restatable}{theorem}{sortingsingleset}
	\label{thm:sorting:singleset}
	For sorting under uncertainty for a single set, there is a polynomial-time algorithm with competitive ratio $\min\{1+k/\opt, 2\}$, for any error measure $k \in \{k_{\#}, k_M, k_h\}$.
\end{restatable}

\section{The minimum problem}
\label{sec:minimum}
We show how to implement our general framework for the minimum problem under uncertainty achieving best possible competitive ratios with respect to
the accuracy measures $k_{\#}$ and $k_M$. 

\subsection{Preliminaries}                 

We start by giving a characterization of (prediction) mandatory queries and witness sets. Secondly, we present a verification algorithm, i.e., an optimal offline algorithm.

\begin{restatable}{lem}{LemmaMandatoryMin}
\label{lema_mandatory_min}
An interval~$I_i$ is {\em mandatory} for the minimum problem if and only if
(a) $I_i$ is a true minimum of a set~$S$ and contains $w_j$ of another interval $I_j\in S\setminus \{I_i\}$ (in particular, 
if $I_j\subseteq I_i$), or
(b) $I_i$ is not a true minimum of a set~$S$ with $I_i \in S$ but contains the value of the true minimum of $S$. 
{\em Prediction mandatory} intervals are characterized equivalently, replacing true values by predicted values.
\end{restatable}

\begin{proof}
If~$I_i$ is a true minimum of~$S$ and contains~$w_j$ of another interval $I_j\in S$, then~$S$ cannot be solved even if we query all intervals in $S \setminus \{I_i\}$. 
If~$I_i$ is not a true minimum of a set $S$ with $I_i\in S$ and contains the true minimum value~$w^*$ of $S$, then~$S$ cannot be solved even if we query all intervals in $S \setminus \{I_i\}$, as we cannot prove that $w^* \leq w_i$.

If~$I_i$ is the true minimum of a set $S$, but $w_j \notin I_i$ for every $I_j \in S \setminus \{I_i\}$, then $S \setminus \{I_i\}$ is a feasible solution for~$S$.
If~$I_i$ is not a true minimum of a set~$S$ and does not contain the true minimum value of~$S$, then again $S \setminus \{I_i\}$ is a feasible solution for~$S$. If every set $S$ that contains $I_i$ falls into one of these two cases, then querying all intervals except $I_i$ is a feasible query set for the whole instance.
\end{proof}

	Lemma~\ref{lema_mandatory_min} does not only enable us to identify mandatory intervals given full knowledge of the true values, but also implies criteria to identify \emph{known mandatory} intervals, i.e., intervals that are known to be mandatory given only the intervals, and true values revealed by previous queries.
	We call an interval {\em leftmost} in a set~$S$ if it is an interval with minimum lower limit in~$S$.
	The following corollary follows from 
	Lemma~\ref{lema_mandatory_min} and gives a characterization of known mandatory intervals.

\begin{restatable}{coro}{CorMinLeftMandatory} \label{cor_min_left_mandatory}
If the leftmost interval~$I_l$ in a set~$S$ contains the true value of another interval in~$S$, then~$I_l$ is mandatory.
In particular, if $I_l$ is leftmost in $S$ and $I_j \subseteq I_l$ for some $I_j \in S \setminus \{I_l\}$, then $I_l$ is mandatory.
\end{restatable}
Every algorithm can query known mandatory intervals without worsening its competitive ratio.
	In addition to exploiting prediction mandatory elements, our algorithms rely on identifying witness sets of size two by using the following lemma.

\begin{lemma}[\cite{kahan91queries}]
		A set $\{I_i, I_j\} \subseteq S$ with $I_i \cap I_j \neq \emptyset$, and $I_i$ or~$I_j$ leftmost in~$S$, is a witness set.
\end{lemma}

Since following the predictions can lead to an arbitrarily bad robustness, our algorithms execute additional queries to verify that a prediction mandatory interval is indeed mandatory.
To do so, we observe that some intervals are rendered prediction mandatory by a single predicted value.
A predicted value~$\pred{w}_j$ {\em enforces} another interval~$I_i$ if $\pred{w}_j \in I_i$ and $I_i, I_j \in S$, where~$S$ is a set such that either~$I_i$ is leftmost in~$S$, or~$I_j$ is leftmost in~$S$ and~$I_i$ is leftmost in $S \setminus \{I_j\}$.
Corollary~\ref{cor_min_left_mandatory} implies that $I_i$ is mandatory if the predicted value of $I_j$ is correct.
It is then easy to see that, if $\pred{w}_j$ enforces~$I_i$, then $\{I_i, I_j\}$ is a witness set and $I_i$ is prediction mandatory; moreover, if $w_j \in I_i$ then~$I_i$ is mandatory.

We use Lemma~\ref{lema_mandatory_min} to design a verification algorithm. It follows the same two-phase structure as our general framework:
First, we query all mandatory intervals.
After that, each unsolved set~$S$ has the following configuration: The leftmost interval~$I_i$ has true value outside all other intervals in~$S$, and each other interval in~$S$ has true value outside~$I_i$.
Thus we can either query~$I_i$ or all other intervals that intersect~$I_i$ in~$S$ to solve it.
The optimum solution is to query a minimum vertex cover in the graph with a vertex for each interval and, for each unsolved set, an edge between the leftmost interval and the intervals that intersect it.
This algorithm may require exponential time, but this is not surprising as we can
show that the verification version of the minimum problem is NP-hard (a proof is given in Appendix~\ref{sec:nphard}).

\begin{theorem}
	The verification problem for the minimum problem under uncertainty is NP-hard. The above algorithm solves it with a minimum number of queries. 
\end{theorem}

\subsection{Algorithm regarding hop distance} 

We prove Theorem~\ref{thm_minimum:hop} by presenting Algorithm~\ref{ALG_min_beta} with a performance guarantee depending on the hop distance $k_h$, i.e., the competitive ratio $\min\{ (1 + \frac{1}{\gamma})(1 + \frac{k_h}{\opt}), \gamma\}$  for each integer parameter $\gamma \ge 2$.
This guarantee is best possible for $k_h = 0$ and for large~$k_h$ (Theorem~\ref{theo_minimum_combined_lb}).
Note that, for $\gamma = 2$, the algorithm also satisfies Theorem~\ref{thm:main1.5-2} for the minimum problem under uncertainty, since $k_{\#} = 0$ implies $k_h = 0$.
A complete proof of Theorem~\ref{thm_minimum:hop} is
in Appendix~\ref{app:minbeta}; here we only describe the algorithm and high level~arguments.

The algorithm follows the general framework of Section~\ref{sec:methods}.
	After preprocessing the instance by querying known mandatory intervals in Line~\ref{line_min_beta_kh_mandatory_first}, the algorithm implements the first framework phase in
	Lines~\ref{line_min_beta_kh_calc_predict_first} to~\ref{line_min_beta_kh_cond_loop}, while Lines~\ref{line_min_beta_kh_vc} to~\ref{line_min_beta_kh_after_vc} corresponds to the second phase.

\begin{algorithm}[tb]
  \KwIn{Intervals $I_1, \ldots, I_n$, prediction $\pred{w}_i$ for each $I_i$, and family of sets $\mathcal{S}$}
  \label{line_min_beta_start}
  \Repeat{no query was performed in this iteration \label{line_min_beta_kh_cond_loop}}{
    \lWhile{there is a known mandatory interval $I_i$}{query $I_i$ \label{line_min_beta_kh_mandatory_first}}
    $Q \leftarrow \emptyset$; \quad $P \leftarrow$ set of prediction mandatory intervals for current instance\; \label{line_min_beta_kh_calc_predict_first}
    \While{$P \neq \emptyset$ \KwAnd $|Q| < \gamma - 2$}{
      pick and query some $I_j \in P$; \quad $Q \leftarrow Q \cup \{I_j\}$\; \label{line_min_beta_kh_predict}
      \lWhile{there is a known mandatory interval $I_i$}{query $I_i$ \label{line_min_beta_kh_mandatory_loop}}
      $P \leftarrow$ set of prediction mandatory intervals for current instance\; \label{line_min_beta_kh_calc_predict_second}
    }
    \uIf{$\exists \mbox{ distinct } I_i,I_j,I_l$ such that $\pred{w}_j$ enforces $I_i$ and $\{I_j,I_l\}$ is a witness set \label{lin_min_beta_cond_trio}} {
      query $I_j, I_l$\; \label{line_min_beta_wit_trio}
      \lIf{$w_j \in I_i$}{query $I_i$ \label{line_min_beta_mand_trio}}
    } \lElseIf{$\exists I_i,I_j$ such that $\pred{w}_j$ enforces $I_i$ \label{line_min_beta_pair}} {
      query $I_i$ \label{line_min_beta_query_first}
    }
  }
  \lWhile{there is a known mandatory interval $I_i$}{query $I_i$ \label{line_min_beta_kh_before_vc}}
  compute and query a minimum vertex cover~$Q'$ on the current dependency graph\; \label{line_min_beta_kh_vc}
  \lWhile{there is a known mandatory interval $I_i$}{query $I_i$ \label{line_min_beta_kh_after_vc}}
  \caption{Algorithm for the minimum problem under uncertainty w.r.t.\ hop distance $k_h$}
  \label{ALG_min_beta}
\end{algorithm}

The if-statement in Lines~\ref{lin_min_beta_cond_trio} to~\ref{line_min_beta_query_first}, corresponds to the identification of three elements such that, assuming correct predictions, at least two of them are mandatory (cf. paragraph on the general framework in Section~\ref{sec:methods}).
In these lines, the algorithm first tries to identify three distinct intervals $I_i, I_j, I_l$, such that $\pred{w}_j$ enforces $I_i$ and $\{I_j, I_l\}$ is a witness set.
If there is such trio, we query $\{I_j, I_l\}$, and only query~$I_i$ if the prediction that $\pred{w}_j \in I_i$ is correct.
If the prediction is correct, then we have a set of size 3 with at least 2 mandatory intervals; otherwise we only query a witness pair, so we can enforce robustness, and if one interval in this pair is not in $\OPT$ then we can charge this error to the hop distance, since the prediction $\pred{w}_j$ is incorrect.
If there is no such trio, then we try to find a witness pair $\{I_i, I_j\}$ such that $\pred{w}_j$ enforces $I_i$, but initially we only query~$I_i$.
If the prediction $\pred{w}_j$ is correct, then we are querying a mandatory interval.
Otherwise, we show that $I_j$ is either queried in Line~\ref{line_min_beta_kh_mandatory_first} in the next iteration of the loop (so it is a mandatory interval), or is never queried by the algorithm; either way, the fact that this is a witness pair is enough to guarantee robustness, and if $I_i$ is not in $\OPT$ then we can charge this error to the hop distance.
Summing up, for each iteration of the loop we can identify a witness set of size at most~$3$, such that at least a $\frac{2}{3}$ fraction of its elements are prediction mandatory, and those that are not mandatory can be charged to the hop distance. 
For $\gamma = 2$ this concludes the description of the first framework phase.

In case of $\gamma > 2$, the algorithm follows the framework for achieving parameterized guarantees by additionally querying $\gamma -2$  prediction mandatory intervals in Lines~\ref{line_min_beta_kh_calc_predict_first} to~\ref{line_min_beta_kh_calc_predict_second}.
Doing this, each iteration of the loop identifies a witness set of size at most~$\gamma$, excluding the mandatory elements that are potentially queried in Line~\ref{line_min_beta_kh_mandatory_loop} or~\ref{line_min_beta_mand_trio}.
This ensures the robustness.
By adding the prediction mandatory elements, the local consistency (of the queries in a single iteration) improves to $\frac{\gamma+1}{\gamma}$.
All prediction mandatory elements that are not mandatory can be charged to the hop distance.

We may have an iteration of the loop in that we do not query any intervals in Lines~\ref{line_min_beta_wit_trio} and~\ref{line_min_beta_query_first}, but we show that this occurs at most once: If we cannot satisfy the conditions of Lines~\ref{lin_min_beta_cond_trio} and~\ref{line_min_beta_pair}, then there are no more prediction mandatory intervals.
After that, the algorithm will proceed to the second phase of the framework, querying a minimum vertex cover and intervals that become known mandatory.
Here we combine the at most $\gamma-2$ intervals queried in the iteration described with the intervals queried in the second phase, and it is not hard to prove $\gamma$-robustness and that every interval that is not in $\OPT$ can be charged to the hop distance.

\subsection{Algorithm regarding mandatory-query distance}

Now we consider the mandatory-query distance $k_M$ as measure for the prediction accuracy. We prove Theorem~\ref{thm:min-alpha} by presenting Algorithm~\ref{ALG_min_alpha} with a competitive ratio of  $\min\{ (1 + \frac{1}{\gamma-1})(1 + \frac{k_M}{\opt}), \gamma\}$,  for each integer parameter $\gamma \ge 2$.
This upper bound is tight for $k_M = 0$ and large~$k_M$ (Theorem~\ref{theo_lb_sym_diff}). The theorem can be generalized for arbitrary real $\gamma \geq 2$ with a marginally increased competitive ratio; see Appendix~\ref{app:min_alpha}. 
The full proof of Theorem~\ref{thm:min-alpha} appears also in Appendix~\ref{app:min_alpha}; here we outline algorithm and crucial arguments.

Algorithm~\ref{ALG_min_alpha} implements the first framework phase in Lines~\ref{line_min_param_cond} to~\ref{line_min_param_small} and afterwards executes the second phase.
To start the first phase, the algorithm computes the set~$P$ of initial prediction mandatory intervals (Lemma~\ref{lema_mandatory_min}).
Then it tries to find an interval $p \in P$ that is part of a witness set $\{p, b\}$.
If~$|P|\ge\gamma-1$, we query a set $P' \subseteq P$ of size $\gamma-1$ that includes~$p$, plus~$b$ (we allow $b \in P'$).
This is clearly a witness set of size at most $\gamma$, at least a $\frac{\gamma-1}{\gamma}$ fraction of the intervals are in~$P$, and every interval in $P \setminus \OPT$ is in $\mathcal{I}_P \setminus \mathcal{I}_R$.
We then repeatedly query known mandatory intervals, remove the queried intervals from~$P$ and repeat the process without recomputing~$P$, until~$P$ is empty or no interval in~$P$ is part of a witness set.

We may have one last iteration of the loop where $|P| < \gamma -1$.
After that, the algorithm will proceed to the second phase of the framework, querying a minimum vertex cover and intervals that become known mandatory.
Here we combine the at most $\gamma-2$ intervals in $P$ with the intervals queried in the second phase, and it is not hard to prove $\gamma$-robustness, and that the number of those queries can be bounded by the number of intervals in $\OPT$ for the current instance plus the number of intervals that are in $\mathcal{I}_P \setminus \mathcal{I}_R$ or in $\mathcal{I}_R \setminus \mathcal{I}_P$.

Note that this algorithm only uses the initial set of prediction mandatory intervals, and otherwise ignores the predicted values.
Predicting the set of prediction mandatory intervals is sufficient to execute the algorithm.

\begin{algorithm}[tb]
  \KwIn{Intervals $I_1,\ldots,I_n$, predicted value $\overline{w}_i$ for each $I_i$, family of sets $\mathcal{S}$, and parameter $\gamma$}

  $P \leftarrow$ set of initial prediction mandatory intervals (characterized in Lemma~\ref{lema_mandatory_min})\;
  \While{$\exists p \in P$ and an unqueried interval $b$ where $\{p, b\}$ is a witness set \label{line_min_param_cond}}{
    \uIf{$|P| \geq \gamma-1$ \label{line_min_param_size}}{
      pick $P' \subseteq P$ with $p \in P'$ and $|P'| = \gamma-1$\;
      query $P' \cup \{b\}$, $P \leftarrow P \setminus (P' \cup \{b\})$\; \label{line_min_param_big}
     \lWhile{there is a known mandatory interval~$I_i$}{query $I_i$, $P \leftarrow P \setminus \{I_i\}$ \label{line_min_param_mandatory}}
    } \lElse{query $P$, $P \leftarrow \emptyset$ \label{line_min_param_small}}
  }
  \lWhile{there is a known mandatory interval~$I_i$}{query $I_i$ \label{line_min_param_before_vc}}
  query a minimum vertex cover~$Q$ for the current instance \label{line_min_param_vc}\;
  \lWhile{there is a known mandatory interval~$I_i$}{query $I_i$ \label{line_min_param_mand_after_vc}}
  \caption{Algorithm for the minimum problem under uncertainty w.r.t.\ error measure $k_M$}
  \label{ALG_min_alpha}
\end{algorithm}

\section{The minimum spanning tree problem}
\label{sec:mst}
In this section, we show how the framework of Section~\ref{sec:methods} can be implemented to achieve an algorithm for the MST problem with a hop-distance dependent performance guarantee. We prove  Theorem~\ref{theorem_mst2}, which states a  competitive ratio of $\min\{ 1+ \frac{1}{\gamma} +  (5 + \frac{1}{\gamma}) \cdot \frac{k_h}{\opt}, \max\{3,\gamma + \frac{1}{\opt}\}\}$, for each integer $\gamma \ge 2$. Further, we introduce a second algorithm that achieves a better robustness at the cost of not providing an error-sensitive guarantee. More precisely, it is $1.5$-consistent and $2$-robust, which proves the MST related part of Theorem~\ref{thm:main1.5-2}. Both algorithms
use the same first phase to obtain a prediction mandatory free instance  and differ only in the second phase. We present our  implementations of the two phases in Sections~\ref{subsec_mst_phase_1} and~\ref{subsec_mst_phase_2}. 
Full proofs are given in Appendix~\ref{appx:mst}.

Our algorithms build on structural insights and characterizations of mandatory queries given in~\cite{megow17mst}, which we summarize here.
Let the \emph{lower limit tree} $T_L \subseteq E$ be an MST for
{values} $w^L$ with $w^L_e = L_e + \epsilon$ for an infinitesimally small $\epsilon > 0$. 
Analogously, let the \emph{upper limit tree} $T_U$ be an MST for {values}  $w^U$ with $w^U_e = U_e - \epsilon$. It has been shown in~\cite{megow17mst} that any non-trivial edge in $T_L \setminus T_U$ is part of any feasible query set, {i.e., it is mandatory}. {Thus, we may repeatedly query edges in
$T_L \setminus T_U$ until $T_L = T_U$ and this will not worsen the robustness or consistency. By this preprocessing, we may assume~$T_L=T_U$.
Further, we can extend the preprocessing to achieve uniqueness for $T_L$ and $T_U$.
\begin{restatable}{lemma}{LemMSTPreprocessing}
	\label{mst_preprocessing}
	By querying only mandatory elements we can obtain an instance with $T_L = T_U$ such that $T_L$ and $T_U$ are the unique lower limit tree and upper limit tree, respectively. 
\end{restatable}

Consider $T_L$ and
$f_1,\ldots,f_l$ in $E \setminus T_L$ ordered by increasing lower limits. 
For each~$i \in \{1,\ldots,l\}$, let $C_i$ be the unique cycle in $T_L \cup \{f_i\}$. 
For each $e \in T_L$ let $X_e$ be the set of edges in the cut of $G$ defined by the two connected components of $T_L \setminus \{e\}$.
We say an instance is \emph{prediction mandatory free} if it contains no prediction mandatory elements.
Otherwise, we say that the instance is \emph{non-prediction mandatory free.}
The following lemma further characterizes prediction mandatory free instances. Figure~\ref{Ex_mst_phase1_cases}(a) illustrates it.

\begin{restatable}{lemma}{LemMSTPredFreeIff}
	\label{mst_pred_free_characterization}
	{
		An instance $G$ is prediction mandatory free if and only if 
		$\pred{w}_{f_i} \ge U_{e}$ and $\pred{w}_e \le L_{f_i}$ holds for each $e \in C_i \setminus \{f_i\}$ and each cycle $C_i$ with ${i} \in \{1,\ldots,l\}$.}
\end{restatable}

\subsection{Identifying witness sets}
\label{subsec_mst_witness sets}

Before we describe our algorithms in the following sections, we introduce some preliminaries.
As already stated, the first phase of our algorithm handles non-prediction mandatory free instances until they become prediction mandatory free.
In order to achieve this goal while maintaining the desired performance guarantees, the algorithms rely on identifying witness sets on non-prediction mandatory free cycles $C_i$.

Existing algorithms for MST under uncertainty~\cite{erlebach08steiner_uncertainty,megow17mst} essentially follow the algorithms of Kruskal or Prim, and only identify witness sets in the cycle or cut that is currently under consideration.
For a given instance this corresponds to the cycle $C_1$ or the cut $X_{l_1}$, where $l_1$ is the edge with the largest upper limit in~$T_L$.
Regarding Phase $1$ of our algorithm, it might hold for the first non-prediction mandatory free cycle $C_i$ that $C_i \not= C_1$.
Since additionally $C_i \cap X_{l_1} = \emptyset$ might hold, existing methods for identifying witness sets are not sufficient for our purpose.
We show the following new structural insights that are crucial for our algorithms.

\begin{restatable}{lemma}{LemMstWitnessFir}
	\label{lemma_mst_witness_set_1}		
	Consider cycle $C_i$ with $i \in \{1,\ldots,l\}$. 
	Let $l_i \in C_i \setminus\{f_i\}$ such that $I_{l_i} \cap I_{f_i} \not= \emptyset$ and $l_i$ has the largest upper limit in $C_i \setminus \{f_i\}$, then $\{f_i,l_i\}$ is a witness set. 
	Further, if $w_{f_i} \in I_{l_i}$, then $\{l_i\}$ is a witness set.
\end{restatable}
\begin{restatable}{lemma}{LemMstWitnessSec}
	\label{lemma_mst_witness_set_2}
	Let $l_i \in C_i \setminus \{f_i\}$ with $I_{l_i} \cap I_{f_i} \not= \emptyset$ such that $l_i \not\in C_j$ for all $j < i$, then $\{l_i,f_i\}$ is a witness set.
	Furthermore, if $w_{l_i} \in I_{f_i}$, then $\{f_i\}$ is a witness set.
\end{restatable}

\subsection{Handling non-prediction mandatory free instances}
\label{subsec_mst_phase_1}

\begin{algorithm}[t]
	\KwIn{Uncertainty graph $G=(V,E)$ and predictions $\pred{w}_e$ for each {$e \in E$}}
	{Sequentially query prediction mandatory elements while ensuring unique $T_L = T_U$ until either $\gamma - 2$ prediction mandatory elements are queried or the instance is prediction mandatory free\label{line_mst_one_fillup}}\;
	Let $T_L$ be the lower limit tree and $f_1,\ldots,f_l$ be the edges in $E \setminus T_L$ ordered by lower limit non-decreasingly\label{line_mst_one_order}\;
	\ForEach{{$C_i$ with $i=1$ to $l$\label{line_mst_one_foreach_f}\label{line_mst_ine_unique_cycle}}}{
		\lIf{$C_i$ is not prediction mandatory free\label{line_mst_one_if_pred_free}}{
			{Apply Lemma~\ref{mst_phase1_case_a},~\ref{mst_phase1_case_b} or~\ref{mst_phase1_case_c}\label{line_mst_one_lemma_queries} and restart\label{line_mst_one_restart}}
		}
	}	
	\caption{Phase $1$ of the algorithms for MST under uncertainty}
	\label{ALG_mst_part_1}
\end{algorithm}

Algorithm~\ref{ALG_mst_part_1} implements Phase $1$ of our algorithms. 
In each iteration the algorithm starts by querying elements that are prediction mandatory for the current instance.
The set of prediction mandatory elements can be computed using the verification algorithm~\cite{erlebach14mstverification}.
Our algorithm sequentially queries such elements until either $\gamma - 2$ prediction mandatory elements have been queried or no more exist. 
After each query, the algorithm ensures unique $T_L = T_U$ by using Lemma~\ref{mst_preprocessing}.
Note that the set of prediction mandatory elements with respect to the current instance can change when elements are queried, and therefore we query the elements sequentially.
We can prove that each of the {at most} $\gamma-2$ elements is either mandatory or contributes one to the hop distance $k_h$.

Then, the algorithm iterates through $i \in \{1,\ldots,l\}$ and stops if the current cycle $C_i$ is non-prediction mandatory free. 
If it finds such a cycle, it queries edges on the cycle and possibly future cycles and restarts. 
The algorithm terminates when all $C_i$ are prediction mandatory free 
that is {at the latest when} all edges in $E$ have been queried. 
When the algorithm finds a non-prediction mandatory free cycle $C_i$, it carefully selects edges to query such that the following statements hold: 

\begin{compactenum}
	\item The algorithm only queries witness sets of size {one or} two, and sets of size three such that at least two elements are part of any feasible query set.

	\item If the algorithm queries a witness set {$W=\{e_1,e_2\}$} of size two, then either {$W \subseteq Q$} for each feasible query set $Q$ or the hop distances of $e_1$ and $e_2$ satisfy $h_{e_1} + h_{e_2} \ge 1$.

	\item In an only exception, the algorithm queries single elements $e$ that form witness sets $\{e,f(e)\}$ with distinct elements $f(e)$, 
	and the algorithm \emph{guarantees} that $f(e)$ remains unqueried during the complete execution.
	Since $\OPT$ must query at least one element of $\{e,f(e)\}$ and we guarantee that the algorithm queries exactly one element, querying such elements does not hurt the robustness or consistency as long as each such queried $e$ can be matched with a distinct $f(e)$.
	
	Let $E$ be the set of such queried edges.
	For the sake of our analysis, we assume without loss of generality that $E \subseteq \OPT$ and treat each $e \in E$ as a witness set of size one.
	We can do this without loss of generality since if $e \not\in \OPT$ we know $f(e) \in \OPT$ and can charge $e$ against $f(e)$.
\end{compactenum}
Ignoring the last iteration where the instance becomes prediction mandatory free and
{possibly less than} $\gamma-2$ prediction mandatory elements are queried, the algorithm  queries in each iteration~$\gamma - 2$ prediction mandatory elements and a set {$W$} that satisfies 
the three statements. This implies the following lemma.

\begin{restatable}{lem}{MSTEndOfPhaseOne}
	\label{mst_end_of_phase_one}
	After executing Algorithm~\ref{ALG_mst_part_1} the instance is prediction mandatory free and, ignoring the last iteration {of Line~\ref{line_mst_one_fillup}}, $|\ALG| \le \min\{ (1 + \frac{1}{\gamma}) \cdot (|\ALG \cap \OPT| + k_h), \gamma \cdot |\ALG \cap \OPT|\}$ holds for the set of edges $\ALG$ queried by Algorithm~\ref{ALG_mst_part_1} and any optimal solution $\OPT$.
\end{restatable}

In the last iteration, a (possibly empty) set $P$ of at most $\gamma-2$ prediction mandatory elements is queried.
Lemma~\ref{mst_pred_free_characterization} implies that only the last execution of Line~\ref{line_mst_one_fillup} might query less than $\gamma-2$ prediction mandatory elements.
As {each $e \in P$ is} either mandatory or contributes one to the hop distance,
querying $P$ does not violate the consistency.
{For $\gamma > 2$, we ensure the $(\gamma+\frac{1}{\opt})$-robustness by charging $P$ against the queries of the $3$-robust Phase $2$ of the algorithm.
Charging the first $\gamma-3$ elements of $P$ against the $3$-robust Phase $2$ leads to $\gamma$-robustness while the final element of $P$ leads to the additive term $\frac{1}{\opt}$.

The following
lemmas give algorithmic actions with a guarantee that the three statements are fulfilled. Each of them considers a cycle $C_i$ such that all $C_j$ with $j < i$ are prediction mandatory free, $l_i$ is the edge with the highest upper limit in $C_i \setminus \{f_i\}$ and predictions are as indicated in Figure~\ref{Ex_mst_phase1_cases}~b)-d).
The proofs of the three lemmas rely on the structural insights of Subsection~\ref{subsec_mst_witness sets}.

\begin{restatable}{lem}{MstPhaseOneCaseA}
	\label{mst_phase1_case_a}
	If $\pred{w}_{f_i} \in I_{l_i}$ and $\pred{w}_{l_i} \in I_{{f_i}}$, then querying $\{f_i,l_i\}$ satisfies the three statements.
\end{restatable}
\begin{restatable}{lem}{MstPhaseOneCaseB}
	\label{mst_phase1_case_b}
	
	Assume $\w_{f_i} \in I_{l_i}$ but $\w_{l_i} \not\in I_{f_i}$.
	Let $l_i'$ be the edge with the highest upper limit in $C_i \setminus \{f_i,l_i\}$ and  $I_{l_i'} \cap I_{f_i} \not= \emptyset$. 
	If no $l_i'$ exists, then querying $l_i$, and querying $f_i$ only if  $w_{l_i} \in I_{f_i}$, satisfies the three statements. 
	If $l_i'$ exists, then querying $\{f_i,l_i\}$, and querying $l_i'$  only if $w_{f_i} \in I_{l_i}$ and $w_{l_i} \not\in I_{f_j}$ for each $j$ with $l_i \in C_j$, satisfies the three statements.
\end{restatable}
\begin{restatable}{lem}{MstPhaseOneCaseC}
	\label{mst_phase1_case_c}
	Assume $\w_{l_i'} \in I_{f_i}$ for some $l_i' \in C_i\setminus\{f_i\}$ but $\w_{f_i} \not\in I_{l_i}$. Let $f_j$ be the edge with the smallest lower limit in $X_{l_i'} \setminus \{l_i',f_i\}$ and $I_{f_j}\cap I_{l_i'}\not=\emptyset$.
	If $f_j$ does not exist, then querying $f_i$, and querying $l_i'$ only if $w_{f_i} \in I_{l_i'}$, satisfies the three statements. 
	If $f_j$ exists, querying $\{f_i,l_i'\}$, and also querying $f_j$ only if $w_{l_i'}\in I_{f_i}$ and $w_{f_i} \not\in I_{e}$ for each $e \in C_i$, satisfies the three statements.
	
\end{restatable}

 \begin{figure}[th]
		\begin{minipage}{0.24\textwidth}
		\centering
        \begin{tikzpicture}[line width = 0.3mm, scale = 0.8, transform shape]
        \intervalp{$I_{f_i}$}{-3}{-0.5}{0}{-0.7}
        \intervalp{$I_{l_i}$}{-3.5}{-1}{0.5}{-3.3}
        \intervalp{}{-4}{-1.5}{1}{-3.2}
		\path (-2.75, 1.6) -- (-2.75, 1.6) node[font=\LARGE, midway, sloped]{$\dots$};
		\node[] at (-4.75,1.8){$(a)$};
		\end{tikzpicture}
	\end{minipage}
	\begin{minipage}{0.24\textwidth}
		\centering
		\begin{tikzpicture}[line width = 0.3mm, scale = 0.8, transform shape]
		\intervalp{$I_{f_i}$}{-3}{-0.5}{0}{-2.5}
		\intervalp{$I_{l_i}$}{-3.5}{-1}{0.5}{-2.8}
		\interval{}{-4}{-1.5}{1}
		\path (-2.75, 1.6) -- (-2.75, 1.6) node[font=\LARGE, midway, sloped]{$\dots$};
		\node[] at (-4.75,1.8){$(b)$};
		\end{tikzpicture}
	\end{minipage}
	\begin{minipage}{0.24\textwidth}
		\centering
		\begin{tikzpicture}[line width = 0.3mm, scale = 0.8, transform shape]
		\intervalp{$I_{f_i}$}{-3}{-0.5}{0}{-2.5}
		\intervalp{$I_{l_i}$}{-3.5}{-1}{0.5}{-3.3}
		\interval{}{-4}{-1.5}{1}
		\path (-2.75, 1.6) -- (-2.75, 1.6) node[font=\LARGE, midway, sloped]{$\dots$};
		\node[] at (-4.75,1.8){$(c)$};
		\end{tikzpicture}
	\end{minipage}
	\begin{minipage}{0.24\textwidth}
		\centering
		\begin{tikzpicture}[line width = 0.3mm, scale = 0.8, transform shape]
		\intervalp{$I_{f_i}$}{-3}{-0.5}{0}{-0.7}
		\interval{}{-3.5}{-1}{0.5}
		\intervalp{$I_{l_i'}$}{-4}{-1.5}{1}{-1.75}
		\path (-2.75, 1.6) -- (-2.75, 1.6) node[font=\LARGE, midway, sloped]{$\dots$};
		\node[] at (-4.75,1.8){$(d)$};
		\end{tikzpicture}
	\end{minipage}
	\caption{Intervals with predictions indicated as red crosses. {$a)$ Prediction mandatory free cycle. Illustration of the situations in the Lemmas~\ref{mst_phase1_case_a} $(b)$, \ref{mst_phase1_case_b} $(c)$ and~\ref{mst_phase1_case_c} $(d)$.}} 
	\label{Ex_mst_phase1_cases}
\end{figure}
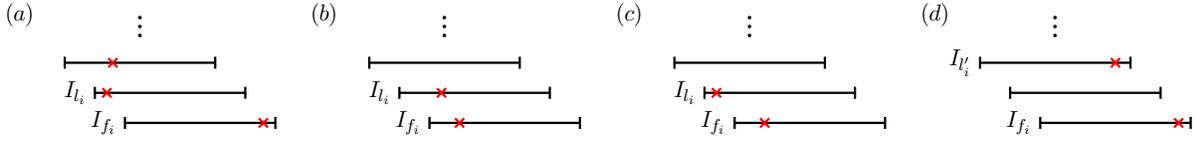

\subsection{Handling prediction mandatory free instances}
\label{subsec_mst_phase_2}

This section describes Phase 2 of our algorithms. For prediction mandatory free instances, we present Algorithm~\ref{ALG_mst_part_2} with \emph{recovery strategies} A and B (in Line~\ref{line_mst_two_recovery}) that lead to the guarantees of Lemma~\ref{mst_end_of_phase_2}.  
Our full algorithms execute Phase $1$ followed by Phase $2$ and differ only in the recovery strategy. 
Using the introduced ideas and lemmas, we can {prove}
the MST part of Theorem~\ref{thm:main1.5-2} (recovery A) and Theorem~\ref{theorem_mst2}~(recovery~B).

\begin{restatable}{lem}{MSTEndOfPhaseTwo}
	\label{mst_end_of_phase_2}
	If Algorithm~\ref{ALG_mst_part_2} is executed on a prediction mandatory free instance, then in each iteration the instance remains prediction mandatory free. Furthermore, recovery strategy A guarantees {$1$-consistency} and $2$-robustness and recovery strategy B guarantees $|\ALG| \le \min\{\opt + {5}\cdot k_h, 3\cdot \opt\}$.
\end{restatable}

Assume again unique $T_L = T_U$. In a prediction mandatory free instance $G=(V,E)$, each $f_i \in E \setminus T_L$ is predicted to be maximal on cycle $C_i$, and each $l \in T_L$ is predicted to be minimal in $X_l$. 
{If these predictions are correct, then the optimal query set is a minimum vertex cover in a bipartite graph $\bar{G} = (\bar{V},\bar{E})$ with $\bar{V} = E$ (excluding trivial edges) and $\bar{E} = \{ \{f_i,e\} \mid i \in \{1,\ldots,l\}, e \in C_i \setminus \{f_i\} \text{ and } I_e \cap I_{f_i} \not= \emptyset\}$ \cite{erlebach14mstverification,megow17mst}. We refer to $\bar{G}$ as the {\em vertex cover instance}.}
Note that {if a query reveals that} an $f_i$ is not maximal on $C_i$ or an $l_i$ is not minimal in $X_{l_i}$, then the vertex cover instance changes.
Because of this and in contrast to the minimum problem, non-adaptively querying $VC$ can lead to a competitiveness worse than $2$.

Let $VC$ be a minimum vertex cover of $\bar{G}$.
The idea of the algorithm is to sequentially query each $e \in VC$ and charge for querying $e$ by a distinct non-queried element $h(e)$ such that $\{e,h(e)\}$ is a witness set.
Querying exactly one element per
distinct witness set implies optimality.
To identify~$h(e)$ for each element $e \in VC$, we use the fact that
K\"onig's Theorem (e.g,~\cite{Biggs1986}) and the duality between minimum vertex covers and maximum matchings in bipartite graphs imply that there is a matching $h$ that maps each $e \in VC$ to a distinct $e' \not\in VC$. 
While the sets $\{e,h(e)\}$ with $e \in VC$ in general are not witness sets, querying $VC$ in a specific order until the vertex cover instance changes guarantees that $\{e,h(e)\}$ is a witness set for each already queried~$e$.  

\begin{restatable}{lem}{MstPhaseTwoOne}
	\label{mst_phase2_1}
	Let $f'_1,\ldots,f'_g$ be the edges in $VC \setminus T_L$ ordered by lower limit non-decreasingly and let $l'_1,\ldots,l'_k$ be the edges in $VC \cap T_L$ ordered by upper limit non-increasingly. 
	Let $b$ be such that each $f'_{i}$ with $i < b$ is maximal in cycle $C_{f'_i}$, then $\{f'_{i},h(f'_{i})\}$ is a witness set for each $i \le b$.
	Let $d$ be such that each $l'_{i}$ with $i < d$ is minimal in cut $X_{l'_i}$, then $\{l'_{i},h(l'_{i})\}$ is a witness set for each $i \le d$.
\end{restatable}

\begin{algorithm}[tb]
	\KwIn{Prediction mandatory free graph $G=(V,E)$ and predictions $\pred{w}_e$ for each {$e \in E$}}
	Compute maximum matching $h$ and minimum vertex cover $VC$ for $\bar{G}$ and initialize $W = \emptyset$\label{line_mst_two_compute_vc}\;
	Let $f'_1,\ldots,f'_g$ and $l'_1,\ldots,l'_k$ be as described in Lemma~\ref{mst_phase2_1}\label{line_mst_two_l_edges}\;
	\For{$e$ chosen sequentially from the {ordered list} $f'_1,\ldots,f'_g,l'_1,\ldots,l'_k$\label{line_mst_two_main_for}}{
		Query $e$, and ensure unique $T_L = T_U$\label{line_mst_two_query_vc}\; 
		If $\{e,h(e)\} \cap W = \emptyset$, add $h(e)$ to $W$. Otherwise query $h(e)$ {and ensure unique $T_L=T_U$}\label{line_mst_two_ensure_three_robust}\;
		If the vertex cover instance changed, execute a recovery strategy and restart at Line $2$\label{line_mst_two_recovery}\label{line_mst_two_restart_if}\;
	}
	\caption{Phase $2$ of the algorithms for MST under uncertainty}
	\label{ALG_mst_part_2}
\end{algorithm}

Algorithm~\ref{ALG_mst_part_2} queries $VC$ in the order $f'_1,\ldots,f'_g,l'_1,\ldots,l'_k$ (Lemma~\ref{mst_phase2_1}) until the vertex cover instance {$\bar{G}$} changes. 
If it does not change, $|VC|$ is a lower bound on $\opt$ and  the algorithm only queries $VC$ and a set $M$ of mandatory elements that were queried as elements of $T_L\setminus T_U$.
Since we can show that each $e \in M$ contributes one to $k_h$, it follows $|\ALG| \le \min\{\opt +k_h,2\cdot \opt\}$.

If the vertex cover instance changes, Line~\ref{line_mst_two_recovery} executes a recovery strategy and restarts.
The challenge here is that, after restarting, an element $h(e)$ that was used to charge for an already queried $e$ might be used again to charge for a different element $e'$.
This can happen if $h(e)$ is not queried and matched to element $e'$ after the restart. To handle this problem the algorithm uses the set $W$ to keep track of the elements $h(e)$ for already queried elements $e$ and uses them in charging schemes when executing the recovery strategies.

\medskip 
\noindent {\bf Recovery strategy A} The first recovery strategy can be used to achieve $2$-robustness and show the MST part of Theorem~\ref{thm:main1.5-2}.
It queries before a restart the set $W$ of non-queried elements that were used to charge for already queried elements, ensures unique $T_L=T_U$ and restarts with re-computed $VC$ and $h$. Since each $h(e) \in W$ forms a witness set with a distinct $e$ (Lemma~\ref{mst_phase2_1}), this strategy ensures~$2$-robustness.

\medskip
\noindent {\bf Recovery strategy B} This strategy can be used to achieve the error-sensitive guarantee $|\ALG| \le \min\{\opt + {5} \cdot k_h, 3 \cdot \opt\}$.
In this case querying $W$ to ensure $2$-robustness might violate $|\ALG| \le \opt + {5} \cdot k_h$. 
Instead of preventing that an element $h(e)$ is used to charge for a second element, the algorithm prevents $h(e)$ from being used to charge for three elements.
To achieve this, Line~\ref{line_mst_two_ensure_three_robust} queries $h(e)$ when it is used to charge for a second element $e'$. 
As $U = \{h(e),e,e'\}$ is a witness set, this ensures $3$-robustness. Furthermore, two elements of $U$ might not be part of an optimal solution. However, executing the Otherwise-part of Line~\ref{line_mst_two_ensure_three_robust} at most $2 \cdot k_h$ times ensures $|\ALG| \le \opt + {5} \cdot k_h$.
Note that the factor is $5$ instead of $4$ because the errors used to bound the number of executions of the Otherwise-part of Line~\ref{line_mst_two_ensure_three_robust} and the errors used to charge for the queried elements of $T_L \setminus T_U$ in Lines~\ref{line_mst_two_query_vc} and~\ref{line_mst_two_ensure_three_robust} are not necessarily disjoint.
This uses Lemma~\ref{mst_phase2_2} and a restart with a specific matching $h'$ that does not contain too many non-queried elements of $W$. 

\begin{restatable}{lem}{MstPhaseTwoTwo}
	\label{mst_phase2_2}
	Let $\bar{G}'= (\bar{V}',\bar{E}')$ be the changed vertex cover instance of Line~\ref{line_mst_two_recovery}, then $\overline{h}=\{\{e,e'\} \in h \mid \{e,e'\} \in \bar{E}'\}$ defines a partial matching for $\bar{G'}$. Let $h'$ be the maximum matching for $\bar{G}'$ computed by completing $\bar{h}$ using
	{a standard augmenting path algorithm~\cite{AhujaMO1993book}} and let $VC'$ be the vertex cover defined by $h'$. Then restarting Algorithm~\ref{ALG_mst_part_2} with $VC=VC'$ and $h=h'$ implies that Line~\ref{line_mst_two_ensure_three_robust} queries at most $2 \cdot k_h$~times.
\end{restatable}

\section{The sorting problem}
\label{sec:sorting}
There is a simple reduction from the sorting problem under uncertainty to the minimum problem.
For each intersecting pair of intervals in the same set in the sorting problem, we construct a corresponding set in the minimum problem, consisting of only this pair.
If we have a solution for the sorting problem, then we clearly can determine the minimum interval in each pair.
Conversely, if we know the minimum interval in each pair, then we can sort the intervals accordingly.
Thus the sorting problem is a particular case of the minimum problem with sets of size 2, and all our algorithmic results for the minimum problem transfer to the sorting problem.
However, the sorting problem has a simpler characterization for mandatory intervals: any interval that contains the true value of another interval in the same set is mandatory~\cite{halldorsson19sortingqueries}.
Still, 
we cannot hope for better guarantees for sorting due to the lower bounds in Theorems~\ref{theo_minimum_combined_lb} and~\ref{theo_lb_sym_diff}.
Moreover, the NP-hardness of the verification problem of the minimum holds for sets of size 2 and, thus, it holds  for sorting overlapping~sets. 

Nevertheless, if we are sorting a single set or disjoint sets, then there is a better algorithm than the ones presented for the minimum problem.
The algorithm  performs at most $\min\{\opt + k, 2 \cdot \opt\}$ queries, for any $k \in \{k_{\#}, k_M, k_h\}$, and thus satisfies Theorem~\ref{thm:sorting:singleset} and is optimal due to Theorem~\ref{thm_lb_error_measure}.
The main fact that makes the sorting problem easier for a single set is that the intersection graph is an interval graph~\cite{lekkeikerker62interval}.
Moreover, witness sets are not influenced by other intervals, because any two intersecting intervals constitute a witness set~\cite{halldorsson19sortingqueries}.
We give a full proof of Theorem~\ref{thm:sorting:singleset} in Appendix~\ref{app:sorting}; in the remainder of the section we present the algorithm and high level arguments on why it obtains the desired result.

To obtain a guarantee of $\opt + k$ for any measure~$k$, the algorithm must trust the predictions as much as possible.
In particular, our algorithm queries all prediction mandatory intervals}.
We can do this and still guarantee 2-robustness because intervals that contain the same predicted value form a clique in the initial interval graph, and in any clique at most one query can be avoided~\cite{halldorsson19sortingqueries}.
However, in the second phase of the framework, when there are no more prediction mandatory intervals and we query a minimum vertex cover and intervals that become known mandatory, we may have different minimum vertex covers, so in order to maintain 2-robustness we must be more careful.
We show that each component of the intersection graph at this point must be a path.
The more intricate case is when we have an even path~$P$, because we have two minimum vertex covers; we devise a charging scheme to decide which of them we query.
This scheme is based on a forest of arborescences that is built according to the relation between prediction mandatory intervals and predicted values, which we define more precisely in the next paragraph.
This forest is then used to partition the prediction mandatory intervals into sets that correspond to cliques in the initial interval graph, in such a way that only roots or children of the roots can be isolated in this clique partition.
(When a root is isolated, we show that we can build a different clique partition for that component without isolated intervals.)
We then show that only the endpoints of~$P$ can be part of this forest of arborescences.
Thus~$P$ along with the isolated vertices in the clique partition that are children of the endpoints of~$P$ constitute an induced subgraph of size at most $|P|+2$.
If the size is $|P| + 1$, then we have to ensure that we query the vertex cover that contains the endpoint that is a parent of an isolated vertex in the clique partition; otherwise we can choose an arbitrary vertex cover.
On the other hand, if we have an odd path, then it has a single minimum vertex cover, and we show that this is always a good choice.

\begin{algorithm}[tb]
\KwIn{Ground set of intervals~$\mathcal{I} = \{I_1, \ldots, I_n\}$, and predictions $\pred{w}_1, \ldots, \pred{w}_n$}
  $\mathcal{E} \leftarrow \emptyset$; \quad $C_1, C_2, \ldots, C_n \leftarrow \emptyset$\;
  \KwLet $\mathcal{I}_P$ be the set of prediction mandatory intervals\label{line:computepm}\;
  \lForEach{$I_i \in \mathcal{I}_P$}{$\pi(i) \leftarrow j$ for some $j \neq i$ with $\pred{w}_j \in I_i$\label{line:findparent}}
  \lWhile{$\exists i \neq j$ with $I_j \subseteq I_i$, or~$I_j$ was queried and $w_j \in I_i$}{query $I_i$\label{line:mandatoryproper}}
  \KwLet $\mathcal{S}$ be the set of intervals in $\mathcal{I}_P$ that were not queried yet\label{line:fixS}\;
  \ForEach{$I_i \in \mathcal{S}$}{
    query $I_i$\label{line:containpredicted}\;
    \lIf{$(I_{\pi(i)}, I_i)$ does not create a cycle in $(\mathcal{I}, \mathcal{E})$}{$\mathcal{E} \leftarrow \mathcal{E} \cup \{(I_{\pi(i)}, I_i)\}$\label{line:addedge}}
  }
  \lWhile{$\exists i \neq j$ where~$I_j$ was queried and $w_j \in I_i$}{query $I_i$\label{line:mandatory}}
  \While{$\mathcal{S} \neq \emptyset$\label{line:cliquepartitionbegin}}{
    \KwLet $I_i$ be a deepest vertex in the forest of arborescences $(\mathcal{I}, \mathcal{E})$ among those in~$\mathcal{S}$\;
    \If{$(I_{\pi(i)}, I_i) \in \mathcal{E}$}{
      $C_{\pi(i)} \leftarrow \{I_{i'} \in \mathcal{S} : (I_{\pi(i)}, I_{i'}) \in \mathcal{E}\}$\;
      \lIf{$I_{\pi(i)} \in \mathcal{S}$}{$C_{\pi(i)} \leftarrow C_{\pi(i)} \cup \{I_{\pi(i)}\}$}
      $\mathcal{S} \leftarrow \mathcal{S} \setminus C_{\pi(i)}$\label{line:cliquepartitionend}\;
    }
    \lElse{$C_i \leftarrow \{I_i\}$; \quad $\mathcal{S} \leftarrow \mathcal{S} \setminus C_i$\label{line:rootisolated}}
  }  
  \While{the problem is unsolved\label{line:vcbeginsort}}{
    \KwLet $P = x_1 x_2 \cdots x_p$ be a component of the current intersection graph which is a path with $p \geq 2$\; \label{line:looppaths}
    \lIf{$p$ is odd}{query $I_{x_2}, I_{x_4}, \ldots, I_{x_{p-1}}$\label{line:pathodd}}
    \Else{
      \lIf{$|C_{x_1}| = 1$}{query $I_{x_1}, I_{x_3}, \ldots, I_{x_{p-1}}$\label{line:patheven1}}
      \lElse{query $I_{x_2}, I_{x_4}, \ldots, I_{x_p}$\label{line:patheven2}}
    }
    \lWhile{$\exists i \neq j$ where~$I_j$ was queried and $w_j \in I_i$}{query $I_i$\label{line:mandatory2}}
  }
\caption{A nicely degrading algorithm for sorting with predictions.}
\label{fig:sorting1cons2rob}
\end{algorithm}

A pseudocode is given in Algorithm~\ref{fig:sorting1cons2rob}.
The first phase of the algorithm consists of Lines~\ref{line:computepm}--\ref{line:mandatory}, in which we query known mandatory and prediction mandatory intervals.
We fix the set $\mathcal{I}_P$ of initial prediction mandatory intervals, and for each interval $I_i \in \mathcal{I}_P$ we assign a {\em parent} $\pi(i)$, meaning that $\pred{w}_{\pi(i)} \in I_i$.
Next we query known mandatory intervals to ensure that the intersection graph becomes a proper interval graph, and let $\mathcal{S}$ be the remaining intervals in $\mathcal{I}_P$.
We then query every $I_i \in \mathcal{S}$ and include in set~$\mathcal{E}$ a directed edge $(I_{\pi(i)}, I_i)$ if that does not create a cycle in the graph $(\mathcal{I}, \mathcal{E})$; this graph $(\mathcal{I}, \mathcal{E})$ is the forest of arborescences previously mentioned.
In Lines~\ref{line:cliquepartitionbegin}--\ref{line:rootisolated}, we partition~$\mathcal{S}$ into cliques of the initial interval graph.
We traverse each component of the forest from the deepest leaves towards the root, so we guarantee that only roots or children of the roots can be isolated in the final partition.
The second phase of the framework consists of Lines~\ref{line:vcbeginsort}--\ref{line:mandatory2}: We use the size of the sets in the clique partition to decide between different minimum vertex covers, and then we query intervals that become known mandatory.

It is clear that this algorithm can be implemented in polynomial time.

\section{Experimental results}
\label{sec:exp}
{We tested the practical performance of our algorithms in simulations and highlight here the results for the minimum problem.} Further results on the MST as well as details on the generation of instances and predictions are  provided in the appendix.
Our instances {were generated} by randomly drawing interval sets from interval graphs,  {obtained from re-interpreted SAT instances from the rich SATLIB library~\cite{hoos2000satlib}.}
Our instances have between~$48$ and $287$ intervals and a variable number of overlapping sets.
For each instance we generated $125$ different predictions while ensuring that the predictions cover a wide range \mbox{of relative errors $k_M/\opt$}.

Figure~\ref{fig_experiments} shows the results of over $230,000$ simulations (instance and predictions pairs).
The figure compares the results of our prediction-based {Algorithms~\ref{ALG_min_beta} and~\ref{ALG_min_alpha} for different choices of the parameter $\gamma$ with the standard {\em witness set algorithm}. The latter sequentially queries witness sets of size two and achieves the best possible competitive ratio of $2$ without predictions~\cite{kahan91queries}.}
The Algorithms~\ref{ALG_min_beta} and~\ref{ALG_min_alpha}, for every selected choice of $\gamma$, outperform the witness set algorithm up to a relative error of approximately $2.8$ and $1.5$, respectively.
For small values of $\gamma$, Algorithm~\ref{ALG_min_beta} outperforms the witness set algorithm even for \emph{every} relative error.
{Further, the parameter $\gamma$ reflects well the robustness-performance tradeoff for both algorithms: } a high value $\gamma$ is beneficial for accurate predictions while a low value for $\gamma$ gives robustness against very inaccurate predictions.
{In the extreme case, $\gamma = |\mathcal{I}|$,  Algorithm~\ref{ALG_min_alpha} directly follows the predictions; while it it is superior for small errors, it gets outperformed by algorithms with smaller $\gamma$ when the relative error is growing.}
The results indicate that Algorithm~\ref{ALG_min_beta} performs better than Algorithm~\ref{ALG_min_alpha}.
For both algorithms the performance gap between the different values for $\gamma$ appears less significant for small relative errors, 
which suggests that selecting $\gamma$ not too close to the maximum value $|\mathcal{I}|$ might be beneficial.
For Algorithm~\ref{ALG_min_beta} the results even suggest that selecting $\gamma = 2$ might be the most beneficial choice.
\begin{figure}[tbh]
	\centering
	
	\begin{tikzpicture}
	    \begin{groupplot}[group style={group size= 2 by 1,horizontal sep = 1.25cm}]
	    	\nextgroupplot[
					,title style={font=\tiny}
					,xlabel={$k_M / \opt$}
					,ylabel={Competitive ratio (mean)}
					,grid=minor
					,xmin = 0
					,ymin = 1
					,label style={font=\tiny},
					,tick label style={font=\tiny}  
					,legend entries={
									Alg.~\ref{ALG_min_beta} ($\gamma = 2$),
									Alg.~\ref{ALG_min_beta} ($\gamma = 3$),
									Alg.~\ref{ALG_min_beta} ($\gamma = 4$),
									Alg.~\ref{ALG_min_beta} ($\gamma = 8$),
									Alg.~\ref{ALG_min_beta} ($\gamma = \mathcal{I}$),
									Alg.~\ref{ALG_min_alpha} ($\gamma = 2$),
									Alg.~\ref{ALG_min_alpha} ($\gamma = 3$),									
									Alg.~\ref{ALG_min_alpha} ($\gamma = 4$),										
									Alg.~\ref{ALG_min_alpha} ($\gamma = 8$),					 					
									Alg.~\ref{ALG_min_alpha} ($\gamma = 12$),									
									Alg.~\ref{ALG_min_alpha} ($\gamma = |\mathcal{I}|$),	
									Witness Set Alg.}
					,legend style={at={(axis cs:4.9,2.1)},anchor= east, font = \tiny}
					,table/col sep=comma,
					width=10cm,
					height=5.5cm]
				\addplot[blue!60!black,mark color=white,mark=halfsquare*] table [x = Error_bin, y = PredictionOptimalAlgorithm_002_cr] {exp_results/min_error_sum.csv};
				\addplot[cyan!60!black,mark color=white,mark=halfsquare*,mark options={rotate=180}] table [x = Error_bin, y = PredictionOptimalAlgorithm_003_cr] {exp_results/min_error_sum.csv};
				\addplot[yellow!60!black,mark color=white,mark=halfsquare left*] table [x = Error_bin, y = PredictionOptimalAlgorithm_004_cr] {exp_results/min_error_sum.csv};
				\addplot[orange!60!black,mark color=white,mark=halfsquare right*] table [x = Error_bin, y = PredictionOptimalAlgorithm_008_cr] {exp_results/min_error_sum.csv};
				\addplot[red!60!black,mark color=red!60!black,mark=halfsquare right*] table [x = Error_bin, y = PredictionOptimalAlgorithm_292_cr] {exp_results/min_error_sum.csv};
					
				\addplot[blue,mark color=white,mark=halfcircle*,mark options={rotate=180}] table [x = Error_bin, y = ParameterizedAlgorithm_002_cr] {exp_results//min_error_sum.csv};	
				\addplot[cyan,mark color=white,mark=halfcircle*] table [x = Error_bin, y = ParameterizedAlgorithm_003_cr] {exp_results//min_error_sum.csv};	
				\addplot[yellow!75!black,mark color=white,mark=halfcircle*,mark options={rotate=90}] table [x = Error_bin, y = ParameterizedAlgorithm_004_cr] {exp_results/min_error_sum.csv};	
				\addplot[orange,mark color=white,mark=halfcircle*,mark options={rotate=270}] table [x = Error_bin, y = ParameterizedAlgorithm_008_cr] {exp_results/min_error_sum.csv};	
				\addplot[red,mark color=red,mark=*] table [x = Error_bin, y = ParameterizedAlgorithm_012_cr] {exp_results/min_error_sum.csv};	
	 			\addplot[purple!60!black,mark color=purple!60!black,mark=triangle*] table [x = Error_bin, y = ParameterizedAlgorithm_292_cr] {exp_results/min_error_sum.csv};	
	 			\addplot[darkgray,mark color=darkgray,mark=diamond*] table [x = Error_bin, y = WitnessSetAlgorithm_cr] {exp_results/min_error_sum.csv};
	    \end{groupplot}
		
	\end{tikzpicture}
	\caption{Experimental results for the minimum problem under uncertainty. Instances and predictions were grouped into equal size bins ($0.2$) according to their relative error $k_M/\opt$.}
	\label{fig_experiments}
\end{figure}
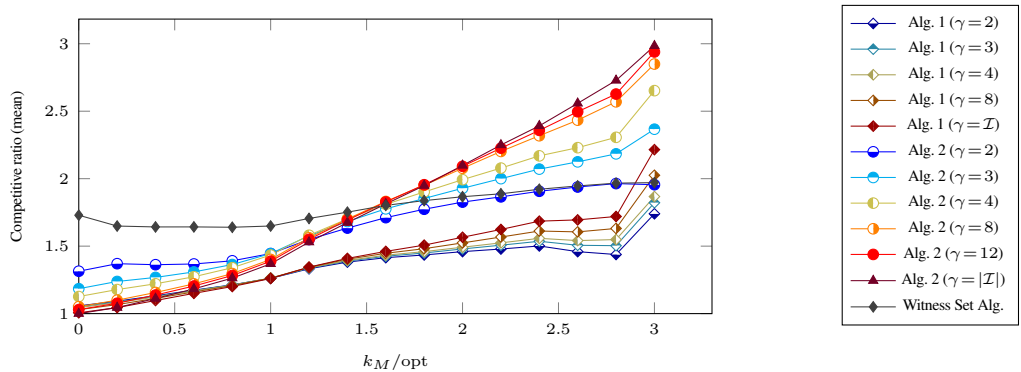

\section*{Final remarks}
In this paper we propose to use (possibly machine-learned) predictions for improving query-based algorithms when coping with explorable uncertainty. Our methods prove that untrusted predictions allow for rigorous worst-case guarantees that {overcome} known lower bounds. We also discuss different measures for the inaccuracy of predictions, which might be of independent interest, and for which we show the learnability of predictions with small error.
By providing trustable guarantees, we contribute to the challenge of building trustable AI systems and the applicability also for safety-critical applications where such guarantees are obligatory. 
It would be interesting to study the power of predictions for other (theoretically interesting and practically relevant) problems and, possibly, for other predictor models. {We also hope to foster research on explorable uncertainty with untrusted predictions for other natural (optimization) problems.} 

While we ask, in our work, for the minimum number of queries to solve a problem {\em exactly}, it would be natural to ask for approximate solutions. The bad news is that for all problems considered here, there is no improvement over the robustness guarantee of $2$ possible even when allowing an arbitrarily large approximation of the exact solution. This follows directly from a lower bound example with two uncertain elements used in the context of finding MSTs in~\cite[Section 10]{megow17mst}. However, it remains open whether an improved consistency resp. error-dependent competitive ratio is possible.

In contrast to our adversarial model, one may assume that the realization of an uncertain value is drawn randomly according to some known distribution. Improved results are possible {\em in expectation}, as has been shown for sorting~\cite{ChaplickHLT20} and the minimum problem~\cite{BampisDEdLMS21},
as well as a related scheduling problem~\cite{LeviMS19}. Still, it would be interesting whether learning-augmented algorithms can improve upon such probabilistic results and can add a robustness guarantee that holds for any realization. 

\medskip \noindent {\bf Acknowledgement.} We thank Alexander Lindermayr for his great support in conducting the experiments, and Michael Hoffmann for
very helpful discussions in the early phase of this research.

\newpage 
\normalsize
\appendix

\section{Appendix {for Distance Measures} and Lower Bounds (Section \ref{sec:prelim})}
\subsection{Lower bound on the consistency-robustness tradeoff}

\ThmLBTradeoffWithoutError*

\begin{proof}
We state the proof for the minimum problem first.
Assume, for the sake of contradiction, that there is a deterministic $\beta$-robust algorithm that is $\alpha$-consistent with $\alpha=1 + \frac{1}{\beta}-\eps$, for some $\eps>0$.
Consider the instance in Figure~\ref{fig_combined_lb_min_single_set} with $\beta+1$ intervals and a single set. 
The algorithm must query the intervals $\{I_{1}, \ldots, I_{\beta}\}$ first as otherwise, it would query $\beta+1$ intervals in case all predictions are correct, while there is an optimal query set of size $\beta$.  
Suppose w.l.o.g.\ that the algorithm queries the intervals $\{I_{1}, \ldots, I_{\beta}\}$ in order of increasing indices. Consider the adversarial choice $w_i = \pred{w}_i$, for $i = 1, \ldots, \beta - 1$, and then $w_{\beta} \in I_0$ and $w_0 \notin I_1 \cup \ldots \cup I_{\beta}$.
This forces the algorithm to query also~$I_0$, while an optimal solution only queries~$I_0$.
Thus any such algorithm has robustness at least~$\beta+1$, a contradiction.

The second part of the theorem directly follows from the first part and the known general lower bound of $2$ on the competitive ratio~\cite{erlebach08steiner_uncertainty,kahan91queries}. Assume there is an $\alpha$-consistent deterministic algorithm with some $\alpha=1 + \frac{1}{\beta'}$, for some $\beta'\in [1,\infty)$. Consider the instance above with $\beta=\beta'-1$. Then the algorithm has to query intervals $\{I_{1}, \ldots, I_{\beta}\}$ first to ensure $\alpha$-consistency as otherwise it would have a competitive ratio of $\frac{\beta+1}{\beta}>1+\frac{1}{\beta'}=\alpha$ in case that all predictions are correct. By the argumentation above, the robustness factor of the algorithm is at least $\beta+1=\beta'=\frac{1}{\alpha-1}$.

The same arguments can be used for the sorting problem, the only difference is that we take the input sets
$\{I_0,I_i\}$ for $1\le i\le \beta$.

For the MST problem, we first translate the above construction for the minimum problem to the maximum problem
(defined in the obvious way) and then consider the MST instance consisting of a single cycle whose edges are
associated with the weight intervals $I_0,I_1,\ldots,I_\beta$.
\end{proof}

\begin{figure}[b]
  \centering
  \subfigure[]{\label{fig_combined_lb_min_single_set}
  \begin{tikzpicture}[line width = 0.3mm, scale=0.8]
    \intervalpr{$I_1$}{1}{3}{0}{2.5}{2.5}
    \intervalpr{$I_2$}{1}{3}{0.5}{2.5}{2.5}
    \intervalpr{$I_\beta$}{1}{3}{1.5}{2.5}{1.5}
    \path (2, 0.5) -- (2, 1.6) node[font=\normalsize, midway, sloped]{$\dots$};

    \intervalpr{$I_0$}{0}{2}{-0.5}{1.5}{0.5}
  \end{tikzpicture}}\quad
  \subfigure[]{\label{fig_lb_wrong_predictions}
  \begin{tikzpicture}[line width = 0.3mm, scale=0.8]
    \intervalpr{$I_1$}{0}{2}{0}{0.5}{0.5}
    \intervalpr{$I_2$}{0}{2}{0.5}{0.5}{0.5}
    \intervalpr{$I_n$}{0}{2}{1.5}{0.5}{1.5}
    \path (1, 0.5) -- (1, 1.6) node[font=\normalsize, midway, sloped]{$\dots$};

    \intervalpr{$I_{n+1}$}{1}{3}{2}{2.5}{2.5}
    \intervalpr{$I_{n+2}$}{1}{3}{2.5}{2.5}{2.5}
    \intervalpr{$I_{2n}$}{1}{3}{3.5}{2.5}{2.5}
    \path (2, 2.5) -- (2, 3.6) node[font=\normalsize, midway, sloped]{$\dots$};
  \end{tikzpicture}}\quad
  \subfigure[]{\label{fig_lb_error_measure}
  \begin{tikzpicture}[line width = 0.3mm, scale=0.8]
    \intervalpr{$I_1$}{0}{2}{0}{1.5}{1.5}
    \intervalpr{$I_2$}{1}{3}{0.5}{1.5}{2.5}
  \end{tikzpicture}}
   \caption{Instances for lower bounds.
  Red crosses indicate predicted values, and green circles show correct values.
  \subref{fig_combined_lb_min_single_set}~Lower bound on robustness-consistency tradeoff.
  \subref{fig_lb_wrong_predictions} Lower bound based on the number of inaccurate predictions.
  \subref{fig_lb_error_measure}~Lower bound based on error measures.
  \label{fig_lowerbounds}
  }
\end{figure}
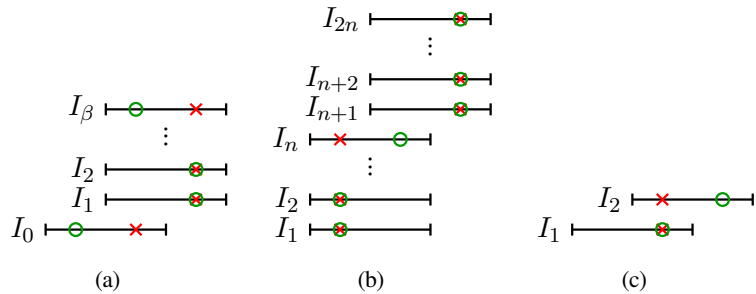

\subsection{Number of inaccurate predictions}

Let $k_{\#}$ denote the \emph{number of inaccurate predictions}, i.e., the number of intervals $I_i \in \mathcal{I}$ with $w_i \not= \w_i$.

This is a very natural but impractical prediction measure as the following theorem shows. It rules out using predictions to improve the competitive ratio on the known lower bound of $2$.

\begin{theorem}\label{theo_lb_wrong_predictions}
  If $k_{\#} \geq 1$, then any deterministic algorithm for the minimum, sorting or MST problem under uncertainty has competitive ratio $\rho\geq 2$.
\end{theorem}

\begin{proof}
First, we discuss the proof for the minimum problem.
Consider $2n$ intervals as depicted in Figure~\ref{fig_lb_wrong_predictions} and sets 
$S_i=\{I_i, I_{n+1}, I_{n+2}, \ldots, I_{2n}\}$, for $i = 1, \ldots, n$.
Assume w.l.o.g.\ that the algorithm queries the left-side intervals in the order $I_1, I_2, \ldots, I_n$ and the right side in the order $I_{n+1}, I_{n+2}, \ldots, I_{2n}$.
Before the algorithm queries~$I_n$ or~$I_{2n}$, the adversary sets all predictions as correct, so the algorithm will eventually query~$I_n$ or~$I_{2n}$.
If the algorithm queries~$I_n$ before~$I_{2n}$, then the adversary chooses a value for $I_n$ that forces a query in $I_{n+1}, \ldots, I_{2n}$, and the {predicted} values for the remaining right-side intervals as correct, so the optimum solution only queries $I_{n+1}, \ldots, I_{2n}$.
A symmetric argument holds if the algorithm queries~$I_{2n}$ before~$I_n$.

For the sorting problem, use the same intervals but take the $n^2$ sets
$\{I_i,I_j\}$ for $1\le i\le n, n+1\le j\le 2n$.

For the MST problem, first translate the above construction to the maximum problem.
Call the resulting intervals $I_1',\ldots,I_{2n}'$, and note that both the true and
the predicted values of $I_i'$ for $1\le i\le n$ are larger than both the real
and the predicted values of $I_j'$ for $n+1\le j\le 2n$.
Consider a graph that consists of a path with $n$ edges with weight intervals
$I_{n+1}',\ldots,I_{2n}'$, and let $s$ and $t$ denote the two end vertices
of that path. Then add $n$ parallel edges between $s$ and $t$ with weight
intervals $I_1',\ldots,I_n'$. The adversary then proceeds as in the construction
for the minimum problem.
\end{proof}

\subsection{Mandatory query distance}

\HopDistanceMandatoryDistance*

\begin{proof}
We first discuss the minimum problem. Consider an instance with uncertainty intervals~$\mathcal{I}$, 
true values $w$ and predicted values $\pred{w}$. 
Recall that $\mathcal{I}_P$ and
$\mathcal{I}_R$ are the sets of prediction mandatory elements and real mandatory elements, respectively.
 Observe that $k_M$ counts the intervals that
are in $\mathcal{I}_P\setminus \mathcal{I}_R$ and those that are in $\mathcal{I}_R\setminus \mathcal{I}_P$.
We will show the claim that, for
every interval $I_i$ in those sets, there is an interval $I_j$ such that
the value of $I_j$ passes over $L_i$ or $U_i$ (or both) when going
from $\w_j$ to $w_j$.
This means
that each interval $I_i\in \mathcal{I}_P \sym \mathcal{I}_R$ is mapped to a unique
pair $(j,i)$ such that the value of $I_j$ passes over at least one endpoint of~$I_i$, and
hence each such pair contributes at least one to~$k_h$. This implies $k_M \le k_h$.

It remains to prove the claim. Consider an $I_i\in \mathcal{I}_P\setminus \mathcal{I}_R$. (The argumentation
for intervals in $\mathcal{I}_R\setminus \mathcal{I}_P$ is symmetric, with the roles of $w$ and $\w$ exchanged.)
As $I_i$ is not in $\mathcal{I}_R$, replacing all intervals
in $\mathcal{I}\setminus\{I_i\}$ by their true values yields an instance
that is solved.
This means that in every set $S\in\mathcal{S}$ that contains $I_i$, one of the following
cases holds:
\begin{itemize}
\item[(a)] $I_i$ is known not to be the minimum of $S$ {w.r.t.\ true values $w$}. It follows that there
is an interval $I_j$ in $S$ with $w_j\le L_i$.
\item[(b)] $I_i$ is known to be the minimum of $S$ {w.r.t.\ true values $w$}. It follows that all
intervals $I_j\in S\setminus\{I_i\}$ satisfy $w_j\ge U_i$.
\end{itemize}
As $I_i$ is in $\mathcal{I}_P$, replacing all intervals
in $\mathcal{I}\setminus\{I_i\}$ by their predicted values yields an instance
that is not solved.
This means that there exists at least one set $S'\in\mathcal{S}$ that contains $I_i$ and satisfies
the following:
\begin{itemize}
\item[(c)] All intervals $I_j$ in $S'\setminus\{I_i\}$ satisfy $\w_j>L_i$, and there
is at least one such $I_j$ with $L_i<\w_j<U_i$.
\end{itemize}
If $S'$ falls into case (a) above, then by (a) there is an interval
$I_j$ in $S'$ with $w_j\le L_i$, and by (c) we have $\w_j>L_i$. This
means that the value of $I_j$ passes over~$L_i$.
If $S'$ falls into case (b) above, then by (c) there exists an
interval $I_j$ in $S'$ with $\w_j < U_i$, and by (b) we have $w_j\ge U_i$.
Thus, the value of $I_j$ passes over~$U_i$. This establishes the claim,
and hence we have shown that $k_M\le k_h$ for the minimum problem. 

We show in Appendix~\ref{app:sorting} that every instance
of the sorting problem can be transformed into an equivalent instance
of the minimum problem on the same uncertainty intervals. The transformation
ensures that, for any two intervals $I_i$ and $I_j$, there exists a set
in the sorting instance that contains both $I_i$ and $I_j$ if and only if
there exists a set in the minimum instance that contains both $I_i$ and $I_j$.
Therefore, the error measures for both instances are the same, and the
above result for the minimum problem implies that $k_M\le k_h$ also holds
for the sorting problem.

Finally, we consider the MST problem. Consider an instance $G=(V,E)$ with
uncertainty interval $I_e=(L_e,U_e)$, true value $w_e$ and predicted
value $\w_e$ for $e\in E$.
Let $E_P$ and $E_R$ be the mandatory queries with respect to
the predicted and true values, respectively. Again, $k_M$ counts the edges
in $E_P\sym E_R$.
We claim that, for
every interval $I_e$ of an edge $e$ in this set, there is an interval $I_g$
of an edge $g$ such that
the value of $I_g$ passes over $L_e$ or $U_e$ (or both) when going
from $\w_e$ to $w_e$ and $g$ lies on a cycle with~$e$.
This means
that each edge $e\in E_P \sym E_R$ is mapped to a unique
pair $(e,g)$ such that~$g$ lies in the same biconnected component
as $e$ and the value of $I_g$ passes over at least one endpoint of~$I_e$, and
hence each such pair contributes at least one to~$k_h$. This implies $k_M \le k_h$.

It remains to prove the claim. Consider an edge $e\in E_P\setminus E_R$. (The argumentation
for edges in $E_R\setminus E_P$ is symmetric, with the roles of $w$ and $\w$ exchanged.)
As $e$ is not in $E_R$, replacing all intervals $I_g$
for $g\in E\setminus\{e\}$ by their true values yields an instance
that is solved.
This means that for edge $e$ one of the following
cases applies:
\begin{itemize}
\item[(a)] $e$ is known to be in the MST. Then there is a cut $X_e$
containing edge $e$ (namely, the cut between the two vertex sets
obtained from the MST by removing the edge~$e$) such that $e$ is known to
be a minimum weight edge in the cut, i.e., every other edge $g$ in the cut
satisfies $w_g\ge U_e$.
\item[(b)] $e$ is known not to be in the MST. Then there is a cycle $C_e$
in $G$ (namely, the cycle that is closed when $e$ is added to the MST)
such that $e$ is a maximum weight edge in $C_e$, i.e., every other
edge $g$ in the cycle satisfies $w_g\le L_e$.
\end{itemize}
As $e$ is in $E_P$, replacing all intervals $I_g$
for $g\in E\setminus\{e\}$ by their predicted values yields an instance~$\Pi$
that is not solved. Let $T'$ be the minimum spanning tree of
$G'=(V,E\setminus\{e\})$ for~$\Pi$. Let $C'$ be the cycle
closed in $T'$ by adding $e$, and let $f$ be an edge with the largest predicted
value in $C'\setminus\{e\}$. Then there are only
two possibilities for the minimum spanning tree of $G$ for $\Pi$: Either
$T'$ is also a minimum spanning tree of $G$ (if $\w_e\ge \w_f$),
or the minimum spanning tree is $T'\cup\{e\}\setminus\{f\}$.
As knowing whether $e$ is in the minimum spanning tree would
allow us to determine which of the two cases applies, it
must be the case that we cannot determine whether $e$
is in the minimum spanning tree or not without querying~$e$.
If $e$ satisfied case (a) with cut $X_e$ above, then there must be an edge
$g$ in $X_e\setminus\{e\}$ with $\w_g<U_e$, because otherwise
$e$ would also have to be in the MST of $G$ for $\Pi$, a contradiction.
Thus, the value of $I_g$ passes over $U_e$.
If $e$ satisfied case (b) with cycle $C_e$ above, then there must be an edge
$g$ in $C_e\setminus\{e\}$ with $\w_g>L_e$, because otherwise
$e$ would also be excluded from the MST of $G$ for $\Pi$, a contradiction.
Thus, the value of $I_g$ passes over $L_e$.
This establishes the claim,
and hence we have shown that $k_M\le k_h$ for the MST problem.
\end{proof}  

\ThmLBMandQueryDist*
\begin{proof}
We first establish the following auxiliary claim, which is slightly weaker than
the statement of the theorem:

\begin{claim}
\label{claim:LBM}%
Let $\gamma'\ge 2$ be a fixed rational number. Every deterministic algorithm for the minimum,
sorting or MST problem has competitive ratio at least $\gamma'$ for $k_M=0$
or has competitive ratio at least $1+\frac{1}{\gamma'-1}$ for arbitrary $k_M$.
\end{claim}

Let $\gamma' = \frac{a}{b}$, with integers $a \geq 2b > 0$.
We first give the proof of the claim for the minimum problem.
Consider an instance with $a$ intervals as depicted in Figure~\ref{fig_lb_sym_diff_input}, with sets $S_i = \{I_i, I_{b+1}, I_{b+2}, \ldots, I_{a}\}$ for $i = 1, \ldots, b$.
Suppose, w.l.o.g., that the algorithm queries the left-side intervals in the order $I_1, I_2, \ldots, I_b$, and the right side in the order $I_{b+1}, I_{b+2}, \ldots, I_{a}$.
Let the predictions be correct for $I_{b+1}, \ldots, I_{a-1}$, and $w_{1}, \ldots, w_{b-1} \notin I_1 \cup \ldots \cup I_a$.

If the algorithm queries~$I_a$ before~$I_{b}$, then the adversary sets $w_a \in I_{b}$ and $w_{b} \notin I_a$.
(See Figure~\ref{fig_lb_sym_diff_afirst}.)
This forces a query in all left-side intervals, so the algorithm queries all $a$ intervals, while the optimum solution queries only the~$b$ left-side intervals.
Thus the competitive ratio is at least $\frac{a}{b} = \gamma$ for arbitrary~$k_M$.

If the algorithm queries~$I_{b}$ before~$I_a$, then the adversary sets $w_a = \pred{w}_a$ and $w_{b} \in I_a$; see Figure~\ref{fig_lb_sym_diff_bfirst}.
This forces the algorithm to query all remaining right-side intervals, i.e., $a$ queries in total, while the optimum queries only the $a-b$ right-side intervals. 
Note, however, that $k_M = 0$, since the right-side intervals
are mandatory for both predicted and correct values, while $I_{1}, \ldots, I_{b}$ are not mandatory in either of the solutions.
Thus, the competitive ratio is at least $\frac{a}{a-b} = 1 + \frac{1}{\gamma-1}$ for $k_M = 0$.

For the sorting problem, the proof is the same except that we take the $b(a-b)$ input sets $\{I_i,I_j\}$
for $1\le i\le b$, $b+1\le j\le a$.

For the MST problem, we again translate the above construction from the minimum problem to the
maximum problem and then consider an MST instance consisting of a path with edge weight intervals
$I_j$ for $b+1\le j\le a$, and parallel edges with weight intervals $I_i$ for $1\le i\le b$
(where $I_i$ and $I_j$, with slight abuse of notation, refer to the intervals of the maximum
instance) between the endpoints of the path, see Figure~\ref{fig_lb_mst}.
This completes the proof of Claim~\ref{claim:LBM}.

Now we are ready to prove the theorem. The argument is the same for all three problems.
Let $\gamma\ge 2$ be a fixed rational. Assume that there is a deterministic algorithm $A$
that is $\gamma$-robust and has competitive ratio strictly smaller than
$1 + \frac{1}{\gamma-1}$, say $1+\frac{1}{\gamma+\varepsilon-1}$ with $\varepsilon>0$,
for $k_M = 0$. Let $\gamma'$ be a rational number with $\gamma < \gamma' <\gamma+\varepsilon$.
Then $A$ has competitive ratio strictly smaller than $\gamma'$ for arbitrary $k_M$
and competitive ratio strictly smaller than $1+\frac{1}{\gamma'-1}$ for $k_M=0$,
a contradiction to Claim~\ref{claim:LBM}. This shows the first statement of the theorem.

Let $\gamma\ge 2$ again be a fixed rational. Assumume that there is a deterministic algorithm $A$
that has competitive ratio $1 + \frac{1}{\gamma-1}$ for $k_M=0$ and is
$(\gamma-\varepsilon)$-robust, where $\varepsilon>0$.
As there is a lower bound of $2$ on the robustness of any deterministic algorithm
for all three problems, no such algorithm can exist for $\gamma=2$. So we only
need to consider the case $\gamma>2$ and $\gamma-\varepsilon\ge 2$. Let $\gamma'$ be a rational number
with $\gamma-\varepsilon<\gamma'<\gamma$. Then $A$ has competitive ratio strictly smaller than
$1 + \frac{1}{\gamma'-1}$ for $k_M=0$ and competitive ratio strictly smaller than
$\gamma'$ for arbitrary $k_M$, a contradiction to Claim~\ref{claim:LBM}. This shows
the second statement of the theorem.
\end{proof}

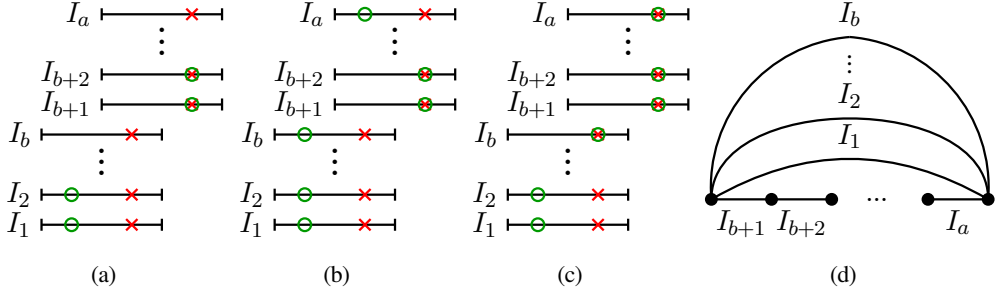
\begin{figure}[tb]
  \centering
  \subfigure[]{\label{fig_lb_sym_diff_input}
  \begin{tikzpicture}[line width = 0.3mm, scale=0.8]
    \intervalpr{$I_1$}{0}{2}{0}{1.5}{0.5}
    \intervalpr{$I_2$}{0}{2}{0.5}{1.5}{0.5}
    \intervalp{$I_b$}{0}{2}{1.5}{1.5}
    \path (1, 0.5) -- (1, 1.6) node[font=\LARGE, midway, sloped]{$\dots$};

    \intervalpr{$I_{b+1}$}{1}{3}{2}{2.5}{2.5}
    \intervalpr{$I_{b+2}$}{1}{3}{2.5}{2.5}{2.5}
    \intervalp{$I_{a}$}{1}{3}{3.5}{2.5}
    \path (2, 2.5) -- (2, 3.6) node[font=\LARGE, midway, sloped]{$\dots$};
  \end{tikzpicture}}%
  \subfigure[]{\label{fig_lb_sym_diff_afirst}
  \begin{tikzpicture}[line width = 0.3mm, scale=0.8]
    \intervalpr{$I_1$}{0}{2}{0}{1.5}{0.5}
    \intervalpr{$I_2$}{0}{2}{0.5}{1.5}{0.5}
    \intervalpr{$I_b$}{0}{2}{1.5}{1.5}{0.5}
    \path (1, 0.5) -- (1, 1.6) node[font=\LARGE, midway, sloped]{$\dots$};

    \intervalpr{$I_{b+1}$}{1}{3}{2}{2.5}{2.5}
    \intervalpr{$I_{b+2}$}{1}{3}{2.5}{2.5}{2.5}
    \intervalpr{$I_{a}$}{1}{3}{3.5}{2.5}{1.5}
    \path (2, 2.5) -- (2, 3.6) node[font=\LARGE, midway, sloped]{$\dots$};
  \end{tikzpicture}}%
  \subfigure[]{\label{fig_lb_sym_diff_bfirst}
  \begin{tikzpicture}[line width = 0.3mm, scale=0.8]
    \intervalpr{$I_1$}{0}{2}{0}{1.5}{0.5}
    \intervalpr{$I_2$}{0}{2}{0.5}{1.5}{0.5}
    \intervalpr{$I_b$}{0}{2}{1.5}{1.5}{1.5}
    \path (1, 0.5) -- (1, 1.6) node[font=\LARGE, midway, sloped]{$\dots$};

    \intervalpr{$I_{b+1}$}{1}{3}{2}{2.5}{2.5}
    \intervalpr{$I_{b+2}$}{1}{3}{2.5}{2.5}{2.5}
    \intervalpr{$I_{a}$}{1}{3}{3.5}{2.5}{2.5}
    \path (2, 2.5) -- (2, 3.6) node[font=\LARGE, midway, sloped]{$\dots$};
  \end{tikzpicture}}  
  \subfigure[]{\label{fig_lb_mst}
  \centering
   \begin{tikzpicture}[thick, scale=0.8, line width = 0.3mm]
    \draw (0, 0) node[vertex]{} -- node[midway, below]{$I_{b+1}$} (1, 0);
    \draw (1, 0) node[vertex]{} -- node[midway, below]{$I_{b+2}$} (2, 0) node[vertex]{};
    \path (2, 0) -- (3.5, 0) node[font=\normalsize, midway]{$\dots$};
    \draw (3.6, 0) node[vertex]{} -- node[midway, below]{$I_{a}$} (4.6, 0) node[vertex]{};
    
    \draw plot (0, 0) to [bend left=30] node[midway, above]{$I_{1}$} (4.6, 0);
    \draw plot (0, 0) to [bend left=90] node[midway, above]{$I_{2}$} (4.6, 0);
    \draw plot (0, 0) to [bend left=45] (2.3, 2.7) node[above]{$I_{b}$};
    \draw plot (2.3, 2.7) to [bend left=45] (4.6, 0);
    \path (2.3, 1.8) -- (2.3, 2.7) node[font=\large, midway, sloped]{$\dots$};
   \end{tikzpicture}}
  \caption{Instance for lower bound based on the mandatory query distance and MST instance.}
  \label{fig_lb_sym_difference}
\end{figure}
  
\subsection{Lower bound for all error measures}

\ThmLBAllErrorMeasures*

\begin{proof}
	Consider the input instance of the minimum or sorting problem consisting of a single set $\{I_1,I_2\}$ as
	shown in Figure~\ref{fig_lb_error_measure}.
If the algorithm starts querying~$I_1$, then the adversary sets $w_1 = \pred{w}_1$ and the algorithm is forced to query~$I_2$. Then $w_2 \in I_2 \setminus I_1$, so the optimum queries only~$I_2$.
It is easy to see that $k_{\#} = k_M = k_h = 1$.
A symmetric argument holds if the algorithm starts querying~$I_2$. In that case, $w_2 = \pred{w}_2$ which enforces to query $I_1$ with $w_1 \in I_1 \setminus I_2$. Taking multiple copies of this instance gives the  result for any $k \leq \opt$.
For the MST problem, we can place copies of the instance (translated to the maximum
problem as usual) in disjoint cycles that are connected by a tree structure.
\end{proof}
\subsection{Learnability of predictions}
\label{sec:learnability}
In this section, we argue about the learnability of our predictions with regard to the different error measures for a given instance of one of the considered problems with the set of uncertainty intervals $\mathcal{I}$.

We assume that the realization $w$ of true values for $\mathcal{I}$ is i.i.d.~drawn from an unknown distribution $\ud$, and
that we can i.i.d.~sample realizations from $\ud$
to obtain a training set.
Let $\hs$ denote the set of all possible
prediction vectors $\pred{w}$, 
with $\pred{w_i} \in I_i$ for each $I_i \in \mathcal{I}$.
Let $k_h(w,\pred{w})$ denote the hop distance of the prediction $\pred{w}$ for the realization with the real values $w$.
Since $w$ is drawn from $\ud$, the value $k_h(w,\pred{w})$ is a random variable. 
Analogously, we consider $k_M(w,\pred{w})$ with regard to the mandatory query distance.
Our goal is to learn predictions $\pred{w}$ that (approximately) minimize the expected error $\EX_{w \sim \ud}[k_h(w,\pred{w})]$ respectively $\EX_{w \sim \ud}[k_M(w,\pred{w})]$.
In the following, we argue separately about the learnability with respect to $k_h$ and $k_M$.

\subsubsection{Learning with respect to the hop distance}

As a main result of this section, we show the case $k=k_h$  of the following theorem.

\TheoLearningHop*

Since each $I_i$ is an open interval, there are infinitely many predictions $\pred{w}$, and, thus, the set $\hs$ is also infinite.
In order to reduce the size of $\hs$, we discretize each $I_i$ by fixing a finite number of potentially predicted values $\pred{w}_i$ of $I_i$.
We define the set $\hs_i$  of
predicted values for $I_i$ as follows.
Let $\{B_1,\ldots,B_l\}$ be the set of lower and upper limits of intervals in $\mathcal{I} \setminus I_i$ that are contained in $I_i$.
Assume that $B_1,\ldots,B_l$ are indexed by increasing value. 
Let $B_0 = L_i$ and $B_{l+1} = U_i$ and, for each $j \in \{0,\ldots,l\}$, let $h_j$ be an arbitrary value of $(B_j,B_{j+1})$.
We define $\hs_i = \{B_1,\ldots,B_l,h_0,\ldots,h_l\}$. 
Since two values $\pred{w}_i,\pred{w}_i' \in (B_j,B_{j+1})$ always lead to the same hop distance for interval $I_i$, there will always be an element of $\hs_i$ that minimizes the expected hop distance for $I_i$.
As $k_h(w,\pred{w})$ is just the sum of the hop distances over all $I_i$, and the hop distances of two intervals $I_i$ and $I_j$ with $i \not= j$ are independent, restricting $\hs$ to the set $\hs_1 \times \hs_2 \times \ldots \times \hs_n$ does not affect the accuracy of our predictions.
Each $\hs_i$ contains at most $\mathcal{O}(n)$ values, and, thus, the discretization reduces the size of $\hs$ to at most $\mathcal{O}(n^n)$.
In particular, $\hs$ is now finite.

To efficiently learn predictions that satisfy Theorem~\ref{theo_learnability_hop}, we again exploit that the hop distances of two intervals $I_i$ and $I_j$ with $i \not= j$ are independent. This is, because the hop distance of $I_i$ only depends on the predicted value $\pred{w}_i$ and the true value $w_i$, but is independent of all $\pred{w}_j$ and $w_j$ with $j\not= i$.
Let $h_i(w_i,\pred{w}_i)$ denote the hop distance of interval $I_i$ for the predicted value $\pred{w}_i$ and the true value $w_i$, and,
for each $i \in \{1,\ldots,n\}$, let $\pred{w}^*_i$ denote the predicted value that minimizes $\EX_{w\sim \ud}[h_i(w_i,\pred{w}_i)]$.
Since the hop distances of the single intervals are independent, the vector $\pred{w}^*$ then minimizes the expected hop distance of the complete instance.
Thus, if we can approximate the individual $\pred{w}^*_i$, then we can show Theorem~\ref{theo_learnability_hop}.

\begin{lemma}
	\label{lemma_learnability_hop}
	For any~$\eps, \delta \in (0,1)$, and any $i \in \{1,\ldots,n\}$, there exists a learning algorithm that, using a training set of size~$$m  \in \mathcal{O}\left(\frac{(\log(n) - \log(\delta))\cdot (2n)^2}{\eps^2}\right),$$ returns a predicted value $\pred{w}_i \in \hs_i$ in time polynomial in~$n$ and~$m$, such that
	$\EX_{w \sim \ud}[h_i(w_i,\pred{w}_i)] \le \EX_{w \sim \ud}[h_i(w_i,\pred{w}_i^*)] + \eps$ holds with probability at least~$(1-\delta)$, where~$\pred{w}_i^* = \arg\min_{\pred{w}_i \in \hs} \EX_{w \sim \ud}[h_i(w_i,\pred{w}_i)]$.
\end{lemma}

\begin{proof}
	We show that the basic \emph{empirical risk minimization (ERM)} algorithm already satisfies the lemma. ERM first i.i.d.~samples a trainingset~$S=\{w^1,\ldots,w^m\}$ of~$m$ true value vectors from~$\ud$.
	Then, it returns the~$\pred{w}_i \in \hs_i$ that minimizes the \emph{empirical error}~$h_S(\pred{w}_i) = \frac{1}{m} \sum_{j=1}^{m} h_i(w^j_i,\pred{w}_i)$.
	
	Recall that, as a consequence of the discretization,~$\hs_i$ contains at most $\mathcal{O}(n)$ values.
	Since~$\hs_i$ is finite, and the error function $h_i$ is bounded by the interval $[0,2n]$, it satisfies the \emph{uniform convergence property}; cf.~\cite{Shalev2014}. (This follows also from the fact that~$\hs_i$ is finite and, thus, has finite VC-dimension; cf.~\cite{Vapnik1992}.) 
	This implies that, for~$$m = \left\lceil \frac{2\log(2|\hs_i|/\delta)(2n)^2}{\eps^2} \right\rceil \in \mathcal{O}\left(\frac{(\log(n) - \log(\delta))\cdot (2n)^2}{\eps^2}\right),$$ it holds~$\EX_{w \sim \ud}[h_i(w_i,\pred{w}_i)] \le \EX_{w \sim \ud}[h_i(w_i,\pred{w}_i^*)] + \eps$ with probability at least~$(1-\delta)$, where~$\pred{w}_i$ is the predicted value learned by ERM (cf.~\cite{Shalev2014,vapnik1999}). 
	As $|\hs_i| \in \mathcal{O}(n)$, ERM also satisfies the running time requirements of the lemma.
\end{proof}

\begin{proof}[Proof of Theorem~\ref{theo_learnability_hop} (Case $k=k_h$)]
	Let $\eps' = \frac{\eps}{n}$ and $\delta' = \frac{\delta}{n}$.
	Furthermore, let $\hs_{\max} = \arg\max_{\hs' \in \{\hs_1,\ldots,\hs_n\}} |\hs'|$.
	To learn predictions that satisfy the theorem, we first sample a training set $S=\{w^1,\ldots,w^m\}$ with $m = \left\lceil \frac{2\log(2|\hs_{\max}|/\delta')(2n)^2}{\eps'^2} \right\rceil$.
	Next, we apply Lemma~\ref{lemma_learnability_hop} to each $\hs_i$ to learn a predicted value $\pred{w}_i$ that satisfies the guarantees of the lemma for $\eps',\delta'$.
	In each application of the lemma, we use the \emph{same} training set~$S$ that was previously sampled.
	
	For each $\pred{w}_i$ learned by applying the lemma, the probability that the guarantee of the lemma is \emph{not} satisfied is less than $\delta'$.
	By the union bound this implies that the probability that at least one $\pred{w}_i$ with $i \in \{1,\ldots,n\}$ does not satisfy the guarantee is upper bounded by $\sum_{i=1}^n \delta' \le n \cdot \delta' = \delta$. 
	Thus, with probability at least $(1-\delta)$, all $\pred{w}_i$ satisfy $\EX_{w \sim \ud}[h_i(w_i,\pred{w}_i)] \le \EX_{w \sim \ud}[h_i(w_i,\pred{w}^*_i)] + \eps'$.
	Since by linearity of expectations $\EX_{w\sim \ud}[k_h(w,\pred{w}^*)] = \sum_{i=1}^n \EX_{w \sim \ud}[h_i(w_i,w^*_i)]$, we can conclude that the following inequality, where $\pred{w}$ is the vector of the learned predicted values, holds with probability at least $(1-\delta)$, which implies the theorem:
	\begin{align*}
	\EX_{w\sim \ud}[k_h(w,\pred{w})] &= \sum_{i=1}^n \EX_{w \sim \ud}[h_i(w_i,\pred{w}_i)]\\
	&\le \sum_{i=1}^n \EX_{w \sim \ud}[h_i(w_i,\pred{w}^*_i)] + \eps'\\
	&\le \left(\sum_{i=1}^n \EX_{w \sim \ud}[h_i(w_i,\pred{w}	^*_i)]\right) + n \cdot \eps'\\
	&\le \EX_{w\sim \ud}[k_h(w,\pred{w}^*)] + \eps.
	\end{align*}
\end{proof}

\subsubsection{Learning with respect to the mandatory query distance}

In this section we discuss the learnability of predictions with respect to $k_M$.
Recall that $k_M$ is defined as $k_M = |\mathcal{I}_P \sym \mathcal{I}_R|$, where $\mathcal{I}_P$ is the set of predictions mandatory elements and $\mathcal{I}_R$ is the set of mandatory elements.
In contrast to learning predictions $\pred{w}$ w.r.t.~$k_h$, an interval $I_i \in \mathcal{I}$ being part of $\mathcal{I}_P \sym \mathcal{I}_R$ depends not only on $\pred{w}_i$ and $w_i$, but on the predicted and true values of $\mathcal{I} \setminus \{I_i\}$. 
Thus, the events of $I_i$ and $I_j$ with $i \not= j$ being part of $\mathcal{I}_P \sym \mathcal{I}_R$ are not necessarily independent.
Therefore, we cannot separately learn the predicted values $\pred{w}_i$ for each $I_i \in \mathcal{I}$.

However, we can still use the discretized $\hs$ as described in the previous section without losing any precision, as shown in the following lemma.
Here, for any vector $\pred{w}$ of predicted values, $\mathcal{I}_{\pred{w}}$ denotes the set of prediction mandatory intervals.
Since we only have algorithms with a guarantee depending on $k_M$ for the sorting and minimum problems, we show the following lemma only for those problems but remark that a similar proof is possible for the MST problem.

\begin{lemma}
	\label{lem_discretization}
	For a given instance of the sorting or minimum problem with intervals $\mathcal{I}$, let $\pred{w}$ be a vector of predicted values that is not contained in the discretized $\hs$.
	Then, there is a $\pred{w}' \in \hs$ such that $\mathcal{I}_{\pred{w}} = \mathcal{I}_{\pred{w}'}$.
\end{lemma}

The lemma implies that, for each $\pred{w}$, there is an $\pred{w}' \in \hs$ that has the same error w.r.t.~$k_M$ as $\pred{w}$.
Thus, there always exists an element $\pred{w}$ of the discretized $\hs$ such that $\pred{w}$ minimizes the expected error over all possible vectors of predicted values. 

\begin{proof}[Proof of Lemma~\ref{lem_discretization}]
	Since the sorting problem can be reduced to the minimum problem (cf. Section~\ref{sec:sorting}), it suffices to show the statement for the minimum problem.
	
	Given the vector of predicted values $\pred{w}$, we construct a vector $\pred{w}' \in \hs$.
	Recall that $\hs = \hs_1 \times \ldots \times \hs_n$ with $\hs_i = \{B_1,\ldots,B_l,h_0,\ldots,h_l\}$ where $B_1,\ldots,B_l$ are the interval borders that are contained in $I_i$ of intervals in $\mathcal{I}\setminus\{I_i\}$ ordered by non-decreasing value, and $h_j$ is an arbitrary value in $(B_j,B_{h+1})$ with $B_0=L_i$ and $B_{l+1}=U_i$. 
	For each $\pred{w}_i$, we construct $\pred{w}'_i$ as follows:
	If $\pred{w}_i = B_j$ for some $j \in \{1,\ldots,l\}$, then we set $\pred{w}_i' = B_j$.
	Otherwise it must hold $\pred{w}_i \in (B_j,B_{j+1})$ for some $j \in \{1,\ldots,l\}$, and we set $\pred{w}'_i = h_j$.
	For each $I_j \in \mathcal{I}$ it holds that $\pred{w}_i \in I_j$ if and only if $\pred{w}_i' \in I_j$.
	
	We show that each $I_i \in \mathcal{I}_{\pred{w}}$ is also contained in $\mathcal{I}_{\pred{w}'}$.
	Since $I_i$ is mandatory assuming true values $\pred{w}$, Lemma~\ref{lema_mandatory_min} implies that there is a set $S$ such that either (i) $\pred{w}_i$ is a true minimum of $S$ and $\pred{w}_j \in I_i$ for some $S_j \in S\setminus\{I_i\}$ or (ii) $\pred{w}_i$ is not a true minimum of $S$ but contains the value $\pred{w}_j$ of a true minimum $I_j$ of set $S$.
	
	Assume $I_i$ satisfies case (i) for the predicted values $\pred{w}$.
	By construction of $\pred{w}'$ it then also holds $\pred{w}'_j \in I_i$.
	Thus, if $I_i$ is the true minimum for the values $\pred{w}'$, then $I_i \in \mathcal{I}_{\pred{w}'}$.
	Otherwise, some $I_j' \in S\setminus\{I_i\}$ must be the true minimum in $S$ for values $\pred{w}'$.
	Since $I_i$ is a true minimum for values $\pred{w}$, it must hold $\pred{w}'_{j'} < \pred{w}'_i$ but $\pred{w}_{j'} \ge \pred{w}_i$.
	By construction, this can only be the case if $\pred{w}'_j,\pred{w}_j \in I_i$.
	This implies that $I_i$ satisfies case (ii) for the values $\pred{w}'$ and, therefore $I_i \in \mathcal{I}_{\pred{w}'}$.

	Assume $I_i$ satisfies case (ii) for the predicted values $\pred{w}$.
	By construction of $\pred{w}'$ it then also holds $\pred{w}'_j \in I_i$.
	Thus, if $I_j$ is the true minimum for the values $\pred{w}'$, then $I_i \in \mathcal{I}_{\pred{w}'}$.
	Otherwise, some $I_{j'} \in S\setminus\{I_j\}$ must be the true minimum in $S$ for values $\pred{w}'$.
	If $j' = i$, then $I_i$ satisfies case (i) for the values $\pred{w}'$ and, therefore, $I_i \in \mathcal{I}_{\pred{w}'}$.
	Otherwise, as $I_j$ is a true minimum for values $\pred{w}$, it must hold $\pred{w}'_{j'} < \pred{w}'_j$ but $\pred{w}_{j'} \ge \pred{w}_j$.
	By construction and since $\pred{w}_j \in I_i$, this can only be the case if $\pred{w}'_{j'},\pred{w}_{j'} \in I_i$.
	This implies that $I_i$ satisfies case (ii) for the values $\pred{w}'$ and, therefore $I_i \in \mathcal{I}_{\pred{w}'}$.

	Symmetrically, we can show that each $I_i \in \mathcal{I}_{\pred{w}'}$ is also contained in $\mathcal{I}_{\pred{w}}$.
	This implies $\mathcal{I}_{\pred{w}} = \mathcal{I}_{\pred{w}'}$.
\end{proof}

Since the discretized $\hs$ is still finite and $k_M$ is bounded by $[0,n]$, we still can apply ERM to achieve guarantees similar to the ones of Theorem~\ref{theo_learnability_hop}.
However, because we cannot learn the $\pred{w}_i$ separately, we would have to find the element of $\hs$ that minimizes the empirical error. 
As $\hs$ is of exponential size, a straightforward implementation of ERM requires exponential running time.
We use this straightforward implementation to proof the case $k=k_M$ for Theorem~\ref{theo_learnability_hop}.

\begin{proof}[Proof of Theorem~\ref{theo_learnability_hop} (Case $k=k_M$)]
	Since $|\hs| \in \mathcal{O}(n^n)$ and $k_M$ is bounded by $[0,n]$, ERM achieves the guarantee of Theorem~\ref{theo_learnability_hop} with a sample complexity of $m \in \mathcal{O}\left(\frac{(n \cdot \log(n) - \log(\delta))\cdot (n)^2}{(\eps)^2}\right)$.
	The prediction $\overline{w} \in \hs$ that minimizes the empirical error $k_S(\pred{w}_i) = \frac{1}{m} \sum_{j=1}^{m} k_M(w^j_i,\pred{w}_i)$, where $S=\{w^1,\ldots,w^m\}$ is the training set, can be computed by iterating through all elements of $S \times H$.
	Thus, the running time of ERM is polynomial in $m$ but exponential in $n$.
\end{proof}

To circumvent the exponential running time, we present an alternative approach.
In contrast to $k_h$, for a fixed realization, the value $k_M$ only depends on $\mathcal{I}_P$.
Instead of showing the learnability of the predicted values, we prove that the set $\mathcal{I}_P$ that leads to the smallest expected error can be (approximately) learned.
To be more specific, let $\mathcal{P}$ be the power set of $\mathcal{I}$, let $\mathcal{I}_w$ denote the set of mandatory elements for the realization with true values $w$, and let $k_M(\mathcal{I}_w, P)$ with $P \in \mathcal{P}$ denote the mandatory query distance under the assumption that $\mathcal{I}_P = P$ and $\mathcal{I}_R = \mathcal{I}_w$.
Since $w$ is drawn from $\ud$, the value $\EX_{w \sim \ud}[k_M(\mathcal{I}_w,P)]$ is a random variable.
As main result of this section, we show the following theorem.

\TheoLearningMan*

Note that this theorem only allows us to learn a set of prediction mandatory intervals $P$ that (approximately) minimizes the expected mandatory query distance. 
It does not, however, allow us to learn the predicted values $\pred{w}$ that lead to the set of prediction mandatory intervals $P$.
In particular, it can be the case, that no such predicted values exist. 
Thus, there may not be a realization with true values $w$ and $k_M(\mathcal{I}_w,P)=0$.
On the other hand, learning $P$ already allows us to execute Algorithm~\ref{ALG_min_alpha} for the minimum problem.
Applying Algorithm~\ref{fig:sorting1cons2rob} for the sorting problem would require knowing the corresponding predicted values $\pred{w}$.

\begin{proof}[Proof of Theorem~\ref{theo_learnability_mandatory}]
	We again show that the basic \emph{empirical risk minimization (ERM)} algorithm already satisfies the lemma. 
	ERM first i.i.d.~samples a trainingset~$S=\{w^1,\ldots,w^m\}$ of~$m$ true value vectors from~$\ud$.
	Then, it returns the~$P \in \mathcal{P}$ that minimizes the \emph{empirical error}~$k_S(P) = \frac{1}{m} \sum_{j=1}^{m} k_M(\mathcal{I}_{w^j},P)$.
	In contrast to the proof of Lemma~\ref{lemma_learnability_hop}, since $\mathcal{P}$ is of exponential size, i.e., $|\mathcal{P}| \in \mathcal{O}(2^n)$, we cannot afford to naively iterate through $\mathcal{P}$ in the second stage of ERM, but have to be more careful.
	
	By definition, $\mathcal{P}$ contains $\mathcal{O}(2^n)$ elements and, thus, is finite.
	Since~$\mathcal{P}$ is finite, and the error function $k_M$ is bounded by the interval $[0,n]$, it satisfies the \emph{uniform convergence property} (cf.~\cite{Shalev2014}).
	This implies that, for~$$m = \left\lceil \frac{2\log(2|\mathcal{P}|/\delta)n^2}{\eps^2} \right\rceil \in \mathcal{O}\left(\frac{(n - \log(\delta))\cdot n^2}{\eps^2}\right),$$ it holds~$\EX_{w \sim \ud}[k_M(\mathcal{I}_{w},P)] \le \EX_{w \sim \ud}[k_M(\mathcal{I}_w,P^*)] + \eps$ with probability at least~$(1-\delta)$, where~$P$ is the set $P\in\mathcal{P}$ learned by ERM (cf.~\cite{Shalev2014,vapnik1999}).

	It remains to show that we can compute the set $P \in \mathcal{P}$ that minimizes the empirical error~$k_S(P) = \frac{1}{m} \sum_{j=1}^{m} k_M(\mathcal{I}_{w^j},P)$ in time polynomial in $n$ and $m$.
	For each $I_i$, let $p_i = |\{\mathcal{I}_{w^j} \mid 1 \le j \le m \land I_i \in  \mathcal{I}_{w^j}\}|$ and let $q_i = m - p_i$.
  	For an arbitrary $P \in \mathcal{P}$, we can rewrite $k_S(P)$ as follows:
  	\begin{align*}
  		k_S(P) &= \frac{1}{m} \sum_{j=1}^{m} k_M(\mathcal{I}_{w^j},P)
  		= \frac{1}{m} \sum_{j=1}^{m} |\mathcal{I}_{w^j} \sym P|\\
  		&= \frac{1}{m} \sum_{j=1}^m |P\setminus \mathcal{I}_{w^j}| + |\mathcal{I}_{w^j} \setminus P|\\
  		&= \frac{1}{m} \left(\sum_{I_i \in P} q_i + \sum_{I_i\not\in P} p_i \right).
  	\end{align*}
  	A set $P \in \mathcal{P}$ minimizes the term $k_S(P) = \frac{1}{m} (\sum_{I_i \in P} q_i + \sum_{I_i\not\in P} p_i)$, if and only if, $q_i \le p_i$ holds for each $I_i \in P$.
  	Thus, we can compute the $P \in \mathcal{P}$ that minimizes $k_S(P)$ as follows:
  	\begin{enumerate}
  		\item Compute $q_i$ and $p_i$ for each $I_i \in \mathcal{I}$.
  		\item Return $P = \{I_i \in \mathcal{I} \mid q_i \le p_i\}$.
  	\end{enumerate}
  	Since this algorithm can be executed in time polynomial in $n$ and $m$, the theorem follows.
\end{proof}

To conclude the section, we show the following lemma that compares the quality of the hypothesis sets $\hs$ and $\mathcal{P}$ w.r.t.~$k_M$.
It shows that, with respect to the expected $k_M$, the best prediction of $\mathcal{P}$ dominates the best prediction of $\hs$.

\begin{lemma}
	For any instance of one of our problems and any distribution $D$, it holds $\EX_{w \sim \ud}[k_M(w,\pred{w}^*) \ge \EX_{w \sim \ud}[k_M(\mathcal{I}_w,P^*)]$, where
	$\pred{w}^* = \arg\min_{\pred{w}' \in \hs} \EX_{w \sim \ud}[k_M(w,\pred{w}')]$ and let $P^* =  \arg\min_{P' \in \mathcal{P}} \EX_{w \sim \ud}[k_M(\mathcal{I}_w,P')]$.
	Furthermore, there exists an instance and a distribution $D$ such that $\EX_{w \sim \ud}[k_M(w,\pred{w}^*) > \EX_{w \sim \ud}[k_M(\mathcal{I}_w,P^*)]$.
\end{lemma}

\begin{proof}
	Let $\pred{w}^*$ and $P^*$ be as described in the lemma, and let $I_{\pred{w}^*}$ denote the set of intervals that are mandatory in the realization with true values $\pred{w}^*$.
	By definition, $\EX_{w \sim \ud}[k_M(w,\pred{w}^*) = \EX_{w \sim \ud}[k_M(\mathcal{I}_w,\mathcal{I}_{\pred{w}^*})$.
	Since $\mathcal{I}_{\pred{w}^*} \in \mathcal{P}$ and $P^* =  \arg\min_{P' \in \mathcal{P}} \EX_{w \sim \ud}[k_M(\mathcal{I}_w,P')]$, it follows $\EX_{w \sim \ud}[k_M(w,\pred{w}^*) = \EX_{w \sim \ud}[k_M(\mathcal{I}_w,\mathcal{I}_{\pred{w}^*})  \ge \EX_{w \sim \ud}[k_M(\mathcal{I}_w,P^*)]$.

	To show the second part of the lemma, consider the example of Figure~\ref{ex_domination}.
	We assume that the true values for the different intervals are drawn independently.
	In this example, we have $P^* = \{I_1,I_3,I_5\}$.
	Since there is no combination of true values such that all of $P^*$ are mandatory, the in expectation best prediction of the true values $\pred{w}^*$ has either $\mathcal{I}_{\pred{w}^*} = \{I_1,I_3\}$ or $\mathcal{I}_{\pred{w}^*} = \{I_3,I_5\}$.
	Thus, $\EX_{w \sim \ud}[k_M(w,\pred{w}^*) = 1.2601  > 1.2401 = \EX_{w \sim \ud}[k_M(\mathcal{I}_w,P^*)]$.
	\end{proof}
	
	\begin{figure}[t]
			\centering
			\begin{tikzpicture}[line width = 0.3mm, scale = 0.9, transform shape]
			\interval{$I_1$}{0}{3}{0}{1}{1}		
			\interval{$I_2$}{2}{5}{1}{3.25}{5.5}		
			\interval{$I_3$}{4}{7}{0}{3.25}{5.5}					
			\interval{$I_4$}{6}{9}{1}{3.25}{5.5}	
			\interval{$I_5$}{8}{11}{0}{3.25}{5.5}	
			
			 \drawreal{1}{0}
			 \node[black!40!green] (1) at (1,0.5) {$1$};
			 
			 \drawreal{5.5}{0}
			 \node[black!40!green] (1) at (5.5,0.5) {$1$};
			 
			 \drawreal{10}{0}
			 \node[black!40!green] (1) at (10,0.5) {$1$};

			\drawreal{2.5}{1}
			\node[black!40!green] (1) at (2.5,1.5) {$0.51$};
			
			\drawreal{4.5}{1}
			\node[black!40!green] (1) at (4.5,1.5) {$0.49$};

			\drawreal{6.5}{1}
			\node[black!40!green] (1) at (6.5,1.5) {$0.49$};

			\drawreal{8.5}{1}
			\node[black!40!green] (1) at (8.5,1.5) {$0.51$};			
		\end{tikzpicture}
		\caption{Instance of the sorting problem for a single set with a given distribution $D$ for the true values. 
		Can also be interpreted as an instance for the minimum problem with sets $\{I_i,I_{i+1}\}$ for $i\in\{1,\ldots,4\}$.
		The circles illustrate the  possible true values for the different intervals according to $D$. The associated numbers denote the probability that the interval has the corresponding true value.}
	\label{ex_domination}
	\end{figure}
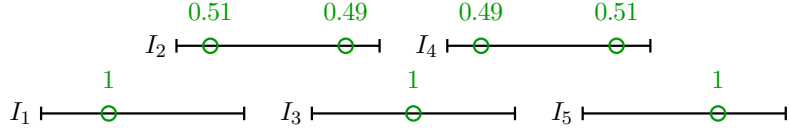

\section{NP-hardness of the Verification Problem for Minimum and Sorting}
\label{sec:nphard}
\begin{theorem}
\label{thm_min_nphard}
It is NP-hard to solve the verification problem for finding the minimum in overlapping sets with uncertainty intervals and given values.
\end{theorem}
\begin{proof}
The proof uses a reduction from the vertex cover problem for 2-subdivision graphs, which is NP-hard~\cite{poljak74subdiv}.
A 2-subdivision is a graph~$H$ which can be obtained from an arbitrary graph~$G$ by replacing each edge by a path of length four (with three edges and two new vertices).
The graph in Figure~\ref{fig:2_subdiv} is a 2-subdivision of the graph in Figure~\ref{fig:base_graph}.

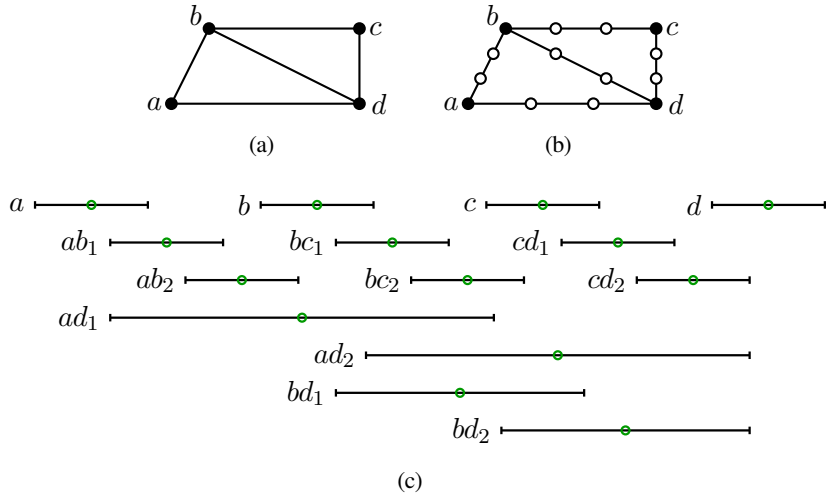
\begin{figure}[!ht]
 \centering
  \subfigure[]{\label{fig:base_graph}
   \tikzstyle{every node}=[circle, draw, fill=black, inner sep=0pt, minimum width=4pt]
   \begin{tikzpicture}[thick, scale=0.5]
    \draw (0, 1) node[label=west:$a$]{} -- (1, 3) node[label=north west:$b$]{};
    \draw (1, 3) -- (5, 3) node[label=east:$c$]{};
    \draw (5, 3) -- (5, 1) node[label=east:$d$]{};
    \draw (1, 3) -- (5, 1);
    \draw (0, 1) -- (5, 1);
   \end{tikzpicture}
  }\quad
  \subfigure[]{\label{fig:2_subdiv}
   \tikzstyle{every node}=[circle, draw, inner sep=0pt, minimum width=4pt]
   \begin{tikzpicture}[thick, scale=0.5]
    \draw (0, 1) node[fill=black, label=west:$a$]{} -- (1, 3) node[fill=black, label=north west:$b$]{};
    \draw (0.33, 1.67) node[fill=white]{};
    \draw (0.67, 2.33) node[fill=white]{};
    \draw (1, 3) -- (5, 3) node[fill=black, label=east:$c$]{};
    \draw (2.33, 3) node[fill=white]{};
    \draw (3.67, 3) node[fill=white]{};
    \draw (5, 3) -- (5, 1) node[fill=black, label=east:$d$]{};
    \draw (5, 1.67) node[fill=white]{};
    \draw (5, 2.33) node[fill=white]{};
    \draw (1, 3) -- (5, 1);
    \draw (2.33, 2.33) node[fill=white]{};
    \draw (3.67, 1.67) node[fill=white]{};
    \draw (0, 1) -- (5, 1);
    \draw (1.67, 1) node[fill=white]{};
    \draw (3.33, 1) node[fill=white]{};
   \end{tikzpicture}
  }\quad
  \subfigure[]{\label{fig:minimum_reduction}
   \begin{tikzpicture}[line width = 0.3mm, scale=0.5]
     \intervalr{$a$}{0}{3}{0}{1.5}
     \intervalr{$b$}{6}{9}{0}{7.5}
     \intervalr{$c$}{12}{15}{0}{13.5}
     \intervalr{$d$}{18}{21}{0}{19.5}
     
     \intervalr{$ab_1$}{2}{5}{-1}{3.5}
     \intervalr{$ab_2$}{4}{7}{-2}{5.5}
     
     \intervalr{$bc_1$}{8}{11}{-1}{9.5}
     \intervalr{$bc_2$}{10}{13}{-2}{11.5}
     
     \intervalr{$cd_1$}{14}{17}{-1}{15.5}
     \intervalr{$cd_2$}{16}{19}{-2}{17.5}

     \intervalr{$ad_1$}{2}{12.2}{-3}{7.1}
     \intervalr{$ad_2$}{8.8}{19}{-4}{13.9}

     \intervalr{$bd_1$}{8}{14.6}{-5}{11.3}
     \intervalr{$bd_2$}{12.4}{19}{-6}{15.7}
   \end{tikzpicture}
  }
  \caption{NP-hardness reduction for the minimum problem, from the vertex cover problem on 2-subdivision graphs. 
  \subref{fig:base_graph} A graph and \subref{fig:2_subdiv} its 2-subdivision.
  \subref{fig:minimum_reduction} The corresponding instance for the minimum problem.}
  \label{fig_min_np_hard}
\end{figure}

Given a graph~$H$ which is a 2-subdivision of a graph~$G$, we build an instance of the minimum problem in the following way.
If~$n$ is the number of vertices in~$G$, then we start by creating $n$ intervals that do not intersect each other.
For each edge~$uv$ of~$G$, we create two new intervals~$uv_1$ and~$uv_2$, in such a way that we have intersection pairs $u(uv_1), (uv_1)(uv_2), (uv_2)v$.
We then create a set of size~2 for each edge in~$H$, consisting of the corresponding intervals.
Finally, we can clearly assign true values such that, for any set $xy$, neither $w_x \in I_y$ nor $w_y \in I_x$.
Therefore, to solve each set, it is enough to query one of the intervals, so clearly any solution to the problem corresponds to a vertex cover of~$H$, and vice-versa.
See Figure~\ref{fig:minimum_reduction} for an example.
\end{proof}

This reduction constructs an instance of the minimum problem where
all sets have size~2. For instances with sets of size~2, the minimum problem and the sorting problem are
equivalent. Hence, Theorem~\ref{thm_min_nphard} implies the following.

\begin{coro}
It is NP-hard to solve the verification problem for sorting overlapping sets with uncertainty intervals and given values.
\end{coro}

\section{Appendix for the Minimum Problem (Section \ref{sec:minimum})}
\label{app:minimum}

\begin{restatable}{lem}{LemmaMinFirstImplMand}
\label{lemma_min_first_implies_mandatory}
Let~$I_l$ be the leftmost interval in a set~$S$, and assume that~$I_l$ does not contain another interval in~$S$.
If~$I_l$ is queried, then~$S$ can be solved by querying only intervals that become known mandatory.
\end{restatable}

\begin{proof}
Since~$I_l$ does not contain another interval in~$S$, either~$w_l$ is found to be the minimum in~$S$, or $w_l$ is contained in the next leftmost interval in~$S$, which becomes a known mandatory interval due to Corollary~\ref{cor_min_left_mandatory}, and then the claim follows by induction.
\end{proof}

The following lemma guarantees that Algorithms~\ref{ALG_min_beta} and~\ref{ALG_min_alpha} indeed solve the problem.
In Algorithm~\ref{ALG_min_beta}, the instance in Line~\ref{line_min_beta_kh_vc} has no known mandatory intervals because of Line~\ref{line_min_beta_kh_before_vc}.
In Algorithm~\ref{ALG_min_alpha}, the instance in Line~\ref{line_min_param_vc} has no known mandatory intervals because of Line~\ref{line_min_param_before_vc}.

\begin{restatable}{lem}{LemmaMinQueryVCEnough}
\label{lemma_min_query_vc_enough}
Suppose an instance without known mandatory intervals.
After querying a vertex cover, the minimum problem can be solved by querying only intervals that become known mandatory.
\end{restatable}

\begin{proof}
The vertex cover is computed in the dependency graph, a graph with a vertex for each interval
and, for each unsolved set $S$, edges between the leftmost interval~$I_l$ in $S$ and all other
intervals in $S$ that intersect $I_l$.
Consider an unsolved set~$S$ with leftmost interval~$I_l$.
The vertex cover must either (a) contain~$I_l$ (and maybe other intervals in $S$), or (b) contain all intervals in $S \setminus \{I_l\}$ that intersect $I_l$.
If case~(b) holds, then either~$I_l$ becomes known to be the minimum interval because no true value is in~$I_l$, or some other true value is in~$I_l$ and the claim follows from Corollary~\ref{cor_min_left_mandatory}.
If case~(a) holds, then the claim follows from Lemma~\ref{lemma_min_first_implies_mandatory}.
\end{proof}

\subsection{Analysis of Algorithm~\ref{ALG_min_beta} (Theorem \ref{thm_minimum:hop})}
\label{app:minbeta}

Recall the following lemma for identifying witness sets.

\begin{lem}[\cite{kahan91queries}]\label{lemma_witness_min}
A set $\{I_i, I_j\} \subseteq S$ with $I_i \cap I_j \neq \emptyset$, and $I_i$ or~$I_j$ leftmost in~$S$, is always a witness set.
\end{lem}

An important lemma for proving upper bounds for this algorithm is the following.

\begin{lem} \label{lemma_min_kh_enforce_mandatory}
If $\pred{w}_j$ enforces~$I_i$, then $\{I_i, I_j\}$ is a witness set.
Also, if $w_j \in I_i$, then $I_i$ is mandatory.
\end{lem}

\begin{proof}
Since $\pred{w}_j$ enforces~$I_i$, there must be a set $S$ with $I_i,I_j \in S$ such that $\pred{w}_j \in I_i$ and either $I_i$ is leftmost in $S$ or $I_j$ is leftmost in $S$ and $I_i$ is leftmost in $S \setminus \{I_j\}$.
The first claim follows from Lemma~\ref{lemma_witness_min}.
If~$I_i$ is a true minimum of~$S$ or is leftmost in~$S$, then the second claim follows from Lemma~\ref{lema_mandatory_min} and Corollary~\ref{cor_min_left_mandatory}.
Otherwise, the fact that $w_j \in I_i$ and that~$I_i$ is leftmost in $S \setminus \{I_j\}$ implies that~$I_i$ contains the minimum true value, so the claim follows from Lemma~\ref{lema_mandatory_min}.
\end{proof}

\begin{lem}
\label{lemma_min_beta_kh_enforce_predmand}
Consider a point of execution of the algorithm in which~$\pred{w}_j$ enforces~$I_i$.
It holds that~$I_i$ is prediction mandatory for the current instance.
\end{lem}

\begin{proof}
Consider the set~$S$ as in the definition of~$\pred{w}_j$ enforcing~$I_i$.
If~$I_i$ is leftmost in~$S$, then the claim holds from a similar argument as in the proof of Corollary~\ref{cor_min_left_mandatory}.
If~$I_j$ is leftmost in~$S$ and~$I_i$ is leftmost in $S \setminus \{I_j\}$, then we have three cases, in all of which the claim follows from Lemma~\ref{lema_mandatory_min}: (a)~$\pred{w}_j$~is the predicted minimum value in~$S$, and $\pred{w}_j \in I_i$ holds by definition; (b)~$\pred{w}_i$~is the predicted minimum value in~$S$, and $\pred{w}_j \in I_i$ holds by definition; 
(c) some $\pred{w}_l$ with $l \neq i, j$ is the predicted minimum value in~$S$, and in this case we claim that $\pred{w}_l \in I_i$.
To see that the claim for the last case holds, note that $L_i \leq L_l$ because $I_j$ is leftmost in~$S$ and $I_i$ is leftmost in $S \setminus \{I_j\}$, and $\pred{w}_l < U_i$ because $\pred{w}_l$ is the predicted minimum.
\end{proof}

\begin{lem}
\label{lemma_min_beta_kh_unique_loop}
There is at most one execution of the loop consisting of Lines~1--\ref{line_min_beta_kh_cond_loop} in which Line~\ref{line_min_beta_kh_predict} is executed but no query is performed in Lines~\ref{line_min_beta_wit_trio}--\ref{line_min_beta_query_first}, and that is the last execution of this loop in which any query is performed.
After this point, the instance has no prediction mandatory intervals.
\end{lem}

\begin{proof}
Suppose that there is an iteration in which Line~\ref{line_min_beta_kh_predict} is executed but no query is performed in Lines~\ref{line_min_beta_wit_trio}--\ref{line_min_beta_query_first}.
We claim that, before the test in Line~\ref{lin_min_beta_cond_trio} is executed, the current instance has no prediction mandatory intervals.
Suppose, for the sake of contradiction, that a set~$S$ satisfies some condition in Lemma~\ref{lema_mandatory_min} to contain a prediction mandatory interval.
Let~$I_i$ be leftmost in~$S$, and $I_j$ be leftmost in $S \setminus \{I_i\}$.
We claim that~$\pred{w}_i$ enforces~$I_j$ or some $\pred{w}_l$ enforces~$I_i$ with $I_l \in S$, which contradicts the fact that no query is performed in Lines~\ref{line_min_beta_wit_trio}--\ref{line_min_beta_query_first}.
Let $\pred{w}^*$ be the predicted minimum value in~$S$.
If~$\pred{w}_i \neq w^*$, then $\pred{w}^* \in I_i$ because~$I_i$ is leftmost, so~$\pred{w}^*$ enforces~$I_i$.
Otherwise $\pred{w}_i = w^*$ and we have two cases: (a) if $\pred{w}_i > L_j$, then~$\pred{w}_i \in I_j$ because Line~\ref{line_min_beta_kh_mandatory_loop} prevents that $I_j \subseteq I_i$, so~$\pred{w}_i$ enforces~$I_j$; (b) if $\pred{w}_i \leq L_j$, then some interval $I_l \in S$ must have $\pred{w}_l \in I_i$, otherwise there would be no prediction mandatory intervals in~$S$, so $\pred{w}_l$ enforces $I_i$.

This claim implies that no query is performed in the next iteration because:
\begin{enumerate}
 \item No query can be performed in Line~\ref{line_min_beta_kh_mandatory_first} because known mandatory intervals are queried in the last execution of Line~\ref{line_min_beta_kh_mandatory_loop}.
 \item From the previous item, the instance does not change and remains without prediction manda\-to\-ry intervals, so no query is performed in Lines~\ref{line_min_beta_kh_predict} and~\ref{line_min_beta_kh_mandatory_loop}.
 \item No query can be made in Lines~\ref{line_min_beta_wit_trio}--\ref{line_min_beta_query_first} because, if $\pred{w}_j$ enforces~$I_i$, then Lemma~\ref{lemma_min_beta_kh_enforce_predmand} implies that~$I_i$ is prediction mandatory, which contradicts the previous item.
\end{enumerate}

Therefore the instance after this point has no prediction mandatory intervals and the lemma holds.
\end{proof}

\begin{lem}
\label{lemma_min_no_witness_implies_solved}
Let $I_i, I_j$ be a pair that satisfies the condition in Line~\ref{line_min_beta_pair} leading to a query of~$I_i$.
Every set~$S$ containing~$I_j$ will be solved after querying~$I_i$, or will be solved using only known mandatory queries in Line~\ref{line_min_beta_kh_mandatory_first} in the next consecutive iterations of the loop, or~$I_j$ is certainly not the minimum in~$S$.
\end{lem}

\begin{proof}
Consider the instance before~$I_i$ is queried.
Due to the test in Line~\ref{lin_min_beta_cond_trio}, for every set~$S$ containing~$I_j$, the following facts hold:
\begin{enumerate}[(1)]
 \item If~$I_j$ is leftmost in~$S$, then $S$ is already solved, or $I_i \in S$ and $I_i$ is the only interval in~$S$ that intersects~$I_j$.
 \item If~$I_j$ is not leftmost in~$S$ but intersects the leftmost interval~$I_{i'}$ in~$S$, then $I_{i'} = I_i$.
 \item If~$I_j$ is not leftmost in~$S$ and does not intersect the leftmost interval in~$S$, then~$I_j$ is certainly not the minimum in~$S$.
\end{enumerate}
If condition~(1) holds and~$S$ is not solved then, after querying~$I_i$, either $w_i \notin I_j$ and~$S$ becomes solved, or $w_i \in I_j$ and $I_j$ will be queried in Line~\ref{line_min_beta_kh_mandatory_first} due to Corollary~\ref{cor_min_left_mandatory}, and then~$S$ becomes solved.
If condition~(2) holds, then the result follows from Lemma~\ref{lemma_min_first_implies_mandatory} because~$I_i$ is leftmost in~$S$.
The result follows trivially if condition~(3) holds.
\end{proof}

This lemma clearly implies the following corollary.

\begin{coro}
\label{cor_min_beta_no_witness_anymore}
Let $I_i, I_j$ as in Lemma~\ref{lemma_min_no_witness_implies_solved}.
After~$I_i$ is queried, $I_j$~will no longer be a prediction mandatory interval identified in Line~\ref{line_min_beta_kh_calc_predict_first} or \ref{line_min_beta_kh_calc_predict_second}, or be part of a triple satisfying the condition in Line~\ref{lin_min_beta_cond_trio}, or a pair satisfying the condition in Line~\ref{line_min_beta_pair}, or be part of the vertex cover~$Q$ queried in Line~\ref{line_min_beta_kh_vc}, or be a known mandatory interval queried in Line~\ref{line_min_beta_kh_mandatory_loop},~\ref{line_min_beta_kh_before_vc} or~\ref{line_min_beta_kh_after_vc}.
\end{coro}

\thmminimumhop*

\begin{proof}
{\bf $\gamma$-robustness.}
Intervals queried in Line~\ref{line_min_beta_mand_trio} are in any feasible solution.

Fix an optimum solution $\OPT$.
Let $\mathcal{I}'$ be the set of unqueried intervals in Line~\ref{line_min_beta_kh_calc_predict_first} at the iteration of the loop consisting of Lines~1--\ref{line_min_beta_kh_cond_loop} in which Line~\ref{line_min_beta_kh_predict} is executed but no query is performed in Lines~\ref{line_min_beta_wit_trio}--\ref{line_min_beta_pair} (or before Line~\ref{line_min_beta_kh_before_vc} if no such iteration exists).
Recall that Lemma~\ref{lemma_min_beta_kh_unique_loop} states that there is at most one such iteration, and it has to be the last iteration in which some interval is queried.
If the problem is undecided at this point, then $|\OPT \cap \mathcal{I}'| \geq 1$, and $|Q| \leq \gamma - 2$ implies $|Q| \leq (\gamma - 2) \cdot |\OPT \cap \mathcal{I}'|$.
Also, since~$Q'$ is a minimum vertex cover, then $|Q'| \leq |\OPT \cap \mathcal{I}'|$.
Let~$M$ be the set of intervals in~$\mathcal{I}'$ that are queried in Lines~\ref{line_min_beta_kh_mandatory_loop},~\ref{line_min_beta_kh_before_vc} and~\ref{line_min_beta_kh_after_vc}; clearly $M \subseteq \OPT \cap \mathcal{I}'$.
Thus $|Q| + |Q'| + |M| \leq \gamma \cdot |\OPT \cap \mathcal{I}'|$.

Now consider an iteration of the loop in which some query is performed in Lines~\ref{line_min_beta_wit_trio}--\ref{line_min_beta_pair}.
Let~$P'$ be the set of intervals queried in Lines~\ref{line_min_beta_kh_predict},~\ref{line_min_beta_wit_trio} and~\ref{line_min_beta_query_first}.
If Line~\ref{line_min_beta_wit_trio} is executed, then note that $\{I_j, I_l\}$ is a witness set.
If a query is performed in Line~\ref{line_min_beta_query_first}, then note that $\{I_i, I_j\}$ is a witness set.
Due to Lemma~\ref{lemma_min_no_witness_implies_solved}, if~$I_j$ is queried, then it is in Line~\ref{line_min_beta_kh_mandatory_first} at the next iteration. 
Independent of $I_j$ being queried, if a query is performed in Line~\ref{line_min_beta_query_first}, we include~$I_j$ in~$P'$ for the sake of this analysis.
Due to Corollary~\ref{cor_min_beta_no_witness_anymore}, $I_j$~is not considered more than once in this case, and is not considered in any of the previous cases.
Either way, it holds that $P'$ is a witness set of size at most~$\gamma$.

The remaining intervals queried in Lines~\ref{line_min_beta_kh_mandatory_first} and~\ref{line_min_beta_kh_mandatory_loop} are in any feasible solution.

{\bf Bound of $(1 + \frac{1}{\gamma})(1 + \frac{k_h}{\opt})$.}
Fix an optimum solution $\OPT$.
Let $h'(I_j)$ be the number of intervals $I_i$ such that $I_i, I_j \in S$ for some $S \in \mathcal{S}$, and the value of $I_i$ passes over an endpoint of $I_j$.
From the arguments in the proof of Theorem~\ref{Theo_hop_distance_mandatory_distance}, it can be seen that, for each interval~$I_j$ that is prediction mandatory at some point and is not in $\OPT$, we have that $h'(I_j) \geq 1$.
For a subset $\mathcal{J} \subseteq \mathcal{I}$, let $h'(\mathcal{J}) = \sum_{I_j \in \mathcal{J}} h'(I_j)$.
Note that $k_h = h'(\mathcal{I})$ holds by reordering summations.

In the following, we will show for various disjoint subsets $\mathcal{J}\subseteq \mathcal{I}$ that $|\mathcal{J} \cap \ALG| \le (1+\frac{1}{\gamma})\cdot (|\OPT \cap \mathcal{J}| + h'(\mathcal{J}))$.
The subsets $\mathcal{J}$ will form a partition of $\mathcal{I}$, so it is clear that the bound of $(1+\frac{1}{\gamma}) \cdot (1 + \frac{k_h}{\opt})$ on the competitive ratio of the algorithm follows.
Furthermore, if $\gamma = 2$, then we will show for every $\mathcal{J}$ that $|\mathcal{J} \cap \ALG| \le 1.5 \cdot |\OPT \cap \mathcal{J}| + h'(\mathcal{J})$, so we have a bound of $1.5 + k_h / \opt$ on the competitive ratio.

Intervals queried in Lines~\ref{line_min_beta_kh_mandatory_loop} and~\ref{line_min_beta_kh_before_vc} are in any feasible solution, so the set~$P_0$ of these intervals satisfies $|P_0| \leq |\OPT \cap P_0|$.

If there is an execution of the loop consisting of Lines~1--\ref{line_min_beta_kh_cond_loop} that does not perform queries in Lines~\ref{line_min_beta_wit_trio}--\ref{line_min_beta_pair}, then let~$P_1$ be the set of intervals queried in Line~\ref{line_min_beta_kh_predict}.
Every interval~$I_j \in P_1$ is prediction mandatory, so if $I_j \notin \OPT$ then $h'(I_j) \geq 1$.
Thus we have that $|P_1| \leq |P_1 \cap \OPT| + h'(P_1)$.

Let $\mathcal{I}'$ be the set of unqueried intervals before Line~\ref{line_min_beta_kh_vc} is executed.
Since $Q'$ is a minimum vertex cover, we have that $|Q'| \leq |\OPT \cap \mathcal{I}'|$.
Let~$M$ be the set of intervals queried in Line~\ref{line_min_beta_kh_after_vc}.
Due to Lemma~\ref{lema_mandatory_min}, each interval $I_i \in M$ is known mandatory because it contains the value~$w_j$ of an interval $I_j \in Q$.
But Lemma~\ref{lemma_min_beta_kh_unique_loop} implies that $\pred{w}_j \notin I_i$ when Line~\ref{line_min_beta_kh_vc} was executed, so $h'(I_i) \geq 1$.
Thus we have that $|\mathcal{I}' \cap \ALG| = |Q' \cup M| \leq |\mathcal{I}' \cap \OPT| + h'(M) \leq |\mathcal{I}' \cap \OPT| + h'(\mathcal{I}')$.

Consider an execution of the loop in which some query is performed in Lines~\ref{line_min_beta_wit_trio}--\ref{line_min_beta_pair}.
Let~$Q$ be the set of intervals queried in Line~\ref{line_min_beta_kh_predict}, and let~$W$ be the set of intervals queried in Lines~\ref{line_min_beta_wit_trio}--\ref{line_min_beta_query_first}.
If a query is performed in Line~\ref{line_min_beta_query_first} and $I_j$ is queried in Line~\ref{line_min_beta_kh_mandatory_first} at the next iteration, then include~$I_j$ in~$W$ as well.
Note that $|Q| = \gamma - 2$. If~$\pred{w}_j$ enforces~$I_i$ in Line~\ref{lin_min_beta_cond_trio} or~\ref{line_min_beta_pair}, then~$I_i$ is prediction mandatory due to Lemma~\ref{lemma_min_beta_kh_enforce_predmand}.
Also, note that $h'(Q) \geq |Q \setminus \OPT|$, since every interval in~$Q$ is prediction mandatory at some point.
We divide the proof in three cases.
For a pair $\{I_i, I_j\}$ as in Line~\ref{line_min_beta_pair}, note that, due to Corollary~\ref{cor_min_beta_no_witness_anymore}, $I_j$~is not considered more than once, and is not considered in any of the previous cases.

\begin{enumerate}[(a)]
 \item If $|W| = 1$, then some interval~$I_i$ was queried in Line~\ref{line_min_beta_query_first} because $\pred{w}_j$ enforces $I_i$, and $I_j$ is not queried by the algorithm due to Lemma~\ref{lemma_min_no_witness_implies_solved}.
 Then it suffices to note that $\{I_i, I_j\}$ is a witness set to see that $|Q \cup W| \leq |\OPT \cap (Q \cup \{I_i, I_j\})| + h'(Q)$.
 \item Consider $|W| = 2$.
 If $W$ is a pair of the form $\{I_j, I_l\}$ queried in Line~\ref{line_min_beta_wit_trio}, then $h'(I_i) \geq 1$ because~$\pred{w}_j$ enforces~$I_i$ but $w_j \notin I_i$.
 We can conceptually move this contribution in the hop distance to~$I_j$, making $h'(I_i) := h'(I_i) - 1$ and $h'(I_j) := h'(I_j) + 1$.
 (If $I_i$ is considered another time in Line~\ref{lin_min_beta_cond_trio} or in another point of the analysis because it is enforced by some predicted value, then it has to be the predicted value of an interval $I_{j'} \neq I_j$, so we are not counting the contribution to the hop distance more than once.)
 If $W$ is a pair of the form $\{I_i, I_j\}$ queried in Line~\ref{line_min_beta_query_first} and in Line~\ref{line_min_beta_kh_mandatory_first} at the next iteration, then either $W \subseteq \OPT$ or $h'(I_i) = 1$: It holds that~$I_j$ is mandatory, so if~$I_i$ is not in $\OPT$ then it suffices to see that $\pred{w}_j$ enforces $I_i$.
 Either way, the fact that~$W$ is a witness set is enough to see that $|Q \cup W| \leq |\OPT \cap (Q \cup W)| + h'(Q) + h'(W)$.
 \item If $|W| = 3$, then $W = \{I_i, I_j, I_l\}$ as in Line~\ref{lin_min_beta_cond_trio}, and $|Q \cup W| = \gamma + 1$.
 Also, it holds that~$I_i$ and at least one of $\{I_j, I_l\}$ are in any feasible solution.
 This implies that at least $\frac{\gamma}{\gamma+1} \cdot |Q \cup W| - h'(Q)$ of the intervals in $Q \cup W$ are in $\OPT$, so $|Q \cup W| \leq (1 + \frac{1}{\gamma})(|\OPT \cap (Q \cup W)| + h'(Q))$.
 If $\gamma = 2$, then $Q = \emptyset$ and we have that $|W| \leq 1.5 \cdot |\OPT \cap W|$.
\end{enumerate}

The remaining intervals queried in Line~\ref{line_min_beta_kh_mandatory_first} are in any feasible solution.
\end{proof}

\subsection{Analysis of Algorithm~\ref{ALG_min_alpha} (Theorem \ref{thm:min-alpha})}
\label{app:min_alpha}

\begin{restatable}{lem}{LemmaMinNoPredMandAfterVC}
\label{lemma_min_no_pred_mand_after_vc}
Every interval queried in Line~\ref{line_min_param_mand_after_vc} of Algorithm~\ref{ALG_min_alpha} is in $\mathcal{I}_R \setminus \mathcal{I}_P$.
\end{restatable}

\begin{proof}
Clearly every such interval is in~$\mathcal{I}_R$ because it is known {to be} mandatory, so it remains to prove that it is not in $\mathcal{I}_P$.
Consider a set~$S$.
If an interval $I_j \in S \cap \mathcal{I}_P$ is unqueried in Line~\ref{line_min_param_vc}, then the condition for identifying a witness set in Line~\ref{line_min_param_cond} implies that, before Line~\ref{line_min_param_vc} is executed, $I_j$~is not leftmost in~$S$ and does not intersect the leftmost interval in~$S$.
Thus,
{it is not necessary
to query~$I_j$ to solve~$S$, and we can conclude that $I_j$ is not queried in Line~\ref{line_min_param_mand_after_vc}.}
\end{proof}

\ThmMinAlpha*

\begin{proof}
{\bf $\gamma$-robustness.}
Given $P' \cup \{b\}$ queried in Line~\ref{line_min_param_big}, at least one interval is in any feasible solution since $\{b, p\}$ is a witness set, thus $P' \cup \{b\}$ is a witness set of size $\gamma$.

Line~\ref{line_min_param_small} is executed at most once, since the size of~$P$ never increases.
Fix an optimum solution $\OPT$, and let~$\mathcal{I}'$ be the set of unqueried intervals before Line~\ref{line_min_param_small} is executed (or before Line~\ref{line_min_param_before_vc} if Line~\ref{line_min_param_small} is never executed).
Let~$P$ be the set of intervals queried in Line~\ref{line_min_param_small}.
If the problem is undecided at this point, then $|\OPT \cap \mathcal{I}'| \geq 1$, so $|P| \leq \gamma - 2$ implies $|P| \leq (\gamma - 2) \cdot |\OPT \cap \mathcal{I}'|$.
Also, since~$Q$ is a minimum vertex cover, then $|Q| \leq |\OPT \cap \mathcal{I}'|$.
Let~$M$ be the set of intervals in~$\mathcal{I}'$ that are queried in Lines~\ref{line_min_param_before_vc} or~\ref{line_min_param_mand_after_vc}; clearly $M \subseteq \OPT \cap \mathcal{I}'$.
Thus $|P| + |Q| + |M| \leq \gamma \cdot |\OPT \cap \mathcal{I}'|$.

The intervals queried in Line~\ref{line_min_param_mandatory} are in any feasible solution, and the claim follows.

\smallskip 
\noindent \textbf{Bound of $(1+\frac{1}{\gamma-1}) \cdot (1 + \frac{k_M}{\opt})$.}
Fix an optimum solution $\OPT$.
In the following, we will show for various disjoint subsets $\mathcal{J}\subseteq \mathcal{I}$ that $|\mathcal{J} \cap \ALG| \le (1+\frac{1}{\gamma-1})\cdot (|\OPT \cap \mathcal{J}| + k_{\mathcal{J}})$, where $k_{\mathcal{J}} \leq |\mathcal{J} \cap (\mathcal{I}_P \sym \mathcal{I}_R)|$.
The subsets $\mathcal{J}$ will form a partition of $\mathcal{I}$, so it is clear that the bound of $(1+\frac{1}{\gamma-1}) \cdot (1 + \frac{k_M}{\opt})$ on the competitive ratio of the algorithm follows.

Intervals queried in Lines~\ref{line_min_param_mandatory} and~\ref{line_min_param_before_vc} are part of any feasible solution, hence the set $P_0$ of these intervals satisfies $|P_0|\le | \OPT \cap P_0|$.

Given $P' \cup \{b\}$ queried in Line~\ref{line_min_param_big}, at least $\frac{\gamma-1}{\gamma}$ of the intervals in $P_{\gamma} \cup \{b\}$ are prediction mandatory for the initial instance.
Among those, let~$k' \leq k_M$ be the number of intervals in $\mathcal{I}_P \setminus \mathcal{I}_R$.
Thus $|\OPT \cap (P' \cup \{b\})| \geq \frac{\gamma-1}{\gamma} \cdot |P' \cup \{b\}| - k'$, which gives the desired bound, i.e., $|P' \cup \{b\}| \le (1+\frac{1}{\gamma-1})\cdot (|\OPT \cap (P' \cup \{b\})| + k')$.

Every interval queried in Line~\ref{line_min_param_small} that is not in $\OPT$ is in $\mathcal{I}_P \setminus \mathcal{I}_R$.
Hence, if there are $k''$ such intervals, then the set $P$ of intervals queried in Line~\ref{line_min_param_small} satisfies $|P| \le |\OPT \cap P| + k'' < (1+\frac{1}{\gamma})\cdot (|\OPT \cap P| + k'')$.

Let~$\mathcal{I}'$ be the set of unqueried intervals before Line~\ref{line_min_param_vc} is executed.
Then $|Q| \leq |\OPT \cap \mathcal{I}'|$ because~$Q$ is a minimum vertex cover.
Let~$M$ be the set of intervals that are queried in Line~\ref{line_min_param_mand_after_vc}.
It holds that $|Q \cup M| \leq |\OPT \cap \mathcal{I}'| + |M|$, so the claimed bound follows from Lemma~\ref{lemma_min_no_pred_mand_after_vc}.
\end{proof}

The parameter $\gamma$ in Theorem \ref{thm:min-alpha} is restricted to integral values since it determines sizes of query sets. Nevertheless, a generalization to arbitrary $\gamma \in \RR_{+}$ is possible at a small loss in the guarantee. We give the following rigorous upper bound on the achievable tradeoff of robustness and error-dependent competitive ratio.

\begin{restatable}{theorem}{ThmMinFractionalGamma}
	\label{thm:min-problem-arbitrary-gamma}
	For any real number $\gamma \geq 2$, there is a randomized algorithm for the minimum and sorting problem under uncertainty that achieves a competitive ratio of $\min\{(1+\frac{1}{\gamma-1}+\xi)\cdot(1+\frac{k_M}{\opt}), \gamma\}$, for $\xi \leq \frac{\gamma - \lfloor \gamma \rfloor}{(\gamma-1)^2} \leq1.$
\end{restatable}

\begin{proof}
	\newcommand{\E}[1]{\mathbb{E}\left[\,#1\,\right]}
For $\gamma \in \ZZ$, we run Algorithm \ref{ALG_min_alpha} and achieve the performance guarantee from Theorem \ref{thm:min-alpha}. 
Assume $\gamma \notin\ZZ$, and let $\{\gamma\}:=\gamma - \lfloor \gamma \rfloor = \gamma - \lceil \gamma \rceil +1$ denote its fractional part. We run the following randomized variant of Algorithm \ref{ALG_min_alpha}. We randomly chose $\gamma'$ as $\lceil \gamma \rceil$  with probability $\{\gamma\}$ and as $\lfloor \gamma \rfloor$ with probability $1-\{\gamma\}$, and then we run the algorithm with $\gamma'$ instead of $\gamma$. We show that the guarantee from Theorem \ref{thm:min-alpha} holds in expectation with an additive term less than $\{\gamma\}$, more precisely, we show the competitive ratio
\[
\min\left\{\left( 1+\frac{1}{\gamma-1} +\xi \right) \cdot \left(1+\frac{k_M}{\opt} \right), \gamma\right\}, \text{ for } \xi = \frac{ \{\gamma\}(1-\{\gamma\}) }{(\gamma -1) \lfloor \gamma \rfloor (\lfloor \gamma \rfloor - 1)} \leq \frac{\{\gamma\}}{(\gamma-1)^2} .
\]

Following the arguments in the proof of Theorem \ref{thm:min-alpha} on the robustness, the ratio of the algorithm's number of queries $|\ALG|$ and $|\OPT|$ is bounded by $\gamma'$. In expectation the robustness is
\begin{align*}
	\E{\gamma'} &= (1-\{\gamma\}) \cdot \lfloor \gamma \rfloor + \{\gamma\} \cdot \lceil \gamma \rceil \\
	&= (1-\{\gamma\}) \cdot (\gamma - \{\gamma\}) + \{\gamma\} \cdot (\gamma - \{\gamma\} +1) \\
	&= \gamma.
\end{align*}

The error-depending bound on the competitive ratio is in expectation (with $\opt$ and $k_M$ being independent of~$\gamma$)
\begin{align*}
	\E{ \left(1+\frac{1}{\gamma'-1}\right) \cdot \left( 1+ \frac{k_M}{\opt} \right) } = \left(1+\E{\frac{1}{\gamma'-1}}\right)\cdot \left( 1+ \frac{k_M}{\opt} \right).
\end{align*}

Applying simple algebraic transformations, we obtain
\begin{align*}
	\E{\frac{1}{\gamma'-1}} &= \frac{1-\{\gamma\}}{\lfloor \gamma \rfloor-1} + \frac{\{\gamma\}}{\lceil \gamma \rceil-1}\ 
	=\ \frac{1-\{\gamma\}}{\gamma - \{\gamma\}-1} + \frac{\{\gamma\}}{\gamma - \{\gamma\}}\\
	&= \frac{(1-\{\gamma\})(\gamma - \{\gamma\}) + \{\gamma\}(\gamma - \{\gamma\}-1)}{(\gamma - \{\gamma\}-1)(\gamma - \{\gamma\})} \\
	&= \frac{ \gamma -2 \{\gamma\} }{(\gamma - \{\gamma\}-1)(\gamma - \{\gamma\})} \ 
	=\ \frac{1}{\gamma-1} - \frac{1}{\gamma-1} + \frac{ \gamma -2\{\gamma\}}{(\gamma - \{\gamma\}-1)(\gamma - \{\gamma\})} \\
	&=\ \frac{1}{\gamma-1} + \frac{ \{\gamma\}(1-\{\gamma\}) }{(\gamma-1) (\gamma - \{\gamma\}-1)(\gamma - \{\gamma\})} \ = \ \frac{1}{\gamma-1} + \frac{ \{\gamma\}(1-\{\gamma\}) }{(\gamma-1) \lfloor \gamma \rfloor (\lfloor \gamma \rfloor - 1)}. 
\end{align*}

Hence, the competitive ratio is in expectation
\[
\left( 1+\frac{1}{\gamma-1} +\xi \right) \cdot \left(1+\frac{k_M}{\opt} \right) \text{ with } \xi = \frac{ \{\gamma\}(1-\{\gamma\}) }{(\gamma -1) \lfloor \gamma \rfloor (\lfloor \gamma \rfloor -1)} \leq \frac{\{\gamma\}}{(\gamma-1)^2},
\]
which concludes the proof.
\end{proof}

\section{Appendix for the MST Problem (Section \ref{sec:mst})}
\label{appx:mst}
In this section we prove the main theorems of Section~\ref{sec:mst}. 
We first show the MST part of Theorem~\ref{thm:main1.5-2}
and Theorem~\ref{theorem_mst2} assuming that the  lemmas of Section~\ref{sec:mst} hold, and then we prove that these lemmas are indeed true. We first restate the the MST part of Theorem~\ref{thm:main1.5-2}.

\begin{restatable}{theorem}{MstTheoremOne}[partial restatement of Theorem~\ref{thm:main1.5-2}]
	\label{theorem_mst1}
	There is a $1.5$-consistent and $2$-robust algorithm for MST under uncertainty.
\end{restatable}

Consider the algorithm that first executes Algorithm~\ref{ALG_mst_part_1} with $\gamma = 2$ and then Algorithm~\ref{ALG_mst_part_2} using {recovery strategy~A}.

\begin{proof}	
	Let $\ALG = \ALG_1 \cup \ALG_2$ be the query set queried by the algorithm, where $\ALG_1$ and $\ALG_2$ are the queries of Algorithm~\ref{ALG_mst_part_1} and Algorithm~\ref{ALG_mst_part_2}, respectively. Let $\OPT = \OPT_1 \cup \OPT_2$ be an optimal query set with $\OPT_1 = \OPT \cap \ALG_1$ and $\OPT_2 = \OPT \setminus \ALG_1$. 
	Since
	$\gamma = 2$, Line~\ref{line_mst_one_fillup} {of Algorithm~\ref{ALG_mst_part_1}} does not query any elements.
	Therefore Lemma~\ref{mst_end_of_phase_one} implies $|\ALG_1| \le \min\{(1+\frac{1}{2}) \cdot (|\OPT_1| + k_h),2 \cdot |\OPT_1|\}$.
	
	We continue by analyzing Algorithm~\ref{ALG_mst_part_2} using recovery strategy A. 
	{According to Lemma~\ref{mst_end_of_phase_one}, the input instance of Algorithm~\ref{ALG_mst_part_2} is prediction mandatory free and we can apply Lemma~\ref{mst_end_of_phase_2}.
	The lemma implies $|\ALG_2| \le |\OPT_2|$ if all predictions are correct and $|\ALG_2| \le 2 \cdot |\OPT_2|$ otherwise.} 

	Summing up,	this implies $1.5$-consistency and $2$-robustness.
\end{proof}

\MstTheoremTwo*

Consider the algorithm that first executes Algorithm~\ref{ALG_mst_part_1} with some $\gamma\in \ZZ, \gamma \ge 2$ and then Algorithm~\ref{ALG_mst_part_2} using recovery strategy B.

\begin{proof}
	Let $\ALG = \ALG_1 \cup P \cup \ALG_2$ be the query set queried by the algorithm, where $\ALG_1$ is the set of elements queried by Algorithm~\ref{ALG_mst_part_1} without the elements queried in the last iteration of Line~\ref{line_mst_one_fillup}, $P$ is the set of elements queried in the last iteration {of Line~\ref{line_mst_one_fillup}} and $\ALG_2$ is the set {of} elements queried by Algorithm~\ref{ALG_mst_part_2}. 
	Let $\OPT = \OPT_1 \cup \OPT_2$ be an optimal query set with $\OPT_1 = \OPT \cap \ALG_1$ and $\OPT_2 = \OPT \setminus \ALG_1$. 
	Lemma~\ref{mst_end_of_phase_one} implies $\ALG_1 \le \min\{(1+\frac{1}{\gamma}) \cdot (|\OPT_1| + k_h),\gamma \cdot |\OPT_1|\}$. 
	
	We continue by analyzing $\ALG_2$ and $P$ and first show $|\ALG_2\cup P| \le \max\{3 \cdot |\OPT_2|,\gamma \cdot |\OPT_2| + 1\}$. 
	{According to Lemma~\ref{mst_end_of_phase_one}, the input instance of Algorithm~\ref{ALG_mst_part_2} is prediction mandatory free and we can apply Lemma~\ref{mst_end_of_phase_2}.
	The lemma implies $|\ALG_2| \le 3 \cdot |\OPT_2 \setminus P|$.}
		
	For $\gamma = 2$, it holds $P=\emptyset$ by definition of the algorithm. 
	Thus, $|\ALG_2| \le 3 \cdot |\OPT_2 \setminus P| = 3 \cdot |\OPT_2|$.
	Since $|\ALG_1| \le \gamma \cdot |\OPT_1| = 2 \cdot |\OPT_1|$, it follows $|\ALG_1 \cup P \cup \ALG_2| \le 3  \cdot |\OPT_1 \cup \OPT_2| = 3 \cdot |\OPT|$.
	For $\gamma \ge 3$, observe that, if $P \not= \emptyset$, then $|\OPT_2| \ge 1$.
	This is because if the optimal query set was empty, the instance is solved at this point and
	no prediction mandatory elements can exist after querying $\ALG_1$, which contradicts $P \not= \emptyset$.
	This implies $|P| \le \gamma-2 \le (\gamma - 3) \cdot |\OPT_2| + 1$ and $|\ALG_2 \cup P| \le \gamma \cdot |\OPT_2| + 1$.
	Combining the arguments for $\gamma = 2$ and $\gamma \ge 3$, it follows $|\ALG_2 \cup P| \le \max\{3\cdot|\OPT_2|, \gamma \cdot |\OPT_2| + 1\}$.
	Using $|\ALG_1| \le \gamma \cdot |\OPT_1|$ we can conclude $|\ALG| = |\ALG_1 \cup P \cup \ALG_2| \le \max\{3\cdot |\OPT|,\gamma  \cdot |\OPT| + 1\}$.

	We continue by showing the consistency. 
	Since each element of $P$ is prediction mandatory, {the proof of Theorem~\ref{Theo_hop_distance_mandatory_distance}} implies that each such element is either mandatory or contributes {at least} one to the hop distance.
	It follows that querying $P$ only improves the consistency and at most $k_h$ elements of $P$ are not part of $\OPT_2$.

	Lemma~\ref{mst_end_of_phase_2} implies $|\ALG_2| \le |\OPT_2\setminus P| + {5}\cdot k_h$.
	Summing up the guarantee for $P$ with the guarantee for $\ALG_2$ directly gives us $|\ALG_2 \cup P| \le |\OPT_2| + 6 \cdot k_h$.
	By combining this guarantee with $|\ALG_1| \le (1+\frac{1}{\gamma}) \cdot (|\OPT_1| + k_h)$, we directly obtain $|\ALG| = |\ALG_1 \cup P \cup \ALG_2| \le (1+\frac{1}{\gamma}) \cdot \opt + (7+\frac{1}{\gamma}) \cdot k_h$.
	However, in Corollary~\ref{mst_error_sensitiv_guarantee} we observe that we can exploit disjointness between the errors that we charge against to achieve the guarantees for $\ALG_1$, $\ALG_2$ and $P$, to improve the guarantee for $\ALG$ to $|\ALG| = |\ALG_1 \cup P \cup \ALG_2| \le (1+\frac{1}{\gamma}) \cdot \opt + (5+\frac{1}{\gamma}) \cdot k_h$.
\end{proof}

The proofs of Theorems~\ref{theorem_mst1} and~\ref{theorem_mst2} show that the introduced lemmas imply the theorem. 
The remainder of this section proves that those lemmas indeed hold.

\subsection{Preliminaries}

Before the lemmas of Subsections~\ref{subsec_mst_phase_1} and~\ref{subsec_mst_phase_2} are shown, we introduce some preliminaries that are necessary for the proofs.
We start by showing that we can assume uniqueness of $T_L$ and $T_U$ as well as $T_L = T_U$.

\LemMSTPreprocessing*

\begin{proof}
	Let $T_L$ be a lower limit tree for a given instance $G$ and let $T_U$ be an upper limit tree.
	According to~\cite{megow17mst}, all elements of $T_L \setminus T_U$ are mandatory and we can repeatedly query them for (the adapting) $T_L$ and $T_U$ until $T_L = T_U$. 
	We refer to this process as the \emph{first preprocessing step}.
	
	Consider an $f \in E \setminus T_U$ and the cycle $C$ in $T_U \cup \{f\}$.
	If $f$ is trivial, then the true value $w_f$ is maximal in $C$ and we may delete $f$ without loss of generality.
	Assume otherwise. 
	If the upper limit of $f$ is uniquely maximal in $C$, then $f$ is not part of any upper limit tree.
	If there is an $l \in C$ with $U_f = U_l$, then $T_U' = T_U \setminus \{l\} \cup \{f\}$ is also an upper limit tree.
	Since $T_L \setminus T_U' = \{l\}$, we may execute the first preprocessing step for $T_L$ and $T_U'$.
	We repeatedly do this until each $f \in E \setminus T_U$ is uniquely maximal in the cycle $C$ in $T_U \cup \{f\}$.
	Then, $T_U$ is unique.
	
	To achieve uniqueness for $T_L$, consider some $l \in T_L$ and the cut $X$ of $G$ between the two connected components of $T_L \setminus \{l\}$.
	If $l$ is trivial, then the true value $w_l$ is minimal in $X$ and we may contract $l$ without loss of generality.
	Assume otherwise. 
	If $L_l$ is uniquely minimal in $X$, then $l$ is part of every lower limit tree.
	If there is an $f \in X$ with $L_l = L_f$, then $T_L' = T_L \setminus \{l\} \cup \{f\}$ is also a lower limit tree.
	Since $T_L' \setminus T_U = \{f\}$, we may execute the first preprocessing step for $T_L'$ and $T_U$.
	We repeatedly do this until each $l \in T_L$ is uniquely minimal in the cut $X$ of $G$ between the two components of $T_L \setminus \{l\}$.
	Then, $T_L$ is unique.
\end{proof}

Consider $T_L$ and the edges $f_1,\ldots,f_l$ in $E \setminus T_L$ ordered by lower limit non-decreasingly. 
For each $i \in \{1,\ldots,l\}$, define $C_i$ to be the unique cycle in $T_L \cup \{f_i\}$ and $G_i=(V,E_i)$ to be the sub graph with $E_i = T_L \cup \{f_1,\ldots,f_i\}$. 
Additionally, we define $G_0 = (V,T_L)$. 
During the course of this section, we will make use of the following two lemmas that were shown in~\cite{megow17mst}. 
According to Lemma~\ref{mst_preprocessing}, we can assume $T_L = T_U$ and that $T_L$ and $T_U$ are unique.

\begin{lem}[{\cite[Lemma~5]{megow17mst}}]
	\label{lemma_mst_1}
	Let $i \in \{1,\ldots,l\}$. 
	Given a feasible query set $Q$ for the uncertainty graph $G=(V,E)$, then the set $Q_i := Q \cap E_i$ is a feasible query set for $G_i = (V,E_i)$.
\end{lem}

\begin{lem}[{\cite[Lemma~6]{megow17mst}}]
	\label{lemma_mst_2}
	For some realization of edge weights, let $T_i$ be a verified MST for graph $G_i$ and let $C$ be the cycle closed by adding $f_{i+1}$ to $T_i$.
	Furthermore let $h$ be some edge with the largest upper limit in $C$ and $g \in C \setminus \{h\}$ be an edge with $U_g > L_h$. 
	Then any feasible query set for $G_{i+1}$ contains $h$ or $g$. 
	Moreover, if $I_g$ is contained in $I_h$, any feasible query set contains edge $h$.
\end{lem}

While Lemma~\ref{lemma_mst_2} shows how to identify a witness set on the cycle closed by $f_{i+1}$ after an MST for graph $G_{i}$ is already verified, our algorithms rely on identifying witness sets involving edges $f_{i+1}$ without first verifying an MST for $G_{i}$. 
The remainder of the section derives properties that allow us to identify such witness sets.

\begin{obs}\label{obs_mst_1}
	Let $Q$ be a feasible query set that verifies an MST $T^*$. 
	Consider any path $P \subseteq T_L$ between two endpoints $a$ and $b$, and let $e \in P$ be the edge with the highest upper limit in $P$. 
	If $e \not\in Q$, then the path $\hat{P} \subseteq T^*$ from $a$ to $b$ is such that $e \in \hat{P}$ and $e$ has the highest upper limit in $\hat{P}$ after $Q$ has been queried. 
\end{obs}

\begin{proof}
	For each $i \in \{0,\ldots,l\}$ let $T^*_i$ be the MST for $G_i$ as verified by $Q_i$ and let $\hat{C_i}$ be the unique cycle in $T^*_{i-1} \cup \{f_i\}$. 
	Then $T^*_i = T^*_{i-1} \cup \{f_i\} \setminus \{h_i\}$ holds where $h_i$ is the maximal edge on $\hat{C_i}$.
	Assume $e \not\in Q$. 
	We claim that there cannot be any $\hat{C_i}$ with $e \in \hat{C_i}$ such that $Q_i$ verifies that an edge $e' \in \hat{C_i}$ with $U_{e'} \le U_e$ is maximal in $\hat{C_i}$. 
	Assume otherwise.
	If $e'\not= e$, $Q_i$ would need to verify that $w_{e} \le w_{e'}$ holds. 
	Since $U_{e'} \le U_{e}$, this can only be done by querying $e$, which is a contradiction to $e \not\in Q$.
	If $e'=e$, then $\hat{C}_i$ still contains edge $f_i$.
	Since $T_L = T_U$, $f_i$ has a higher lower limit than $e$. 
	To verify that $e$ is maximal in $\hat{C}_i$, $Q_i$ needs to prove $w_e \ge w_{f_i} > L_{f_i}$. 
	This can only be done by querying $e$, which is a contradiction to $e \not\in Q$.
	
	We show via induction on $i \in \{0,\ldots,l\}$ that each $T^*_i$ contains a path $P^*_i$ from $a$ to $b$ with $e \in P^*_i$ such that $e$ has the highest upper limit in $P^*_i$ after $Q_i$ has been queried. 
	For this proof via induction we define $Q_0 = \emptyset$. 
	\emph{Base case $i=0$}: Since $G_0 = (E,T_L)$ is a spanning tree, $T^*_0 = T_L$ follows. 
	Therefore $P^*_0 = P$ is part of $T_L$ and by assumption $e \in P$ has the highest upper limit in $P^*_0$.
	
	\emph{Inductive step}: By induction hypothesis, there is a path $P^*_i$ from $a$ to $b$ in $T^*_i$ with $e \in P^*_i$ such that $e$ has the highest upper limit in $P^*_i$ after querying $Q_i$.
	Consider cycle $\hat{C}_{i+1}$. 
	If an edge $e' \in \hat{C}_{i+1} \setminus P^*_i$ is maximal in $\hat{C}_{i+1}$, then $T^*_{i+1} = T^*_i \cup \{f_{i+1}\} \setminus \{e'\}$ contains path $P^*_i$. 
	Since $e$ by assumption is not queried, $e$ still has the highest upper limit on $P^*_i = P^*_{i+1}$ after querying $Q_{i+1}$ and the statement follows.

	Assume some $e' \in P^*_i \cap \hat{C}_{i+1}$ is maximal in $\hat{C}_{i+1}$, then $U_{e'} \le U_e$ follows by induction hypothesis since $e$ has the highest upper limit in $P^*_i$. 
	We already observed that $\hat{C}_{i+1}$ then cannot contain $e$. 
	Consider $P' = \hat{C}_{i+1} \setminus P^*_i$. 
	Since $e' \in P^*_i$ is maximal in $\hat{C}_{i+1}$, we can observe that $P' \subseteq T^*_{i+1}$ holds.  
	It follows that path $P^*_{i+1} = P' \cup (P^*_i \setminus\hat{C}_{i+1})$ with $e \in P^*_{i+1}$ is part of $T^*_{i+1}$.
	Since $e$ is not queried, it still has a higher upper limit than all edges in $P^*_i$. 
	Additionally, we can observe that after querying $Q_{i+1}$ no $u \in P'$ can have an upper limit $U_{u} \ge U_{e'}$. 
	If such an $u$ would exist, querying $Q_{i+1}$ would not verify that $e'$ is maximal on $\hat{C}_{i+1}$, which contradicts the assumption. 
	Using $U_e \ge U_{e'}$, we can conclude that $e$ has the highest upper limit on $P^*_{i+1}$ and the statement follows.
\end{proof}

Using this observation, we derive two lemmas that allow us to identify witness sets of size two. 
These lemmas are valuable for our algorithms since they allow us to identify witness sets on a cycle $C_i$ independent of what MST $T_{i-1}$ is verified for graph $G_{i-1}$. 
While most algorithms for MST under uncertainty iteratively resolve the cycles closed by adding the edges $f_1,\ldots,f_l$ (or execute the analogous cut-based algorithm), the following lemmas allow us to query edges using a less local strategy.

\LemMstWitnessFir*

\begin{proof}
	To prove the lemma, we have to show that each feasible query set contains at least one element of $\{f_i,l_i\}$. Let $Q$ be an arbitrary feasible query set. 
	By Lemma~\ref{lemma_mst_1}, $Q_{i-1} := Q \cap E_{i-1}$ is a feasible query set for $G_{i-1}$ and verifies some MST $T_{i-1}$ for $G_{i-1}$. 
	We show that $l_i \not\in Q_{i-1}$ implies either $l_i \in Q $ or $f_i \in Q$. 	
	
	Assume $l_i \not\in Q_{i-1}$ and let $C$ be the unique cycle in $T_{i-1} \cup \{f_i\}$.
	Since $T_L = T_U$, edge $f_i$ has the highest upper limit in $C$ after querying $Q_{i-1}$.
	While we only assume $T_L = T_U$ for the initially given instance,~\cite{megow17mst} show that this still implies that $f_i$ has the highest upper limit in $C$. 
	If we show that $l_i\not\in Q_{i-1}$ implies $l_i \in C$, we can apply Lemma~\ref{lemma_mst_2} to derive that $\{f_i,l_i\}$ is a witness set for graph $G_i$, and thus either $f_i \in Q_i \subseteq Q$ or $l_i \in Q_i \subseteq Q$.
	For the remainder of the proof we show that $l_i\not\in Q_{i-1}$ implies $l_i \in C$. 
	Let $a$ and $b$ be the endpoints of $f_i$, then the path $P = C_i \setminus \{f_i\}$ from $a$ to $b$ is part of $T_L$ and $l_i$ has the highest upper limit in $P$. 
	Using $l_i \not\in Q_{i-1}$ we can apply Observation~\ref{obs_mst_1} to conclude that there must be a path $\hat{P}$ from $a$ to $b$ in $T_{i-1}$ such that $l_i$ has the highest upper limit on $\hat{P}$ after querying $Q_{i-1}$. 
	Therefore, $C = \hat{P} \cup \{f_i\}$ and it follows $l_i \in C$.
	
	If $w_{f_i} \in I_{l_i}$ and $l_i \not\in Q_{i-1}$, then $l_i$ must be queried to identify the maximal edge on $C$, thus it follows that $\{l_i\}$ is a witness set.
\end{proof}

\LemMstWitnessSec*

\begin{proof}
	Consider the set of edges $X_{i}$ in the cut of $G$ defined by the two connected components of $T_L \setminus \{l_i\}$. 
	By assumption, $l_i,f_i \in X_i$. However, $f_j \not\in X_i$ for all $j < i$,  
	otherwise $l_i \in C_j$ for an $f_j \in X_i$ with $j < i$ would follow and contradict the assumption. 
	We can observe $X_i \cap E_i = \{f_i,l_i\}$. 
	
	Let $Q$ be any feasible query set.
	According to Lemma~\ref{lemma_mst_1}, $Q_{i-1} = E_{i-1} \cap Q$ verifies an MST $T_{i-1}$ for $G_{i-1}$. 
	Consider the unique cycle $C$ in $T_{i-1} \cup \{f_i\}$. 
	As observed in~\cite{megow17mst}, $f_i$ has the highest upper limit on $C$ after querying $Q_{i-1}$. 
	Since $f_i \in X_i$, $f_i \in C$ and $C$ is a cycle, it follows that another edge in $X_i\setminus \{f_i\}$ must be part of $C$.
    We already observed $X_i \cap E_i = \{l_i,f_i\}$, and therefore $l_i \in C$. 
    Lemma~\ref{lemma_mst_2} implies that $\{f_i,l_i\}$ is a witness set. 
	If $w_{l_i} \in I_{f_i}$, then $f_i$ must be queried to identify the maximal edge in $C$, so it follows that $\{f_i\}$ is a witness set.
\end{proof}

\subsection{Proofs of Subsection~\ref{subsec_mst_phase_1} (Phase 1)}

In this section, we prove the lemmas of Subsection~\ref{subsec_mst_phase_1}. 
We start by characterizing non-prediction mandatory free instances (and cycles).

\LemMSTPredFreeIff*

\begin{proof}
	{
		Assume $\pred{w}_{f_i} \ge U_{e}$ and $\pred{w}_e \le L_{f_i}$ holds for each $e \in C_i \setminus \{f_i\}$ and each cycle $C_i$ with ${i} \in \{1,\ldots,l\}$.
		Then each $f_i \in E \setminus T_L$ is predicted to be maximal on $C_i$ and each $e \in T_L$ is predicted to be minimal in $X_e$.
		Assuming the predictions are correct, we can observe that each vertex cover of bipartite graph $\bar{G}$ is a feasible query set~\cite{erlebach14mstverification}. 
		Define $\bar{G} = (\bar{V},\bar{E})$ with $\bar{V} = E$ {(excluding trivial edges)} and $\bar{E} = \{ \{f_i,e\} \mid i \in \{1,\ldots,l\}, e \in C_i \setminus \{f_i\} \text{ and } I_e \cap I_{f_i} \not= \emptyset\}$. 
		Since both $Q_1 := T_L$ and $Q_2 := E \setminus T_L$ are vertex covers for $\bar{G}$, $Q_1$ and $Q_2$ are feasible query sets under the assumption that the predictions are correct.
		This implies that no element is part of every feasible solution because $Q_1 \cap Q_2 = \emptyset$.
		We can conclude that no element is prediction mandatory and the instance is prediction mandatory free.
	}
	
	{
		For the other direction we show the contraposition.
		Assume there is a cycle $C_i$ such that $\w_{f_{i}} \in I_e$ or $\w_e \in I_{f_i}$ for some $e \in C_i \setminus \{e\}$.
		Let $C_i$ be such a cycle with the smallest index.
		If $\w_{f_i} \in I_e$ for some $e \in C_i \setminus \{e\}$, then also $\w_{f_i} \in I_{l_i}$ for the edge $l_i$ with the highest upper limit in $C_i \setminus \{f_i\}$.
		(This is because we assume $T_L = T_U$.)
		Under the assumption that the predictions are true, Lemma~\ref{lemma_mst_witness_set_1} implies that $l_i$ is mandatory and thus prediction mandatory.
		It follows that $G$ is not prediction mandatory free.}
	
	{
		Assume $\w_e \in I_{f_i}$.
		We can conclude $e \not\in C_j$ for each $j < i$.
		This is because $\w_e \in I_{f_i}$ and $j < i$ would imply $\w_e \in I_{f_j}$.
		As we assumed that $C_i$ is the first cycle with this property, $e \in C_j$ leads to a contradiction.
		Under the assumption that the predictions are true, Lemma~\ref{lemma_mst_witness_set_2} implies that ${f}_i$ is mandatory and thus prediction mandatory.
		It follows that $G$ is not prediction mandatory free.}
\end{proof}

Recall that $C_i$ is the non-prediction mandatory free cycle with the smallest index such that all $C_j$ with $j < i$ are prediction mandatory free {and that $l_i$ is the edge with the highest upper limit in $C_i \setminus \{f_i\}$.}

\MstPhaseOneCaseA*

\begin{proof}
	{To show that querying $\{f_i,l_i\}$ satisfies the two statements of Section~\ref{subsec_mst_phase_1}, we show that either $\{f_i,l_i\} \subseteq Q$ for each feasible query set $Q$ or $h_{f_i} + h_{l_i} \ge 1$ for the hop distance $h_{f_i} + h_{l_i}$ of $f_i$ and $l_i$.}
	
	By assumption, all $C_j$ with $j < i$ are prediction mandatory free. 
	We claim that this implies $l_i \not\in C_j$ for all $j < i$. 
	Assume, for the sake of contradiction, that there is a $C_j$ with $j < i$ and $l_i \in C_j$. 
	Then, $T_L = T_U$ and $i < j$ imply that $f_i$ and $f_j$ have larger upper and lower limits than $l_i$ and, since $L_{f_i} \ge L_{f_j}$, it follows $I_{l_i} \cap I_{f_i} \subseteq I_{l_i} \cap I_{f_j}$.
	Thus, $\pred{w}_{l_i} \in I_{f_i}$ implies $\pred{w}_{l_i} \in I_{{f}_j}$, which contradicts $C_j$ being prediction mandatory free.	
	According to Lemma~\ref{lemma_mst_witness_set_2}, $\{f_i,l_i\}$ is a witness set.
	 
	Consider any feasible query set $Q$, then $Q_{i-1}$ verifies the MST $T_{i-1}$ for graph $G_{i-1}$ and $Q$ needs to identify the maximal edge on the unique cycle $C$ in $T_{i-1} \cup \{f_i\}$. 
	Following the argumentation of Lemma~\ref{lemma_mst_witness_set_2}, we can observe $l_i,f_i \in C$.
	Since we assume $T_L = T_U$, we can also observe that $f_i$ has the highest upper limit in $C$. 
	By Observation~\ref{obs_mst_1} $l_i$ has the highest upper limit in $C \setminus \{f_i\}$ after querying $Q_{i-1} \setminus \{l_i\}$. 
	
	If $w_{l_i} \in I_{f_i}$, then $f_i$ is part of any feasible query set according to Lemma~\ref{lemma_mst_witness_set_2}.
	Otherwise, $w_{l_i} \le L_{f_i} < \pred{w}_{l_i}$ and ${h_{l_i}} \ge 1$. 
	If $w_{f_i} \in I_{l_i}$, then $l_i$ is part of any feasible query set according to Lemma~\ref{lemma_mst_witness_set_1}.
	Otherwise, $w_{f_i} {\ge} U_{l_i} {>} \pred{w}_{f_i}$ and ${h_{f_i}} \ge 1$.
	In conclusion, either $\{f_i,l_i\} \subseteq Q$ for any feasible query set $Q$ or $h_{f_i} + h_{l_i} \ge 1$
\end{proof}

\MstPhaseOneCaseB*

\begin{proof}
	By assumption, all $C_j$ with $j < i$ are prediction mandatory free. 
	According to Lemma~\ref{lemma_mst_witness_set_1}, $\{f_i,l_i\}$ is a witness set.
	{ 
	Assume that the edge $l_i'$ exists. 
	To show that the query strategy for this case satisfies the statements of Subsection~\ref{subsec_mst_phase_1}, 
	we show that either $|\{f_i,l_i,l_i'\} \cap Q| \ge 2$ for any feasible query set {$Q$} (in case $w_{f_i} \in I_{l_i}$ and $w_{l_i} \not\in I_{f_j}$ for each $j$ with $l_i \in C_j$) or $h_{f_i} + h_{l_i} \ge 1$.
	}
	
	Assume that either $w_{f_i} \not\in I_{l_i}$ or $w_{l_i} \in I_{f_j}$ for some $j$ with $l_i \in C_j$.
	If $w_{f_i} \not\in I_{l_i}$, then $w_{f_i} \ge U_{l_i} > \pred{w}_{f_i}$ and $h_{f_i} \ge 1$ follows. 
	If $w_{l_i} \in I_{f_j}$, then $\pred{w}_{l_i} \le L_{f_j} < w_{l_i}$ and $h_{l_i} \ge 1$ follows.
	Therefore, if either  $w_{f_i} \not\in I_{l_i}$ or $w_{l_i} \in I_{f_j}$ for some $j$ with $l_i \in C_j$, then $h_{l_i} + h_{f_i} \ge 1$ follows.
	Now assume that $w_{f_i} \in I_{l_i}$ and $w_{l_i} \not\in I_{f_j}$ for each~$j$ with $l_i \in C_j$. 
	According to Lemma~\ref{lemma_mst_witness_set_1}, $w_{f_i} \in I_{l_i}$ implies that $l_i$ is part of any feasible query set. 
	Consider the relaxed instance where $l_i$ is already queried, then $w_{l_i} \not\in I_{f_j}$ for each $j$ with $l_i \in C_j$ implies that $l_i$ is minimal in $X_{l_i}$ and that the lower limit tree does not change by querying $l_i$. 
	It follows that $l_i'$ is the edge with the highest upper limit in $C_i \setminus \{f_i\}$ in the relaxed instance and, by Lemma~\ref{lemma_mst_witness_set_1}, $\{f_i,l_i'\}$ is a witness set. 
	In conclusion, either $|\{f_i,l_i,l_i'\} \cap Q| \ge 2$ for any feasible query set $Q$ or $h_{l_i} + h_{f_i} \ge 1$.
	
	Finally, assume that $l_i'$ does not exist. 
	To show that the query strategy for this case satisfies the statements of Subsection~\ref{subsec_mst_phase_1}, 
	we show that we either can guarantee that the algorithm will not query $f_i$ and querying $e=l_i$ satisfies the third statement for $f(e)=f_i$,
	or $h_{l_i} + h_{f_i} \ge 1$ and querying $\{l_i,f_i\}$ satisfies the second statement.
	If $w_{l_i} \in I_{f_i}$, then $\pred{w}_{l_i} \le L_{f_i} < w_{l_i}$ and $h_{l_i} \ge 1$. 
	Assume otherwise. 
	The non-existence of $l_i'$ implies that $l_i$ is the only element of $C_i \setminus \{f_i\}$ with an interval that intersects $I_{f_i}$. 
	Therefore $w_{l_i} \not\in I_{f_i}$ implies that $f_i$ is uniquely maximal on $C_i$ and not part of any MST.
	It follows that $f_i$ can, without loss of generality, be deleted. 
	This guarantees that the algorithm will not query, or even consider, $f_i$ afterwards.
\end{proof}

Recall that, for each $e \in T_L$, the set $X_e$ is defined as the set of edges in the cut between the two connected components of $T_L\setminus \{e\}$.

\MstPhaseOneCaseC*

\begin{proof}
	By assumption, all $C_j$ with $j < i$ are prediction mandatory free. 
	We claim that this implies $l_i' \not\in C_j$ for all $j < i$. 
	Assume, for the sake of contradiction, that there is a $C_j$ with $j < i$ and $l_i' \in C_j$. 
	Then, $T_L = T_U$ and $i < j$ imply that $f_i$ and $f_j$ have larger upper and lower limits than $l_i'$ and, since $L_{f_i} \ge L_{f_j}$, it follows $I_{l_i'} \cap I_{f_i} \subseteq I_{l_i'} \cap I_{f_j}$.
	Thus, $\pred{w}_{l_i'} \in I_{f_i}$ implies $\pred{w}_{l_i'} \in I_{{f}_j}$, which contradicts $C_j$ being prediction mandatory free.	
	According to Lemma~\ref{lemma_mst_witness_set_2}, $\{f_i,l_i'\}$ is a witness set.

	{ 
	Assume that the edge $f_j$ exists. 
	To show that the query strategy for this case satisfies the statements of Subsection~\ref{subsec_mst_phase_1}, 
	we show that either $|\{f_i,f_j,l_i'\} \cap Q| \ge 2$ for any	 feasible query set {$Q$} (in case $\w_{l_i'} \in I_{f_i}$ and $w_{f_i} \not\in I_{e}$ for each $e \in C_i$) or $h_{f_i} + h_{l_i'} \ge 1$.
	}
	
	Assume that either $w_{l_i'} \not\in I_{f_i}$ or $w_{f_i} \in I_{e}$ for some $e \in  C_i$.
	If $w_{l_i'} \not\in I_{f_i}$, then $w_{l_i'} \le L_{f_i} < \w_{l_i'}$ and $h_{l_i'} \ge 1$ follows. 
	If $w_{f_i} \in I_{e}$, then $\w_{f_i} \ge U_e > w_{f_i}$ and $h_{f_i} \ge 1$ follows.
	Therefore  $w_{l_i'} \not\in I_{f_i}$ or $w_{f_i} \in I_{e}$ for some $e \in C_i$ implies $h_{l_i'} + h_{f_i} \ge 1$.
	
	Now assume that $w_{l_i'} \in I_{f_i}$ and $w_{f_i} \not\in I_{e}$ for each $e \in C_i$. 
	According to Lemma~\ref{lemma_mst_witness_set_2}, $w_{l_i'} \in I_{f_i}$ implies that $f_i$ is part of any feasible query set. 
	Consider the relaxed instance where $f_i$ is already queried, then $w_{f_i} \not\in I_{e}$ for each $e \in C_i$ implies that $f_i$ is maximal in $C_i$ and that the lower limit tree does not change by querying $f_i$. 
	It follows that $f_j$ is the edge with the smallest index and $l_i' \in C_j$ in the relaxed instance and, by Lemma~\ref{lemma_mst_witness_set_2}, $\{f_j,l_i'\}$ is a witness set. 
	In conclusion, either $|\{f_i,f_j,l_i'\} \cap Q| \ge 2$ for any feasible query set $Q$ or $h_{l_i'} + h_{f_i} \ge 1$.

	Finally, assume that $f_j$ does not exist. 
	To show that the query strategy for this case satisfies the statements of Subsection~\ref{subsec_mst_phase_1}, 
	we show that we either can guarantee that the algorithm will not query $l_i'$ and querying $e=f_i$ satisfies the third statement for $f(e)=l_i'$,
	or $h_{l_i'} + h_{f_i} \ge 1$ and querying $\{l_i',f_i\}$ satisfies the second statement.
	If $w_{f_i} \in I_{l_i'}$, then $\pred{w}_{f_i} \ge U_{l_i'} > w_{f_i}$ and $h_{f_i} \ge 1$. 
	Assume otherwise. 
	The non-existence of $f_j$ implies that $f_i$ is the only element of $X_{l_i'} \setminus \{f_i\}$ with an interval that intersects $I_{l_i'}$. 
	Therefore $w_{f_i} \not\in I_{l_i'}$ implies that $l_i'$ is uniquely minimal in $X_{l_i'}$ and part of every MST.
	It follows that $l_i'$ can, without loss of generality, be contracted. 
		This guarantees that the algorithm will not query, or even consider, $l_i'$ afterwards.
\end{proof}

\MSTEndOfPhaseOne*

\begin{proof}
	Since Algorithm~\ref{ALG_mst_part_1} only terminates if Line~\ref{line_mst_one_if_pred_free} determines each $C_i$ to be prediction mandatory free, the instance after executing the algorithm is prediction mandatory free by definition {and Lemma~\ref{mst_pred_free_characterization}}.
	All elements queried in Line~\ref{line_mst_one_fillup} to ensure unique $T_L=T_U$ are mandatory by Lemma~\ref{mst_preprocessing} and never decrease the performance guarantee.
	
	We now consider the remaining queries.
	Since the last iteration is ignored, each iteration $i$ queries a set $P_i$ of $\gamma-2$ prediction mandatory elements in Line~\ref{line_mst_one_fillup} and a set $W_i$ in Line~\ref{line_mst_one_lemma_queries}.
	According to Theorem~\ref{Theo_hop_distance_mandatory_distance} each element of $P_i$ is either mandatory or contributes one to the hop distance. 
	To be more precise, let $h'(e)$ with $e \in E$ be the number of edges $e'$ such that the value of $e'$ passes over an endpoint of $I_e$.
	From the arguments in the proof of Theorem~\ref{Theo_hop_distance_mandatory_distance}, it can be seen that, for each edge~$e$ that is prediction mandatory at some point but not mandatory, we have that $h'(e) \geq 1$.
	For a subset $U \subseteq E$, let $h'(U) = \sum_{e \in U} h'(e)$.
	Note that $k_h = h'(E)$ holds by reordering summations.
		
	Lemmas~\ref{mst_phase1_case_a},~\ref{mst_phase1_case_b} and~\ref{mst_phase1_case_c} imply that the three statements of Section~\ref{subsec_mst_phase_1} hold for set $W_i$.
	
	We slightly rephrase the second statement and state the error-dependent guarantee in terms of $h'$:
	
	\begin{compactenum}
		\setcounter{enumi}{1}
		\item  If the algorithm queries a witness set {$W_i=\{e_1,e_2\}$} of size two, then either {$W_i \subseteq Q$} for each feasible query set $Q$ or the hop distances of $e_1$ and $e_2$ satisfy $h'(e_1) + h'(e_2) \ge 1$.
	\end{compactenum}

	Technically, this rephrased statement does not always hold, since Lemmas~\ref{mst_phase1_case_b} and~\ref{mst_phase1_case_c} count \enquote{hops} that are caused by values of queried elements $e$ passing over interval borders of non-queried elements $e'$.
	Those errors are then counted by $h'(e')$ instead of $h'(e)$.
	We conceptually move those errors from $h'(e')$	to $h'(e)$ by increasing $h'(e)$ by one and decreasing $h'(e')$ by one; afterwards, the rephrased statement holds.
	Since the queried elements will not be considered again, this operation does not count errors multiple times.
	Also, this operation does not use errors $h'(P_j)$ of previous iterations $j < i$, as the lemmas only count hops over interval borders of non-trivial (non-queried) intervals and the sets $P_j$ have been queried.

	Consider first the set of elements $E$ that are queried because they fulfill the third statement of Subsection~\ref{subsec_mst_phase_1}.
	We assume without loss of generality that $E \subseteq \OPT$ and treat each $e \in E$ as a witness set of size one.
	We can do this without loss of generality since if $e \not\in \OPT$, we know  $f(e) \in \OPT$ for a distinct $f(e)$ and can charge $e$ against $f(e)$ instead.
	All other sets $Q_i := P_i \cup W_i$ are compared only against $\OPT \cap Q_i$ and therefore we can guarantee that $f(e)$ is not used to charge for another element. 
	Since the proofs of the Lemmas~\ref{mst_phase1_case_b} and~\ref{mst_phase1_case_c} show that each $f(e)$ can be either deleted (because it is uniquely maximal on a cycle) or contracted (because it is uniquely minimal in a cut), no $f(e)$ will be relevant for the second phase.

	Define $Q_i := P_i \cup W_i$. We show that $|Q_i| \le \min\{(1+\frac{1}{\gamma}) \cdot (|Q_i \cap \OPT| + h'(Q_i)), \gamma \cdot |Q_i \cap \OPT|\}$ for each $Q_i$.
	As all $Q_i$'s are disjoint, this implies the lemma.
	
	If the set $W_i$ in Line~\ref{line_mst_one_lemma_queries} contains one element, then we can assume that that element is part of any feasible solution.
	It follows that the set $Q_i$ is a witness set of size $\gamma - 1$ and therefore does not violate the robustness.
	Additionally, at least $|Q_i| - h'(P_i)$ elements of $Q_i$ are part of any feasible query set $Q$. This implies $|Q_i| \le \min\{ |Q_i\cap\OPT| + h'(Q_i), \gamma \cdot |Q_i \cap \OPT|\}$.

	If the set $W_i$ of Line~\ref{line_mst_one_lemma_queries} contains two elements, then either ${W_i} = \{e_1,e_2\} \subseteq Q$ for any feasible query set $Q$ or $h'(e_1) + h'(e_2) \ge 1$.
	{It} follows that {at least} $|Q_i| - h'(Q_i)$ elements of $Q_i$ are part of any feasible query set $Q$.
	This implies $|Q_i| \le (|Q_i \cap \OPT| + h'(Q_i)))$.
	Additionally, $Q_i$ is a witness set of size at most $\gamma$, which implies $|Q_i| \le \gamma \cdot |Q_i \cap \OPT|$.
	
	If the set $W_i$ in Line~\ref{line_mst_one_lemma_queries} contains three elements, then $|{W_i} \cap Q| \ge 2$ for any feasible query set $Q$.
	It follows that at least $\frac{\gamma}{\gamma + 1} \cdot |Q_i| - h'(P_i)$ elements of ${Q_i}$ are part of any feasible query set $Q$.
	Additionally, ${Q_i}$ is a set of size $\gamma+1$ and at least two elements of ${Q_i}$ are part of any feasible solution.
 	It follows $Q_i \le  \min\{(1+\frac{1}{\gamma}) \cdot (|Q_i \cap \OPT| + h'(Q_i)), \gamma \cdot |Q_i \cap \OPT|\}$.
	
	Since $|Q_i| \le \min\{(1+\frac{1}{\gamma}) \cdot (|Q_i \cap \OPT| + h'(Q_i)), \gamma \cdot |Q_i \cap \OPT|\}$ holds for each $Q_i$ and all $Q_i$ are disjoint, the lemma follows.
\end{proof}

\subsection{Proofs of Subsection~\ref{subsec_mst_phase_2} (Phase 2)}

	Recall that the vertex cover instance $\bar{G}$ for an MST under uncertainty instance $G=(V,E)$ is defined as $\bar{G} = (\bar{V},\bar{E})$ with $\bar{V} = E$ (excluding trivial edges) and $\bar{E} = \{ \{f_i,e\} \mid i \in \{1,\ldots,l\}, e \in C_i \setminus \{f_i\} \text{ and } I_e \cap I_{f_i} \not= \emptyset\}$ \cite{erlebach14mstverification,megow17mst}.
	Recall that $VC$ is a minimum vertex cover of $\bar{G}$ and that $h$ is a maximum matching of $\bar{G}$ that matches each $e \in VC$ to a distinct $h(e) \not\in VC$.
	By definition, $\bar{G}$ is a bipartite graph.

\MstPhaseTwoOne*

\begin{proof}
	We start by showing the first part of the lemma, i.e., if $b$ is such that each $f'_{i}$ with $i < b$ is maximal in cycle $C_{f'_i}$, then $\{f'_{i},h(f'_{i})\}$ is a witness set for each $i \le b$.
	
	Consider an arbitrary $f'_i$ and $h(f'_i)$ with $i \le b$.
	By definition of $h$, the edge $h(f'_i)$ is part of the lower limit tree. 
	Let $X_i$ be the cut between the two components of $T_L \setminus \{h(f'_i)\}$, then we claim that $X_i$ only contains $h(f'_i)$ and edges in $\{f'_{1},\ldots,f'_{g}\}$ (and possibly irrelevant edges that do not intersect $I_{h(f'_i)}$).
	To see this, assume an $f_j \in \{f_1,\ldots,f_l\}$ with $f_j \not\in VC$ was part of $X_i$. 
	Since $f_j \not\in VC$, each edge in $C_j \setminus \{f_j\}$ must be part of $VC$ as otherwise $VC$ would not be a vertex cover. 
	But if $f_j$ is in cut $X_i$, then $C_j$ must contain $h(f'_i)$. 
	By definition, $h(f'_i) \not\in VC$ holds  which is a contradiction. 
	We can conclude that $X_i$ only contains $h(f'_i)$ and edges in $\{f'_{1},\ldots,f'_{g}\}$.
	
	Now consider any feasible query set $Q$, then $Q_{i'-1}$ verifies an MST $T_{i'-1}$ for graph $G_{i'-1}$ where $i'$ is the index of $f'_i$ in the order $f_1,\ldots,f_l$. 
	Consider again the cut $X_i$. 
	Since each $\{f'_1,\ldots,f'_{i-1}\}$ is maximal in a cycle by assumption, $h(f'_{i})$ is the only edge in the cut that can be part of $T_{i'-1}$.
	Since one edge in the cut must be part of the MST, we can conclude that $h(f'_i) \in T_{i'-1}$. 
	Finally, consider the cycle $C$ in $T_{i'-1}\cup \{f'_i\}$. 
	We can use that $f'_i$ and $h(f'_i)$ are the only elements in $X_i \cap (T_{i'-1}\cup \{f'_i\})$ to conclude that $h(f'_i) \in C$ holds. 
	According to Lemma~\ref{lemma_mst_2}, $\{f'_i,h(f'_i)\}$ is a witness set. 
	
	We continue by showing the second part of the lemma, i.e., if $d$ is such that each $l'_{i}$ with $i < d$ is minimal in cut $X_{l'_i}$, then $\{l'_{i},h(l'_{i})\}$ is a witness set for each $i \le d$.
	
	Consider an arbitrary $l'_i$ and $h(l'_i)$ with $i \le d$. 
	By definition of $h$, the edge $h(l'_i)$ is not part of the lower limit tree.
	Consider $C_{h(l'_i)}$, i.e., the cycle in $T_L \cup \{h(l_i')\}$, then we claim that $C_{h(l'_i)}$ only contains $h(l'_i)$ and edges in $\{l'_{1},\ldots,l'_{k}\}$ (and possibly irrelevant edges that do not intersect $I_{h(l'_i)}$). 
	To see this, recall that $l_i' \in VC$, by definition of $h$, implies $h(l_i') \not\in VC$.
	For $VC$ to be a vertex cover, each $e \in C_{h(l_i')}\setminus\{h(l_i')\} $ must either be in $VC$ or not intersect $h(l_i')$.
	Consider the relaxed instance where the true values for each $l'_j$ with $j<d$ and $j \not= i$ are already known. 
	By assumption each such $l'_j$ is minimal in its cut $X_{l'_j}$.
	Thus, we can w.l.o.g contract each such edge.
	It follows that in the relaxed instance $l'_i$ has the highest upper limit in $C_{h(l'_i)} \setminus \{h(l'_i)\}$.
	According to Lemma~\ref{lemma_mst_witness_set_1}, $\{l'_i, h(l'_i)\}$ is a witness set.
\end{proof}

	Note that {if a query reveals that} an $f_i$ is not maximal in $C_i$ or an $l_i$ is not minimal in $X_{l_i}$, then the vertex cover instance changes.
	The following lemma considers the situation where the vertex cover instance in Line~\ref{line_mst_two_restart_if} of Algorithm~\ref{ALG_mst_part_2} has changed in comparison to the initial vertex cover instance of the iteration.
	Thus, the algorithm restarts.
	Let $\bar{G}$ be the initial vertex cover instance and let $\bar{G}'$ be the changed vertex cover instance that the algorithm considers after the restart.

\MstPhaseTwoTwo*

\begin{proof}
	The goal of this proof is to show that the number of times Algorithm~\ref{ALG_mst_part_2} executes the Otherwise-part of Line~\ref{line_mst_two_ensure_three_robust} is limited by $2 \cdot k_h$ if the algorithm uses recovery strategy B and restarts with the vertex cover and matching as described in the lemma.
	
	We can observe that the otherwise-part of Line~\ref{line_mst_two_ensure_three_robust} is only executed if the previous line queried an element $e'$ such that $\{e',h(e')\} \cap W \not= \emptyset$. 
	This can only happen if {either $h(e')$ or $e'$} was added to set $W$ in a previous restart, i.e., {$h(e') = h(e)$ or $e' = h(e)$} for an element $e$ that was queried before $e'$.
	Thus, we can show that the number of times the Otherwise-part of Line~\ref{line_mst_two_ensure_three_robust} is executed is limited, by showing that the number of elements $h(e)$ that are re-matched after being matched to a first partner $e$ is limited.
	
	Now consider an $h(e)$ that is added to $W$ in Line~\ref{line_mst_two_ensure_three_robust} after its first partner $e$ gets queried in Line~\ref{line_mst_two_query_vc}, then $h(e)$ can only get re-matched to another element $e'$ if the algorithm restarts afterwards. 
	If the algorithm does not restart, then it continues with the matching that already matched $h(e)$ to its original partner $e$ and therefore $h(e)$ will not be matched to another element.
	Consider the first restart after $e$ was queried and let $\bar{G}$ be the vertex cover instance before the restart, let $h$ be the maximum matching for $\bar{G}$ and let $\bar{G}'$ be the vertex cover instance after the restart.
	We can observe that $\{e,h(e)\}$ is not an edge in vertex cover instance $\bar{G'}$ because $e$ became trivial and $\bar{G'}$ by definition does not contain trivial elements.
	It follows that $h(e)$ is not matched by the partial matching  $\overline{h}=\{\{e,e'\} \in h \mid \{e,e'\} \in \bar{E}'\}$.
	We can conclude that each $h(e) \in W$ that is matched a second time, was not matched by the partial matching $\overline{h}$ in the restart after $e$ was queried.
	The only way for $h(e)$ to be matched a second time is if it is added to the matching while completing $\overline{h}$ using the standard augmenting path algorithm.
	Note that this does not need to happen in the restart directly after $e$ was queried, but since $h(e)$ at some point left the matching it can only be matched a second time if it is re-introduced to the matching when executing the augmenting path algorithm in a restart, i.e., it was not matched before executing the augmenting path algorithm but is matched afterwards.
	Thus, if we show that the total number of elements that were not matched before an execution of the augmenting path algorithm but are matched after the execution is bounded by $2 \cdot k_h$, then the lemma follows.
	
 	Consider any restart $i$.
 	Define $H_i$ to be the set of elements that were queried since the last restart before~$i$.
 	We will show in the subsequent Lemma~\ref{lemma_restart_rec_B} that the number of elements that are not matched by $\overline{h}$ but become matched after executing the augmenting path algorithm in restart $i$ is bounded by $2 \cdot \sum_{e\in H_i} h_e$ where $h_e$ is the hop distance of element $e$.
 	
 	We can observe that all sets $H_i$ are pairwise disjoint because no element is queried multiple times.
	Let $d$ be the total number of restarts. Since the sets $H_i$ with $i \in \{1,\ldots,d\}$ are disjoint, {it follows} $$\sum_{i\in \{1,\ldots,d\}} 2 \cdot \sum_{e \in H_i} h_e \le 2 \cdot k_h.$$ 
	Therefore Lemma~\ref{lemma_restart_rec_B} implies that the total number of elements that were not matched before executing the augmenting path but were matched afterwards is, over all restarts, bounded by $2 \cdot k_h$.
	This implies the lemma.
\end{proof}

In order to show Lemma~\ref{lemma_restart_rec_B}, we first show another auxiliary lemma.

\begin{lemma}
	\label{lemma_mst_tree_change}
	Let $G=(V,E)$ be an instance with unique $T_L$ and $T_U$ such that $T_L = T_U$. 
	Let $G'=(V',E')$ be an instance obtained from $G$ by executing a query set $Q$ with unique $T_L'$ and $T_U'$ such that $T_L' = T_U'$, where $T_L'$ and $T_U'$ are the lower and upper limit tree of $G'$.
	Then $e \in T_L \setminus T_L'$ implies $e \in Q$, and $e \in T_L' \setminus T_L$ implies $e \in Q$.
\end{lemma}

\begin{proof}
	Let $e \in T_L \setminus T_L'$, then $e \in T_L$ and $T_L$ being unique imply that $e$ has the unique minimal lower limit in the cut $X_e$ of $G$ between the two connected components of $T_L \setminus \{e\}$.
	Thus, $e$ is part of any lower limit tree for $G$.
	For $e$ to be not part of $T'_L$, it cannot have the unique minimal lower limit in the cut $X_e$ of $G'$.
	Since querying elements in $X_e \setminus \{e\}$ only increases their lower limits, this can only happen if $e \in Q$.
	
	Let $e \in T_L' \setminus T_L$. Then $T_L = T_U$ and $T_L' = T_U'$ imply $e \in T_U' \setminus T_U$.
	Since $e \not\in T_U$ and $T_U$ is unique, it follows that $e$ has the unique largest upper limit in the cycle $C_e$ of $T_U \cup \{e\}$.
	Thus, $e$ is not part of any upper limit tree for $G$.
	For $e$ to be part of $T_U'$, is cannot have the unique largest upper limit in the cycle $C_e$ of $G'$.
	Since querying elements in $C_e \setminus \{e\}$ only decreases their upper limits, this can only happen if $e \in Q$.
\end{proof}

\begin{lemma}\label{lemma_restart_rec_B}
Let $\bar{G}$ and $\bar{G}'$ be the vertex cover instances before and after a restart, let $H$ be the set of elements queried since the previous restart, and let $h$ be the maximum matching for $\bar{G}$. Then the number of elements that are not matched by $\overline{h}=\{\{e,e'\} \in h \mid \{e,e'\} \in \bar{E}'\}$ but become matched after completing $\overline{h}$ using the augmenting path algorithm is bounded by $2 \cdot \sum_{e \in H} h_e$.
\end{lemma}

\begin{proof}
	Let $h'$ be the maximum matching constructed from $\overline{h}$ using the augmenting path algorithm.
	In each iteration, the augmenting path algorithm increases the size of the matching by one and matches at most two elements that were not matched before.
	If we can show that the number of iterations of the augmenting path algorithm is bounded by $\sum_{e \in H} h_e$, it follows that at most $2\cdot\sum_{e \in H} h_e$ elements that are not matched by $\overline{h}$ are matched by $h'$, which implies the lemma.
	
	According to K\"onig`s Theorem, see, e.g.,~\cite{Biggs1986}, the size of $h'$ is upper bounded by the size of the minimum vertex cover for~$\bar{G}'$. 
	Since the augmenting path algorithm increases the size of the matching by one in each iteration, the number of iterations is bounded by $|VC'| - |\overline{h}|$, where $|VC'|$ is the size of a minimum vertex cover $VC'$ for $\bar{G}'$ and $|\overline{h}|$ is the size of the matching $\overline{h}$, i.e., the number of edges of  $\bar{G}'$ in $\overline{h}$.
	We show that the size of the minimum vertex cover is at most $|\overline{h}| + \sum_{e \in H} h_e$ and the statement follows.

	In order to do so, we define $\overline{VC} = \{e \in VC \mid \exists e' \text{ s.t. } \{e,e'\} \in \overline{h}\}$ for $\bar{G}$ where $VC$ is the vertex cover for $\bar{G}$ as defined by $h$.
	Note that $|\overline{VC}| = |\overline{h}|$.
	We show that we can construct a vertex cover for $\bar{G}'$ by adding at most $\sum_{e \in H} h_e$ elements to $\overline{VC}$, which implies $|VC'| \le |\overline{h}| + \sum_{e \in H} h_e$ for the minimum vertex cover $VC'$ of $\bar{G}'$.
	
	We argue how to extend $\overline{VC}$ such that it covers each edge in $\bar{G}'$ by adding at most $\sum_{e\in H} h_e$ elements.
	We sequentially consider edges $\{f,e\} \in \bar{E}'$ that are not covered by $\overline{VC}$ and add an element of $\{f,e\}$ to $\overline{VC}$ in order to cover $\{f,e\}$.
	In the process we argue that we can charge each added element to a distinct error.
	{We will consider pairs of elements $\{e,e'\}$ such that the value of $e'$ passes over a boundary of $e$.
	In that case we say that $e$ contributes to the hop distance $h_{e'}$.} 
	
	Consider any edge of $\bar{G}'$ that is not covered by the current $\overline{VC}$.
	Either that edge is part of $\bar{E}' \setminus \bar{E}$ or is part of $\bar{E} \cap \bar{E'}$ but both endpoints of the edge are not part of $\overline{VC}$.
	{Let $T_L' = T_U'$ be the unique upper and lower limit tree after the restart, i.e., the vertex cover instance $\bar{G}'$ is based on $T_L'$.}
	
	\textbf{Case 1:} Consider an $\{f,e\} \in \bar{E}' \setminus \bar{E}$ that is not covered by the current $\overline{VC}$ with $e \in T'_L$ and $f \not\in T'_L$. 
	By definition of the vertex cover instance, $\{f,e\} \in \bar{E}'$ and $f \not\in T'_L$ imply $e \in C'_f$ such that $I_f \cap I_{e} \not= \emptyset$, where $C'_f$ is the unique cycle in $T_L' \cup \{f\}$.
	
	By definition, $\bar{G}'$ only contains non-trivial elements and thus $e$ and $f$ are non-trivial.
	According to Lemma~\ref{lemma_mst_tree_change}, $e$ being non-queried and $e \in T_L'$ imply $e \in T_L = T_U$.
	Similar, $f$ being non-queried and $f \not\in T_L'$ imply $f \not\in T_L = T_U$.
	
	Define $X'_{e}$ to be the set of edges in the cut of $G'$ between the two connected components of $T_L' \setminus \{e\}$.
	Remember that $f \in X'_e$ since $f \not\in T_L'$ and $e \in C'_f$.
	Since $\{f,e\} \not\in \bar{E}$ implies $e\not\in C_f$, there must be an element $l \in T_L \cap (X'_{e} \setminus \{e,f\})$ such that $l \in C_f$, where $C_f$ is the unique cycle in $T_L \cup \{f\}$.
	As the instance is prediction mandatory free, $l \in C_f$ implies $\w_l \not\in I_f$.

	{Further, $X'_e \cap T_L' = \{e\}$ holds by definition and implies $l \not\in T_L'$.}
	According to Lemma~\ref{lemma_mst_tree_change}, $l \in T_L \setminus T_L'$ implies that $l$ must have been queried after the previous restart.
	If $w_l \in I_e$, this implies that $l$ has the {unique} smallest upper limit in cut $X'_e$ and is therefore part of $T_U'$.
	This would imply $T_L' \not= T_U'$ which is a contradiction.
	If $w_l \not\in I_e$, then also $w_l \ge U_e$ ($w_l \le L_e$ cannot be the case because $l$ would have the smallest lower limit in $X'_e$ which would contradict $l \not\in T_L'$).
 	As $I_e \cap I_f \not= \emptyset$, we can conclude that $w_l \ge U_e$ implies $w_l > L_f$.
 	
 	{It follows $\w_l \le L_f < w_l$ and $l$ passes over $L_f$. Therefore $f$ contributes to the hop distance $h_l$.}
 	We already argued that $l$ was queried and thus $l \in H$ and we can afford to add $f$ to the vertex cover in order to cover edge $\{f,e\}$.	
 	Afterwards, all edges incident to $f$ are covered and we therefore do not need to consider any of those edges again. 
 	Thus, the error contributed by $f$ to $h_l$ will not be counted multiple times in Case~$1$.
 	{Further, $l \in T_L$ holds for the element $l$ that passes over an interval boundary and $f \not\in T_L$ holds for the non-trivial element $f$ that is passed over.}
 	{This combination of properties for the pair of elements $\{l,f\}$ that contribute an error is mutually exclusive to the properties of the element pairs considered in the following cases (cf. Table~\ref{table_restart_rec_B}).
 	Thus, the error contributed by $f$ to $h_l$ will not be counted in the other cases.
 	}

	\textbf{Case 2:} Consider an edge $\{f,e\}$ with $f\not\in T_L$ and $e \in T_L$ that is part of $\bar{E} \cap \bar{E'}$ but $f,e\not\in \overline{VC}$, then
	$\{f,e\} \in \bar{E}$ implies that either a) $f \in VC$ or b) $e \in VC$. 
	
	\textbf{Case 2a):} Assume $f \in VC$, then $h$ matches $f$ to an element $h(f)$ and we can conclude $\{h(f),f\}\not\in \bar{E}'$ since otherwise $\overline{h}$ would also match $f$ and $h(f)$ and therefore $f \in \overline{VC}$ would follow by definition of $\overline{VC}$.
	
	From $\{h(f),f\}\not\in \bar{E}'$ it follows that either $h(f) \not\in C'_f$, or $h(f) \in C'_{f}$ and $h(f)$ is trivial after the restart.
	By definition, $h(f) \in C'_f$ and $h(f)$ being non-trivial would imply $\{f,h(f)\} \in \bar{E}'$, a contradiction.

	Assume $h(f) \in C'_f$ and $h(f)$ is trivial.
	For $h(f)$ to be trivial it must have been queried after the last restart. 
	We can observe that $h(f)$ was not queried in Line~\ref{line_mst_two_ensure_three_robust} of Algorithm~\ref{ALG_mst_part_2} as a result of $\{f,h(f)\} \cap W = \emptyset$ because then $f$ would have been queried in Line~\ref{line_mst_two_query_vc} and
	$f$ being trivial would contradict $\{e,f\} \in \bar{E}'$.

	Since $f \in VC$ implies $h(f) \not\in VC$, we can conclude that $h(f)$ was also not queried as part of $VC$ in Line~\ref{line_mst_two_query_vc} but to ensure unique $T_L = T_U$ using Lemma~\ref{mst_preprocessing}.
	{Directly after the previous restart}, $h(f)$ had the uniquem minimum upper and lower limit in the cut $X_{h(f)}$ of $G$ between the two connected components of $T_L \setminus \{h(f)\}$
	{because $h(f) \in T_L=T_U$ and $T_L,T_U$ are unique}.
	Since $h(f)$ was queried to ensure $T_L = T_U$, it must, at some point before it was queried, have been part of the current lower limit tree but not of the upper limit tree.
	Thus, some element $f'$ of $X_{h(f)} \setminus \{h(f)\}$ must have been queried and $w_{f'} \in I_{h(f)}$. (Otherwise $h(f)$ would still have the unique smallest upper limit in $X_{h(f)}$ and therefore be still part of the current upper limit tree.)
	 
	Since the input instance is prediction mandatory free, it follows $\w_{f'} \ge U_{h(f)} > w_{f'}$.
	{Therefore $f'$ passes over $U_{h(f)}$ and $h(f)$ contributes to the hop distance $h_{f'}$.}	
	Since ${f'}$ was queried, ${f'} \in H$ and we can afford to add $f$ to $\overline{VC}$ in order to cover the edge $\{e,f\}$.
	
	Afterwards, all edges incident to $f$ are covered and we therefore do not need to consider any of those edges again in Case 2a).
	As $f \in VC$ is the unique partner of $h(f) \not\in VC$ and $f$ will not be considered again in Case 2a), the error contributed by $h(f)$ to $h_{f'}$ will not be counted again in Case 2a).
	For the element $h(f)$ that is passed over, it holds that $h(f) \in T_L$ and $h(f)$ is trivial and, for the element $f'$ that passes over, it holds that $f' \not\in T_L$ and $f'$ is trivial.
	This combination of properties for the pair of elements $\{h(f),f'\}$ that contribute an error is mutually exclusive to the properties of the element pairs considered in the other cases (cf. Table~\ref{table_restart_rec_B}).
	Thus, the error contributed by $h(f)$ to $h_{f'}$ will not be counted in the other cases.
	
	Now assume that $h(f)$ is non-trivial and $h(f) \not\in C'_f$.
	Note that $h(f) \in C_f \cap T_L$, since $f \not\in T_L$ and $\{f, h(f)\} \in \bar{E}$.
	Thus $f \in X_{h(f)}$ and, since $h(f) \not\in C'_f$, some edge $f' \in C'_f \cap T'_L$ must be in $X_{h(f)}$.
	From $f' \in X_{h(f)}$ follows $f' \not\in T_L$, since $T_L$ by definition only contains one element of cut $X_{h(f)}$ and that is $h(f)$.
	Additionally $f' \in X_{h(f)}$ implies $h(f) \in C_{f'}$, and since the instance is prediction mandatory free also $\w_{f'} \ge U_{h(f)}$.
	According to Lemma~\ref{lemma_mst_tree_change}, $f' \in T_l' \setminus T_L$ implies that $f'$ must have been queried {after the previous restart} and therefore is trivial.
	If $w_{f'} \ge U_{h(f)}$, then $I_{h(f)} \cap I_{f}\not=\emptyset$ implies $w_{f'} > L_f$.
	This contradicts $f' \in T'_L$ because $f'$ would have the unique highest lower limit on $C'_f$.
	If $w_{f'} < U_{h(f)}$, then $w_{f'} < U_{h(f)} \le \w_{f'}$.	
	{Therefore $f'$ passes over $U_{h(f)}$ and $h(f)$ contributes to the hop distance $h_{f'}$.}	
	Since ${f'}$ was queried, we have that ${f'} \in H$ and we can afford to add $f$ to $\overline{VC}$ in order to cover the edge $\{e,f\}$.
	
	Afterwards, all edges incident to $f$ are covered and we therefore do not need to consider any of those edges again in Case 2a).
	As $f \in VC$ is the unique partner of $h(f) \not\in VC$ and $f$ will not be considered again in Case 2a) , the error contributed by $h(f)$ to $h_{f'}$ will not be counted again in Case 2a).
	For the element $h(f)$ that is passed over, it holds that $h(f) \in T_L \setminus VC$ and $h(f)$ is non-trivial and, for the element $f'$ that passes over, it holds that $f' \not\in T_L$ and $f'$ is trivial.
	This combination of properties for the pair of elements $\{h(f),f'\}$ that contribute an error is mutually	 exclusive to the properties of the element pairs considered in the other cases (cf. Table~\ref{table_restart_rec_B}).
	Thus, the error contributed by $h(f)$ to $h_{f'}$ will not be counted in the other cases.

	\textbf{Case 2b):} Assume $e \in VC$, then $h$ matches $e$ to an element $h(e)$ and we can conclude that $\{h(e),e\}\not\in \bar{E}'$, since otherwise $\overline{h}$ would also match $e$ and $h(e)$, and therefore $e \in \overline{VC}$ would follow.
	Since $\{h(e),e\} \in \bar{E}$, it follows $e \in C_{h(e)}$ and $I_e\cap I_{h(e)} \not= \emptyset$.
	
	From $\{h(e),e\}\not\in \bar{E}'$ follows that either $e \not\in C'_{h(e)}$, or  $e \in C'_{h(e)}$ and $h(e)$ is trivial after the restart. 
	By definition, $e \in C'_{h(e)}$ and $h(e)$ being non-trivial would imply $\{e,h(e)\} \in \bar{E}'$, a contradiction.
	
	Assume $e \in C'_{h(e)}$ and $h(e)$ is trivial.
	For $h(e)$ to be trivial it must have been queried after the last restart.
	We can observe that $h(e)$ was not queried in Line~\ref{line_mst_two_ensure_three_robust} of Algorithm~\ref{ALG_mst_part_2} as a result of $\{e,h(e)\} \cap W = \emptyset$ because then $e$ would have been queried in Line~\ref{line_mst_two_query_vc}, and
	$e$ being trivial would contradict $\{e,f\} \in \bar{E}'$.
	
	Since $e \in VC$ implies $h(e) \not\in VC$, we can conclude that $h(e)$ was also not queried as part of $VC$ in Line~\ref{line_mst_two_query_vc} but to ensure $T_L = T_U$ using Lemma~\ref{mst_preprocessing}.	 
	{Directly after the previous restart}, $h(e)$ had the unique maximum upper and lower limit in cycle $C_{h(e)}$ {because $h(e) \not\in T_L=T_U$ and $T_L, T_U$ are unique}.
	Since $h(e)$ was queried to ensure $T_L = T_U$, it must,  at some point before it was queried, have been part of the current lower limit tree but not of the upper limit tree.
	Thus, 
	some element $e'$ of $C_{h(e)}$ must have been queried and $w_{e'} \in I_{h(e)}$.
	(Otherwise $h(e)$ would still have the unique largest lower limit in $C_{h(e)}$ and therefore not be part of the current lower limit tree.)
	Since the input instance is prediction mandatory free, $\w_{e'} \le L_{h(e)} < w_{e'}$.
	{Therefore $e'$ passes over $L_{h(e)}$ and $h(e)$ contributes to the hop distance $h_{e'}$.}	
	Since ${e'}$ was queried, ${e'} \in H$ and we can afford to add $e$ to $\overline{VC}$ in order to cover the edge $\{e,f\}$.

	Afterwards, all edges incident to $e$ are covered and we therefore do not need to consider any of those edges again {in Case 2b)}. 
	As $e \in VC$ is the unique partner of $h(e) \not\in VC$ and $e$ will not be considered again in Case 2b), the error contributed by $h(e)$ to $h_{e'}$ will not be counted again in Case 2a).	
	For the element $h(e)$ that is passed over, it holds that $h(e) \not\in T_L$ and $h(e)$ is trivial and, for the element $e'$ that is passes over, it holds that $e' \in T_L$ and $e'$ is trivial.	
	This combination of properties for the pair of elements $\{h(e),e'\}$ that contribute an error by is mutually exclusive to the properties of the element pairs considered in the other cases (cf. Table~\ref{table_restart_rec_B}).
	Thus, the error contributed by $h(e)$ to $h_{e'}$ will not be counted in the other cases.

	Now assume that $e \not\in C'_{h(e)}$.
	Note that $e \in C_{h(e)} \cap T_L$, since $h(e) \not\in T_L$ and $\{e, h(e)\} \in \bar{E}$.
	Thus $h(e) \in X_{e}$ and, since $e \not\in C'_{h(e)}$, some edge $e' \in C'_{h(e)} \cap T'_L$ must be in $X_{e}$.
		
	From $e' \in X_{e}$ it follows that $e' \not\in T_L$, since $T_L$ by definition only contains one element of cut $X_{e}$ and that is $e$.
	Additionally $e' \in X_{e}$ implies $e \in C_{e'}$, and since the instance is prediction mandatory free also $\w_{e'} \ge U_{e}$.
	
	By Lemma~\ref{lemma_mst_tree_change}, $e' \in T_L' \setminus T_L$ implies that $e'$ must have been queried {after the previous restart} and is trivial.
	If $w_{e'} \ge U_{e}$, then $I_{h(e)} \cap I_{e} \not=\emptyset$ implies $w_{e'} > L_{h(e)}$.
	This contradicts $e' \in T'_L$ because $e'$ would have the unique {highest lower limit} on $C'_{h(e)}$.
	If $w_{e'} < U_{e}$, then $w_{e'} < U_{e} \le \w_{e'}$.
	{Therefore $e'$ passes over~$U_{e}$ and $e$ contributes to the hop distance $h_{e'}$.}	
	{Since ${e'}$ was queried, we have that ${e'} \in H$ and we can afford to add $e$ to $\overline{VC}$ in order to cover the edge $\{e,f\}$.}
	
	Afterwards, all edges incident to $e$ are covered and we therefore do not need to consider any of those edges again in Case 2b). Therefore the error contributed by $e$ to $h_{e'}$ will not be counted again in Case 2b).
	For the element $e $ that is passed over, it holds that $e \in T_L \cap VC$ and $e$ is non-trivial and, for the element $e'$ that passes over, it holds that $e' \not\in T_L$ and $e'$ is trivial.
	This combination of properties for the pair of elements $\{e,e'\}$ that contribute an error is mutually exclusive to the properties of the element pairs considered in the other cases (cf. Table~\ref{table_restart_rec_B}).
	Thus, the error contributed by $e$ to $e'$ will not be counted in the other cases.

	\begin{table}[t]
		\centering
		\begin{tabular}{p{2cm}|p{2.5cm}p{1.75cm}p{2.5cm}p{2cm}p{2.5cm}}
			& Case 1 & Case 2a) & Case 2a) & Case 2b) & Case 2b)\\
			\hline
		\\
		Passed \newline Element &$f \not\in T_L$ \newline non-trivial in $G'$& $h(f) \in T_L$ \newline trivial in $G'$  & $h(f) \in T_L\setminus VC$ \newline non-trivial in $G'$ & $h(e) \not\in T_L$ \newline trivial in $G'$ & $e \in T_L \cap VC$ \newline non-trivial in $G'$\\
		\\
		Passing \newline Element & $l \in T_L$ \newline trivial in $G'$  & $f' \not\in T_L$ \newline trivial in $G'$ & $f' \not\in T_L$ \newline trivial in $G'$ & $e' \in T_L$ \newline trivial in $G'$ & $e' \not\in T_L$ \newline trivial in $G'$ \\[2ex]
		\end{tabular}
		\caption{Errors considered in the different cases of the proof of Lemma~\ref{lemma_restart_rec_B}.
		The first row shows properties of the element whose interval border is passed over, and the second row shows properties of the passing element. 
		Note that Cases 2a) and 2b) are listed twice because both contain two sub-cases.}
		\label{table_restart_rec_B}	
	\end{table}

	In summary, we can exhaustively execute Cases $1$ and $2$ until all edges of $\bar{G}'$ are covered while adding at most $\sum_{e\in H} h_e$ elements.
	As explained at the beginning of the proof, this implies the lemma.
\end{proof}

\MSTEndOfPhaseTwo*

\begin{proof}
	{We first show the first statement of the lemma, that the instance remains prediction mandatory free.}
	Let $G$ be the prediction mandatory free instance at the beginning of an iteration and let $T_L$ be the corresponding lower limit tree.
	Let~$G'$ be the instance after the next restart and assume~$G'$ is not prediction mandatory free.
	{We show that $G$ being prediction mandatory free implies that $G'$ is prediction mandatory free via proof by contradiction.}
	
	Let $T_L'$ be the lower limit tree of $G'$, let $f_1',\ldots,f'_{l'}$ be the (non-trivial) edges in $E' \setminus T_L'$ ordered by lower limit non-decreasingly, and let $C_i'$ be the unique cycle in $T_L' \cup \{f_i'\}$.
	By definition of the algorithm, $T_L' = T_U'$ holds and $T_L' = T_U'$ is unique.
	We can w.l.o.g. ignore trivial edges in $E' \setminus T_L'$ since those are maximal in a cycle and can be deleted.
	Since $G'$ is not prediction mandatory free, there must be some $C_i'$ such that either $\w_e \in I_{f_i'}$ or $\w_{f_i'} \in I_e$ for some non-trivial $e \in C_i' \setminus \{f_i'\}$.

	Assume $e \not\in T_L$.
	Since $e$ is part of $T_L' = T_U'$, Lemma~\ref{lemma_mst_tree_change} implies that $e$ must have been queried and therefore is trivial, which is a contradiction.
	Assume $e \in T_L$ and $\w_e \in I_{f_i'}$.
	Since $e \in T_L$, cycle $C_i'$ must contain some $f \in X_e \setminus \{e\}$, where $X_e$ is the cut between the two components of $T_L \setminus \{e\}$ in $G$.
	Instance $G$ being prediction mandatory free implies $\w_e \not\in I_f$ where $I_f$ denotes the uncertainty interval of $f$ before querying it.
	If $\w_e \in I_{f_i'}$, this implies $L_{f_i'} < L_{f}$. It follows that $f$ has the highest lower limit in $C_i'$, which contradicts $f_i'$ having the highest lower limit in $C'_i$.

	Assume $e \in T_L$ and $\w_{f_i'} \in I_e$. 
	Remember that $f_i'$ is non-trivial and $f_i' \not\in T_L' = T_U'$. 
	According to Lemma~\ref{lemma_mst_tree_change}, it follows $f_i' \not\in T_L = T_U$.
	Let $C_{f_i'}$ be the cycle in $T_L \cup \{f_i'\}$.
	Since $G$ is prediction mandatory free, $\w_{f_i'} \not\in I_{e'}$ for each $e' \in C_{f_i'} \setminus \{f_i'\}$, which implies $U_e > U_{e'}$.
	It follows that the highest upper limit on the path between the two endpoints of $f_i'$ in $T_U' = T_L'$ is strictly higher than the highest upper limit on the path between the two endpoints of $f_i'$ in $T_U=T_L$.
	We argue that this cannot happen and we have a contradiction to $e \in T_L$ and $\pred{f'_i} \in I_e$.
	
	{Let $P$ be the path between the endpoints $a$ and $b$ of $f_i'$ in $T_U$ and let $P'$ be the path between $a$ and $b$ in $T_U'$.
	Define $U_P$ to be the highest upper limit on $P$.
	Observe that the upper limit of each edge can only decrease from $T_U$ to $T_U'$ since querying edges only decreases their upper limits.
	Therefore each $e' \in P' \cap P$ cannot have a higher upper limit than $U_P$.
	It remains to argue that the upper limit of each $e' \in P' \setminus P$ cannot be larger than $U_P$.
	Consider the set $\mathcal{S}$ of maximal subpaths $S \subseteq P'$ such that $P \cap P' = \emptyset$.
	Each $e' \in P' \setminus P$ is part of such a subpath $S$.
	Let $S$ be an arbitrary element of $\mathcal{S}$, then there is a cycle $C \subseteq S \cup P$ with $S \subseteq C$.
	Assume $e' \in S$ has a strictly larger upper limit than $U_P$, then an element of $S$ has the unique highest upper limit on $C$.
	It follows that subpath $S$ and path $P'$ cannot be part of any upper limit tree in the instance $G'$, which contradicts the assumption of $P'$ being a path in $T_U'$.}

	{We conclude that the graph $G'$ is prediction mandatory free.}
	{Note that this proof is independent of the individual recovery strategy that is used in Algorithm~\ref{ALG_mst_part_2} and relies only on maintaining the property $T'_L=T'_U$ being unique.}
	
	{We continue by showing the performance guarantees starting with \textbf{recovery A}.}	
	If all predictions are correct, the algorithm queries exactly the minimum vertex cover $VC$ which, according to~\cite{erlebach14mstverification}, is optimal for the input instance of Algorithm~\ref{ALG_mst_part_2}, i.e., $|\ALG| \le \opt$. 
	If not all predictions are correct, the algorithm additionally queries elements of $T_L \setminus T_U$ in Line~\ref{line_mst_two_query_vc}.
	Since those elements are mandatory, querying them does not violate the $2$-robustness.
	Additionally, the algorithm might execute  recovery strategy A in Line~\ref{line_mst_two_recovery}, i.e., query all elements in $W$ and afterwards ensure $T_L = T_U$.
	According to Lemma~\ref{mst_phase2_1}, each $h(e) \in W$ forms a witness set of size two with a distinct already queried $e \in VC$.
	In summary, the algorithm only queries disjoint witness sets of size one and two which implies $|\ALG| \le 2\cdot \opt$.

	{We conclude the proof by showing the performance guarantee of \textbf{recovery B}. First, we show the robustness.} 
	All elements that were queried because they were in $T_L\setminus T_U$ are mandatory and querying them never decrease the robustness.
	According to Lemma~\ref{mst_phase2_1}, each element $e$ queried in Line~\ref{line_mst_two_query_vc} forms a witness set with the element $h(e)$.
	If $h(e) \not\in \ALG$ or $h(e)$ is queried as an element of $T_L\setminus T_U$, then $\{e,h(e)\}$ is disjoint to all other $\{e',h(e')\}$ pairs that are considered at Line~\ref{line_mst_two_query_vc}.
	If $h(e)$ is queried in Line~\ref{line_mst_two_query_vc} or~\ref{line_mst_two_ensure_three_robust} in a later iteration, then $h(e)$ must have been re-matched to an element $e'$ after a restart.
	In this case, the algorithm queries all of $\{h(e),e,e'\}$ because of Line~\ref{line_mst_two_ensure_three_robust}.
	It follows that $\{h(e),e,e'\}$ is a witness set and, since all elements are queried and will not be considered again, this is disjoint to all other considered witness sets.
	In summary, the algorithm queries only (subsets of) disjoint witness sets of at most size three.
	This implies that $|\ALG| \le 3 \cdot |\OPT|$.
	
	{We continue by showing the error-dependent guarantee.}
	Consider set $\ALG$.
	Observe that $|W| \le |\OPT|$ holds since $|VC| \le |\OPT|$ and, {for each element $h(e)\in W$, there is a distinct element $e\in VC$ such that $\{e,h(e)\}$ is a witness set. }
	Each $h(e) \in W$ was added to $W$ in Line~\ref{line_mst_two_ensure_three_robust} after the distinct $e \in VC$ was queried in Line~\ref{line_mst_two_query_vc}.
	Let~$S$ be the set of those queried elements, then $|S| = |W| \le |\OPT|$.
	
	Consider some $e \in \ALG \setminus S$, then $e$ was queried either $(i)$ as an element of $T_L \setminus T_U$ in Line~\ref{line_mst_two_query_vc} or~\ref{line_mst_two_ensure_three_robust},  or $(ii)$ as part of an witness set $\{e,h(e)\}$ with $\{e,h(e)\} \cap W \not= \emptyset$ in Line~\ref{line_mst_two_query_vc} or~\ref{line_mst_two_ensure_three_robust}.
	
	$(i)$ $e$ was queried as an element of $T_L \setminus T_U$ in Line~\ref{line_mst_two_query_vc} or~\ref{line_mst_two_ensure_three_robust}. {We argue that $e$ contributes at least one to the hop distance $k_h$. The element $e$ is not prediction mandatory, as the instance is by Lemma~\ref{mst_end_of_phase_one} prediction mandatory free but, as $e \in T_L\setminus T_U$, it is mandatory. The proof of Theorem~\ref{Theo_hop_distance_mandatory_distance} shows that $e$ contributes at least one to $k_h$.}
	To be more precise, let $h'(e)$ with $e \in E$ be the number of edges $e'$ such that the value of $e'$ passes over an endpoint of $I_e$.
	From the arguments in the proof of Theorem~\ref{Theo_hop_distance_mandatory_distance}, it can be seen that, for each edge~$e$ that is non-prediction mandatory at some point but is mandatory, we have that $h'(e) \geq 1$.
	For a subset $U \subseteq E$, let $h'(U) = \sum_{e \in U} h'(e)$.
	Note that $k_h = h'(E)$ holds by reordering summations.
	Thus at most $h'(E_{i}) \le k_h$ elements that satisfy~$(i)$ are queried, where $E_{i}$ denotes the set of all elements that satisfy~$(i)$.

	$(ii)$ $e$ was queried as part of an witness set $\{e,h(e)\}$ with $\{e,h(e)\} \cap W \not= \emptyset$ in Line~\ref{line_mst_two_query_vc} or~\ref{line_mst_two_ensure_three_robust}.
	Let $\{e,h(e)\}$ be a set that is queried in Line~\ref{line_mst_two_query_vc} and~\ref{line_mst_two_ensure_three_robust}.  
	Since both $e$ and $h(e)$ are queried, there must be some $l \in \{e,h(e)\} \cap W$.
	This means that $l$ must have been matched to a different element in an earlier restart.
	By Lemmas~\ref{mst_end_of_phase_one} and~\ref{mst_end_of_phase_2},
	the instance at the beginning of each restart is prediction mandatory free.
	We can apply Lemma~\ref{mst_phase2_2} and conclude that the number of elements that are re-matched is at most $2 \cdot k_h$.
	It follows that at most $2 \cdot k_h$ sets $\{e,h(e)\}$ can be queried in Line~\ref{line_mst_two_query_vc} and~\ref{line_mst_two_ensure_three_robust} and, {thus, in total} at most $4 \cdot k_h$ elements.
	
	We have that $\ALG = S \cup E_{ii} \cup E_{i}$, where $E_{i}$ denotes the set of elements that satisfy $(i)$ and $E_{ii}$ denotes the set of elements that satisfy $(ii)$.
	Since $|S| \le |\OPT|$, $|E_i| \le k_h$ and $|E_{ii}| \le 4 \cdot k_h$, we obtain a performance guarantee of $|\ALG| \le |\OPT| + {5} \cdot k_h$
\end{proof}

\begin{coro}
	\label{mst_error_sensitiv_guarantee}
Let $\ALG$ be the set of queries made by the algorithm that first executes Algorithm~\ref{ALG_mst_part_1}, and then Algorithm~\ref{ALG_mst_part_2} with recovery strategy B.
Then, $|\ALG| \le (1+\frac{1}{\gamma}) \cdot \opt + (5 + \frac{1}{\gamma}) \cdot k_h$.
\end{coro}

\begin{proof}
	Let $\ALG = \ALG_1 \cup P \cup \ALG_2$, where $\ALG_1$ denotes the queries of Algorithm~\ref{ALG_mst_part_1} without the last execution of Line~\ref{line_mst_one_fillup}, $P$ denotes the queries in the last execution of Line~\ref{line_mst_one_fillup} and $\ALG_2$ denotes the queries of Algorithm~\ref{ALG_mst_part_2}.
	Furthermore, let $\OPT = \OPT_1 \cup \OPT_2$ be an optimal query set with $\OPT_1 = \OPT \cap \ALG_1$ and $\OPT_2 = \OPT \setminus \ALG_1$.
	
	Since each element of $P$ is prediction mandatory, {the proof of Theorem~\ref{Theo_hop_distance_mandatory_distance}} implies that each such element is either mandatory or contributes {at least} one to the hop distance.
	From the arguments in the proof of Theorem~\ref{Theo_hop_distance_mandatory_distance}, it can be seen that, for each edge~$e$ that is prediction mandatory at some point but not mandatory, we have that $h'(e) \geq 1$, where $h'$ is defined as in the proof of~Lemma~\ref{mst_end_of_phase_2}.
	Thus, we have the following guarantees for the sets $\ALG_1$, $P$ and $\ALG_2$:
	\begin{itemize}
		\item For $\ALG_1$, the proof of Lemma~\ref{mst_end_of_phase_one} implies $|\ALG_1| \le (1+\frac{1}{\gamma}) \cdot (|\OPT_1| + h'(\ALG_1))$.
		\item 	For set $P$, we have $|P| \le |P \cap \OPT_2| + h'(P)$.
		\item For $\ALG_2$, we have that $|\ALG_2| \le |\OPT_2| + |E_i| + |E_{ii}| \le |\OPT_2| + h'(E_i) + 4 \cdot k_h$, where $E_i$ and $E_{ii}$ denote the sets of elements that satisfy $(i)$ and $(ii)$, respectively, as defined in the proof of Lemma~\ref{mst_end_of_phase_2}.
	\end{itemize}
	Summing up the guarantees, we get $$|\ALG| = |\ALG_1 \cup P \cup \ALG_2| \le (1+ \frac{1}{\gamma}) \cdot |\OPT| + (1+\frac{1}{\gamma}) (h'(\ALG_1) + h'(P) + h'(E_i)) + 4 \cdot k_h.$$ 
	By definition, it holds that $\ALG_1$, $P$ and $E_i$ are pairwise disjoint, which implies $h'(\ALG_1) + h'(P) + h'(E_i) \le k_h$. Thus, we can conclude 
	$$
		\ALG \le (1+ \frac{1}{\gamma}) \cdot \opt + (5+ \frac{1}{\gamma}) \cdot k_h.
	$$
\end{proof}

\section{Appendix for the Sorting Problem (Section \ref{sec:sorting})}
\label{app:sorting}
We present an analysis of Algorithm~\ref{fig:sorting1cons2rob}, obtaining the following theorem.
Remember that we assume that each interval is either trivial or open; however the algorithm can be modified to allow for closed and half-open intervals as in~\cite{halldorsson19sortingqueries}.
We say that an interval~$I_j$ {\em forces} a query in interval~$I_i$ if~$I_i$ is queried in Line~\ref{line:mandatoryproper}, \ref{line:mandatory} or~\ref{line:mandatory2} because $I_j \subseteq I_i$ or $w_j \in I_i$.

\sortingsingleset*

To prove that the algorithm indeed solves the problem, we must show that, at every execution of Line~\ref{line:looppaths}, every component of the intersection graph is a path.
We must have a proper interval graph, because every interval that contains another interval is queried in Line~\ref{line:mandatoryproper}, and every remaining interval that contains a queried value is queried in Line~\ref{line:mandatoryproper} or~\ref{line:mandatory}.
We claim that, at every execution of Line~\ref{line:looppaths}, the graph contains no triangles; for proper interval graphs, this implies that each component is a path, because the 4-star $K_{1,3}$ is a forbidden induced subgraph~\cite{wegner67properinterval}.
Suppose by contradiction that there is a triangle $abc$, and assume that $L_a \leq L_b \leq L_c$; it holds that $U_a \leq U_b \leq U_c$ because no interval is contained in another.
Since~$I_a$ and~$I_c$ intersect, we have that $U_a \geq L_c$, so $I_b \subseteq I_a \cup I_c$ and it must hold that $\pred{w}_b \in I_a$ or $\pred{w}_b \in I_c$, a contradiction since we query intervals that contain a predicted value in Line~\ref{line:containpredicted}.

We also need to prove that $(\mathcal{I}, \mathcal{E})$ is a forest of arborescences.
It cannot contain directed cycles because we avoid this in Line~\ref{line:addedge}.
Also, there cannot be two arcs with the same destination because we assign a single parent to each prediction mandatory interval.

\begin{theorem}
\label{thm_sorting_opt+kM}
Algorithm~\ref{fig:sorting1cons2rob} spends at most $\opt + k_M$ queries.
\end{theorem}

\begin{proof}
Fix an optimum solution $\OPT$.
Every interval queried in Lines~\ref{line:mandatoryproper} and~\ref{line:mandatory} is in $\OPT$.
Every interval queried in Line~\ref{line:containpredicted} that is not in $\OPT$ is clearly in $\mathcal{I}_P \setminus \mathcal{I}_R$.

For each path~$P$ considered in Line~\ref{line:looppaths}, let $P'$ be the intervals queried in Lines~\ref{line:pathodd}--\ref{line:patheven2}.
It clearly holds that $|P'| \leq |P \cap \OPT|$.
Finally, every interval queried in Line~\ref{line:mandatory2} is in $\mathcal{I}_R \setminus \mathcal{I}_P$ because we query all prediction mandatory intervals at the latest in Line~\ref{line:containpredicted}.
\end{proof}

The next result follows from Theorems~\ref{Theo_hop_distance_mandatory_distance} and~\ref{thm_sorting_opt+kM}.

\begin{theorem}
Algorithm~\ref{fig:sorting1cons2rob} spends at most $\opt + k_h$ queries.
\end{theorem}

For the remaining two theorems in this section, we need the following lemma.

\begin{lemma}
\label{lem:endpoints}%
Fix the state of $\mathcal{S}$ as in Line~\ref{line:fixS}.
For each interval $I_j$, let $S_j = \{ I_i \in \mathcal{S} : \pi(i) = j \}$.
For any path~$P$ considered in Line~\ref{line:looppaths} and any interval $I_j \in P$ that is not an endpoint of~$P$, it holds that $S_j = \emptyset$.
\end{lemma}

\begin{proof}
Suppose by contradiction that some $I_j \in P$ that is not an endpoint of~$P$ has $S_j \neq \emptyset$.
Let~$I_a$ and~$I_b$ be its neighbors in~$P$, and let $I_i \in S_j$.
We have that $\pred{w}_j \notin I_a \cup I_b$, otherwise~$I_a$ or~$I_b$ would have been queried in Line~\ref{line:containpredicted}.
Thus~$(I_i \cap I_j) \setminus (I_a \cup I_b) \neq \emptyset$, because $\pred{w}_j \in I_i$.
It is not the case that $I_i \subseteq I_j$ or $I_j \subseteq I_i$: If $I_i \subseteq I_j$, then~$I_j$ would have been queried in Line~\ref{line:mandatoryproper} before~$P$ is considered; if $I_j \subseteq I_i$, then~$I_i$ would have been queried in Line~\ref{line:mandatoryproper} and $I_i \notin \mathcal{S}$.
Therefore it must be that $I_i \subseteq I_a \cup I_j \cup I_b$, otherwise $I_a \subseteq I_i$ or $I_b \subseteq I_i$ (again a contradiction for $I_i \in \mathcal{S}$), since $(I_i \cap I_j) \setminus (I_a \cup I_b) \neq \emptyset$.
However, if $I_i \subseteq I_a \cup I_j \cup I_b$, then~$I_i$ would have forced a query in~$I_a$,~$I_b$ or~$I_j$ in Line~\ref{line:mandatory}, a contradiction.
\end{proof}

\begin{theorem}
Algorithm~\ref{fig:sorting1cons2rob} performs at most $\opt + k_{\#}$ queries.
\end{theorem}

\begin{proof}
Fix an optimum solution~$\OPT$.
We partition the intervals in~$\mathcal{I}$ into sets with the following properties.
One of the sets~$\tilde{S}$ contains intervals that are not queried by the algorithm.
We have a collection~$\mathcal{S}'$ in which each set has at most one interval not in $\OPT$.
Also, if it has one interval not in~$\OPT$, then we assign a prediction error to that set, in such a way that each error is assigned to at most one set.
(The interval corresponding to the prediction error does not need to be in the same set.)
Let~$\mathcal{I}'$ be the set of intervals with a prediction error assigned to some set in~$\mathcal{S}'$.
Finally, we have a collection~$\mathcal{W}$ such that for every $W \in \mathcal{W}$ it holds that $|\ALG \cap W| \leq |W \cap \OPT| + k_{\#}(W \setminus \mathcal{I}')$, where $k_{\#}(X)$ is the number of intervals in $X$ with incorrect predictions.
If we have such a partition, then it is clear that we spend at most $\opt + k_{\#}$ queries.

We begin by adding a different set to $\mathcal{S}'$ for each interval queried in Lines~\ref{line:mandatoryproper} and~\ref{line:mandatory}; all such intervals are clearly in $\OPT$, and we do not need to assign a prediction error.

Fix the state of $\mathcal{S}$ as in Line~\ref{line:fixS}.
To deal with the intervals queried in Line~\ref{line:containpredicted}, we add to~$\mathcal{S}'$ the set $S_j = \{I_i \in \mathcal{S} : \pi(i) = j\}$ for~$j = 1, \ldots, n$.
Note that each such set is a clique, because all intervals contain~$\pred{w}_j$.
Therefore, at most one interval in~$S_j$ is not in $\OPT$, and if that occurs, then $\pred{w}_j \neq w_j$, and we assign this prediction error to~$S_j$.

Let~$P = x_1 x_2 \cdots x_p$ with $p \geq 2$ be a path considered in Line~\ref{line:looppaths}, and let~$P'$ be the set of intervals in~$P$ that are queried in Lines~\ref{line:pathodd},~\ref{line:patheven1} or~\ref{line:patheven2}.
It clearly holds that $|P'| = \lfloor |P| / 2 \rfloor \leq |P \cap \OPT|$.
It also holds that at most $k_{\#}(P')$ intervals in~$P$ are queried in Line~\ref{line:mandatory2}: Each interval $I_j \in P'$ can force a query in at most one interval~$I_i$ in Line~\ref{line:mandatory2}, and in that case the predicted value of~$I_j$ is incorrect because $w_j \in I_i$ but $\pred{w}_j \notin I_i$, or $I_i$ would have been queried in Line~\ref{line:containpredicted}.
We will create a set $W \in \mathcal{W}$ and possibly modify~$\mathcal{S}'$, in such a way that $P \subseteq W$ and $P' \cap \mathcal{I}' = \emptyset$, so it is enough to show that
\begin{equation}
 |\ALG \cap W| \leq |W \cap \OPT| + k_{\#}(P').\label{eq:sortingpath}
\end{equation}
We initially take $W$ as the intervals in~$P$.
By Lemma~\ref{lem:endpoints}, it holds that $S_j = \emptyset$ for any $j \neq x_1, x_p$.
If $S_{x_1} \subseteq \OPT$, then we do not need to assign a prediction error to~$S_{x_1}$.
Otherwise, let~$I_i$ be the only interval in $S_{x_1} \setminus \OPT$.
The predicted value of $I_{x_1}$ is incorrect because $\pred{w}_{x_1} \in I_i$, and it must hold that $I_{x_1} \in \OPT$, or $\OPT$ would not be able to decide the order between~$I_{x_1}$ and~$I_i$.
If $x_1 \notin P'$, then we will not use its error in the bound of $|\ALG \cap W|$ if we prove Equation~(\ref{eq:sortingpath}).
Otherwise, we add~$I_i$ to~$W$ and remove it from~$S_{x_1}$, and now we do not need to assign a prediction error to~$S_{x_1}$.
We do a similar procedure for~$x_p$, and since at most one of $x_1, x_p$ is in~$P'$, we only have two cases to analyze: (1)~$W = P$, or (2)~$W = P \cup \{I_i\}$ with $\pi(i) \in \{x_1, x_p\}$.
\begin{enumerate}[(1)]
 \item $W = P$. Clearly $|\ALG \cap W| \leq |P'| + k_{\#}(P') \leq |W \cap \OPT| + k_{\#}(P')$.
 
 \item $W = P \cup \{I_i\}$, with $\pi(i) \in \{x_1, x_p\}$.
 Suppose w.l.o.g.\ that $\pi(i) = x_1$.
 Remember that $x_1 \in P'$, that $I_{x_1} \in \OPT$ and that its predicted value is incorrect.
 Since $x_1 \in P'$, it holds that $|P|$ is even and $x_2 \notin P'$.
 We have two cases.
 \begin{enumerate}
   \item $I_{x_2}$ is not queried in Line~\ref{line:mandatory2}.
   Then $I_{x_1}$ does not force a query in Line~\ref{line:mandatory2}, so
   \begin{eqnarray*}
   |\ALG \cap W| & \leq & |P' \cup \{I_i\}| + k_{\#}(P' \setminus \{x_1\}) \\
     & = & |P'| + 1 + k_{\#}(P' \setminus \{x_1\}) \\
     & \leq & |P \cap \OPT| + k_{\#}(P') \\
     & \leq & |W \cap \OPT| + k_{\#}(P').
   \end{eqnarray*}

   \item $I_{x_2}$ is queried in Line~\ref{line:mandatory2}.
   Then $I_{x_1}, I_{x_2} \in \OPT$, and $|\OPT \cap (P \setminus \{I_{x_1}, I_{x_2}\})| \geq |P' \setminus \{I_{x_1}\}|$ because $|P|$ is even.
   Therefore,
   \begin{eqnarray*}
   |\ALG \cap W| & \leq & |P' \cup \{I_{x_2}, I_i\}| + k_{\#}(P' \setminus \{I_{x_1}\}) \\
     & \leq & |P \cap \OPT| + 1 + k_{\#}(P' \setminus \{I_{x_1}\}) \\
     & \leq & |W \cap \OPT| + k_{\#}(P').   
   \end{eqnarray*}
 \end{enumerate}
\end{enumerate}

To conclude, we add the remaining intervals that are not queried by the algorithm to $\tilde{S}$.
\end{proof}

Now it remains to prove that the algorithm is 2-robust.
Consider the forest of arborescences $(\mathcal{I},\mathcal{E})$
that is constructed by Algorithm~\ref{fig:sorting1cons2rob}.
For each of these arborescences, the prediction mandatory intervals
contained in the arborescence are partitioned into cliques by
Lines~\ref{line:cliquepartitionbegin}--\ref{line:rootisolated}
of Algorithm~\ref{fig:sorting1cons2rob}.
Each of these clique partitions may contain a single clique of
size~$1$. As we would like to use the cliques in these clique
partitions as witness sets, cliques of size~$1$ require special
treatment. It turns out that the most difficult case is
where the clique of size~$1$ is formed by a prediction mandatory
interval $I_i$ that is the root of an arborescence. This happens if
the interval $I_{\pi(i)}$ that makes $I_i$ prediction mandatory
is a descendant of $I_i$ in that arborescence.
The following lemma shows that in this case we can revise the clique
partition of that arborescence in such a way that all cliques in that
clique partition have size at least~$2$.

\begin{lemma}
\label{lem:cliquerepartition}%
Consider an out-tree (arborescence) $T$ on a set of prediction mandatory
intervals, where an edge $(I_j,I_i)$ represents that $\pred{w}_j\in I_i$.
Let the root be $I_r$.
Let interval $I_m$ with $\pred{w}_m\in I_r$ be a descendant
of the root somewhere in $T$. Then the intervals in $T$ can be partitioned into
cliques (sets of pairwise overlapping intervals) in such a way
that all cliques have size at least~$2$.
\end{lemma}

\begin{proof}
We refer to the clique partition method of Lines~\ref{line:cliquepartitionbegin}--\ref{line:rootisolated} in Algorithm~\ref{fig:sorting1cons2rob}
as algorithm~CP.
This method will partition the nodes of an arborescence into cliques, each consisting either of a subset of the children of a node, or of a subset of the children of a node plus the parent of those children.
In the case considered in this lemma,
all cliques will have size at least~$2$, except that the clique containing the root of the tree may have size~$1$.

We first modify $T$ as follows: If there is a node $I_i$ in $T$ that is not a child of the root $I_r$ but contains $\pred{w}_r$, then we make $I_r$ the parent of $I_i$ (i.e., we remove the subtree rooted at $I_i$ and re-attach it below the root).
After this transformation, all intervals that contain $\pred{w}_r$ are children of $I_r$.

Apply CP to each subtree of $T$ rooted at a child of $I_r$.
For each of the resulting partitions, we call the clique containing
the root of the subtree the \emph{root clique} of that subtree.
There are several possible outcomes that can be handled directly:
\begin{itemize}
\item At least one of the clique partitions has a root clique
of size~$1$. In that case we combine all these root cliques
of size~$1$ with $I_r$ to form a clique of size at least~$2$, and we are done:
This new clique together with all remaining cliques from the clique
partitions of the subtrees forms the desired clique partition.

\item All of the clique partitions have root cliques of size at least~$2$,
and at least one of them has a root clique of size at least $3$.
Let $I_s$ be the root node of a subtree whose root clique has
size at least~$3$. We remove~$I_s$ from its clique and form
a new clique from $I_s$ and $I_r$, and we are done.

\item All of the clique partitions have root cliques of size exactly~$2$,
and at least one of the children $I_i$ of $I_r$ has $\pred{w}_i\in I_r$.
Then we add $I_r$ to the root clique that contains~$I_i$. We can do
this because all intervals in that root clique contain $\pred{w}_i$.
\end{itemize}

Now assume that none of these cases applies, so we have the following
situation:
All of the clique partitions have root cliques of size exactly~$2$,
and every child of $I_r$ has its predicted value outside $I_i$,
i.e., $\pred{w}_i\notin I_r$. In particular, $I_m$, the interval
that makes $I_r$ prediction mandatory, cannot be a child of $I_r$.

\begin{figure}[t]
\centerline{\scalebox{0.6}{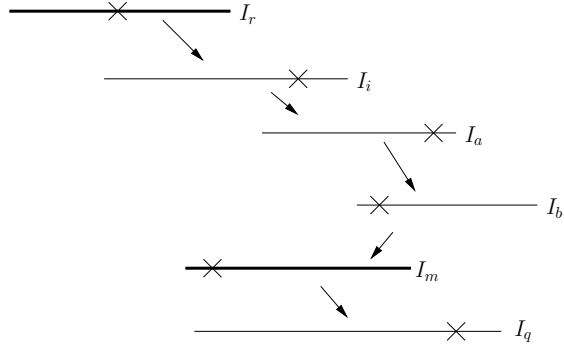}}

\caption{Illustration of path from $I_r$ to $I_m$'s child $I_q$ in $T$}
\label{fig:tree}
\end{figure}

Let $T'$ be the subtree of $T$ that is rooted at a child
of $I_r$ and that contains $I_m$. Let the root of $T'$
be $I_i$.

Observe that $I_i$ is the only interval in $T'$ that
contains $\pred{w}_r$, because all such intervals are
children of $I_r$ in~$T$. Assume w.l.o.g.\ that $\pred{w}_i$ lies
to the right of $I_r$. Then all intervals in~$T'$,
except for $I_i$, lie to the right of~$\pred{w}_r$.
See Figure~\ref{fig:tree} for an illustration of
a possible configuration of the path from
$I_r$ to $I_m$ (and a child $I_q$ of~$I_m$) in~$T$.

Now re-attach the subtree $T_m$ rooted at $I_m$
as a child of $I_r$ (ignoring the fact that $\pred{w}_r$ is
not inside $I_m$), and let $T_i=T'\setminus T_m$ denote
the result of removing $T_m$ from $T'$. Re-apply CP
to the two separate subtrees $T_m$ and $T_i$.
The possible outcomes are:
\begin{itemize}
\item The root clique of at least one of the two subtrees has size~$1$.
We can form a clique by combining $I_r$ with those (one or two)
root cliques of size~$1$.
As both $I_i$ and $I_m$ intersect $I_r$ from the right, the
resulting set is indeed a clique. Together with all other
cliques from the clique partitions of $T_m$ and $T_i$,
and those of the other subtrees of $I_r$ in $T$,
we obtain the desired clique partition.
\item The root cliques of both subtrees have size at
least~$2$. We add $I_r$ to the root clique of $T_m$.
That root clique contains only intervals that contain
$\pred{w}_m$, and $I_r$ also contains $\pred{w}_m$, so
we do indeed get a clique if we add $I_r$ to that
root clique. This new clique, together with all other
cliques from the clique partitions of $T_m$ and $T_i$,
and those of the other subtrees of $I_r$ in $T$,
forms the desired clique partition.
\end{itemize}

This concludes the proof of the lemma.
\end{proof}

\begin{theorem}
Algorithm~\ref{fig:sorting1cons2rob} is $2$-robust.
\end{theorem}

\begin{proof}
Fix an optimum solution $\OPT$.
We partition the input into a set $S$ of intervals that are not queried, plus a set~$S'$ of intervals in $\OPT$, plus a collection~$\mathcal{C}$ of sets with size at least~$2$ that are cliques in the initial intersection graph, plus a collection~$\mathcal{W}$ such that, for each $W \in \mathcal{W}$, the algorithm queries at most $2 \cdot |W \cap \OPT|$ intervals in~$W$.
If we have such a partition, then it is clear that we spend at most $2 \cdot \opt$ queries.

For every arborescence that meets the condition of Lemma~\ref{lem:cliquerepartition}, we take the revised
clique partition whose existence is guaranteed by that lemma and add all its cliques to $\mathcal{C}$. These
arborescences need no longer be considered in the rest of this proof.

We continue by adding all sets~$C_j$
computed by Algorithm~\ref{fig:sorting1cons2rob} to~$\mathcal{C}$, for all $j \in \{1, \ldots, n\}$ where $C_j \neq \emptyset$
and where $C_j$ has not been part of an arborescence that was already handled using Lemma~\ref{lem:cliquerepartition} in the previous
paragraph.
When building a set ~$C_j$ for $j=\pi(i)$ in Algorithm~\ref{fig:sorting1cons2rob},
we always pick a vertex~$I_i \in \mathcal{S}$ that is furthest from the root; therefore, if~$I_j$ is queried in Line~\ref{line:containpredicted}, then~$I_j$ will still be in~$\mathcal{S}$ when we pick~$I_i$ and build~$C_j$.
Thus, if a set $C_j \in \mathcal{C}$ has size~$1$, then~$I_j$ is not queried in Line~\ref{line:containpredicted}.

Now, if there is~$C_j \in \mathcal{C}$ of size~$1$ and~$I_j$ was queried in Line~\ref{line:mandatoryproper} or~\ref{line:mandatory}, then we include~$I_j$ in~$C_j$.

At this point, if there is a $C_j \in \mathcal{C}$ of size~$1$, then~$I_j$ must belong to
some path~$P$ that is a component of the dependency graph just before
Line~\ref{line:vcbeginsort} of Algorithm~\ref{fig:sorting1cons2rob} is executed
for the first time. If $P$ is a path of length one, i.e., consist only of~$I_j$, then we
add $I_j$ to~$C_j$ so that $C_j$ becomes a clique of size~$2$.
It remains to consider the case where $P$ is a path of length at least two.
Then~$P$ is considered in Line~\ref{line:looppaths} and,
by Lemma~\ref{lem:endpoints}, we have that~$I_j$ must be one of the endpoints of~$P$.
Let $P = x_1 x_2 \cdots x_p$.
We add all intervals in~$P$ to a set~$W$.
If~$|C_{x_1}| = 1$, then we make $W := W \cup C_{x_1}$ and $\mathcal{C} := \mathcal{C} \setminus \{C_{x_1}\}$.
We do a similar step if~$|C_{x_p}| = 1$.
Since~$P$ is a path, at least $\lfloor p/2 \rfloor$ of its intervals are in $P \cap \OPT$.
However, the graph induced by~$W$ may no longer be a path.
(For example, an interval in~$C_{x_1}$ may intersect~$I_{x_2}$ and~$I_{x_3}$, but not~$I_{x_4}$.)
Still, it is not hard to see that $|W \cap \OPT| \geq \lfloor |W|/2 \rfloor$ as well,
since any solution to the problem queries a vertex cover in the intersection graph.
If~$|W|$ is even, then we simply add~$W$ to~$\mathcal{W}$ and we are done; so assume that~$|W|$ is odd.
We divide the analysis in two cases:
\begin{enumerate}[(1)]
 \item $p$ is odd.
 Then the algorithm queries $I_{x_2}, I_{x_4}, \ldots, I_{x_{p-1}}$ in Line~\ref{line:pathodd}.
 Each of those intervals can force at most one query in Line~\ref{line:mandatory2}, therefore we have at least one interval in~$W$ that is never queried by the algorithm.
 Since $|W \cap \OPT| \geq \lfloor |W|/2 \rfloor$, clearly the algorithm queries at most $2 \cdot |W \cap \OPT|$ intervals in~$W$, and we include~$W$ in~$\mathcal{W}$.
 \label{case:odd}
 \item $p$ is even.
 Then either~$|C_{x_1}| = 1$ or~$|C_{x_p}| = 1$ (but not both), otherwise~$|W|$ would be even.
 Thus we have two subcases:
 \begin{enumerate}[(2a)]
  \item $W = P \cup C_{x_1}$.
  The algorithm queries $I_{x_1}, I_{x_3}, \ldots, I_{x_{p-1}}$ in Line~\ref{line:patheven1}.
  We begin by adding $C_{x_1} \cup \{x_1\}$ to~$\mathcal{C}$.
  If~$I_{x_1}$ forces a query in~$I_{x_2}$ in Line~\ref{line:mandatory2}, then we add~$I_{x_2}$ to~$S'$, and the remaining of~$P$ to a new set~$W'$ in~$\mathcal{W}$; this will be an even path because it was even to begin with and we remove $I_{x_1}$ and $I_{x_2}$, so we will be fine.
  Otherwise, the remaining of~$P$ is an odd path for which less than half intervals are queried in Line~\ref{line:patheven1}, and we proceed similarly as in case~(\ref{case:odd}).
  \label{case:evenx1}
  \item $W = P \cup C_{x_p}$.
  Then the algorithm queries $I_{x_2}, I_{x_4}, \ldots, I_{x_p}$ in Line~\ref{line:patheven2}, and the analysis is symmetric to the previous subcase.
 \end{enumerate}
\end{enumerate}

At this point, it is clear that there are no more sets of size~$1$ in~$\mathcal{C}$.
The remaining intervals that are not queried are simply included in~$S$, the remaining intervals queried in Lines~\ref{line:mandatoryproper} and~\ref{line:mandatory} are included in~$S'$, and we clearly obtain a partition of the intervals as desired.
\end{proof}

\section{Appendix for the experimental results (Section ~\ref{sec:exp})}
In this section we describe {in detail} the instance and prediction generation as used in the simulations of Section~\ref{sec:exp}.
In addition, Section~\ref{subsub:exp_mst} shows experimental results for the MST problem under uncertainty.

Our algorithms rely on finding minimum vertex covers. 
For the minimum problem, we solved the vertex cover problem using the \emph{Coin-or branch and cut (CBC)\footnote{\url{https://github.com/coin-or/Cbc}, accessed November 4, 2020.}} mixed integer linear programming solver.
Since our MST algorithms only solve the vertex cover problem in bipartite graphs, we determined minimum vertex covers using a standard augmenting path algorithm~\cite{AhujaMO1993book}.

\subsection{Experimental results for the minimum problem}
\label{subsub:exp_min}

We generated test instances by drawing interval sets from interval graphs. As source material we used instances of the \emph{boolean satisfiability problem (SAT)} from the rich SATLIB~\cite{hoos2000satlib} library.
A \emph{clause} of a SAT instance is a set of variables (with polarities) where each variable can be represented by its index, i.e., the variables are {numbered}.
We interpret each clause $c$ as an interval based on the indices of variables in $c$ (ignoring the polarities).
Each clause $c$ can be interpreted as the interval $(L_c, U_c)$ with $L_c = c_{\min} - \eps$ and $U_c = c_{\max} + \eps$ for a small $\eps > 0$, where $c_{\min}$ and $c_{\max}$ denote the lowest and highest variable index in $c$. 
In non-trivial SAT instances, the complexity of the problem is created by clauses sharing variables.
This in many instances leads to a high overlap between the interval representations of the clauses which makes the resulting interval graphs  interesting source material for the minimum problem under uncertainty. 

To generate instances for the minimum problem, we uniformly at random draw a sample of (not necessarily distinct) \emph{root intervals} from the re-interpreted SAT instance.
For each root interval, we add a \emph{root set} $S$ to the minimum problem instance which contains the root interval and {up to} $r_w$ intersecting intervals where $r_w$ is a parameter of the instance generation.
Note that the size of the root sets also depends on whether the source SAT instance contains enough intervals that intersect the root interval.
The number of intersecting intervals (between $1$ and $r_w$) and the intersecting intervals themselves are again drawn uniformly at random.
Note that we only generate preprocessed instances, i.e., instances where the leftmost interval $I_i$ of a set $S$ does not fully contain any interval $I_j \in S\setminus\{I_i\}$, and therefore might discard drawn intervals that would lead to the instance becoming non-preprocessed.
We generate only preprocessed instances as these are the difficult instances: a non-preprocessed instance gives all algorithms access to \enquote{free information} in form of queries that are obviously part of any feasible solution which can be very useful in solving the instance.

To ensure that the generated instances have an interesting underlying interval graph structure, that is, an interesting vertex cover instance (cf. Section~\ref{sec:minimum}), each root set $S$ is used as a starting point of paths with length {up to} $r_d$ in the vertex cover instance where $r_d$ is a parameter of the instance generation.
For each non-leftmost interval $I_S \in S$ we draw an integer $r_d'$ between $0$ and $r_d - 1$ that denotes the length of the path starting at $I_S$.
If $r_d' > 0$, we generate a set $S'$ with $I_S$ as leftmost interval and size of at most $r_w$ by again drawing the number of intersecting intervals (between $1$ and $r_w$) and the intersecting intervals themselves uniformly at random.
The procedure is repeated recursively in set $S'$ with parameter $r_d'< r_d$.
Note that the generation of these paths also depends on whether generating them is possible using the source SAT instance.
This part of the instance generation ensures a more complex underlying interval graph structure in the generated instances.
{In total, the family of sets $\mathcal{S}$ of a generated instance consists of all root sets and all sets $S'$ that are added in the recursive procedure.
The set of intervals $\mathcal{I}$ of a generated instance consists of all intervals that are added during the root set generation and all intervals added in the recursive procedure}

The instances were generated by drawing between $75$ and $150$ root clauses using values $r_w = 10$ and $r_d = 2$.
The resulting instances contain between $47$ and $287$ intervals, and between $15$ and $126$ sets.
Since the number of variables and clauses in the source SAT instances influences the probability with which the generated sets share intervals, we used SAT instances of different sizes, containing between $411$ and $32316$ clauses, and between $100$ and $8704$ variables.

Thus far, we only described how to generate the intervals and sets. 
{This paragraph} describes the generation of the true values.
To generate the true values, we start with initial true values that are placed in a standard form such that no elements are mandatory. 
Note that it is not always possible to set the true values such that no elements are mandatory. 
In such cases, we start with true values such that only a small number of elements {are} mandatory.
Afterwards, we uniformly at random draw the number of mandatory elements of the generated instance and then sequentially select random elements whose true values can be adjusted such that the number of mandatory elements increases until the determined number of mandatory elements is reached (if possible).
To set a true value such that the number of mandatory elements increases, we exploit Lemma~\ref{lema_mandatory_min}.
We generate the true values in this way because an instance of the minimum problem is essentially characterized by its vertex cover instance and its set of mandatory elements. 
Since the set of mandatory elements has an important effect on the instance, it makes sense to generate them in a way such that a wide range of different mandatory element sets is covered and especially mandatory element sets of a wide range of sizes are considered.

Regarding the generation of predictions, we observed that just picking predicted values uniformly at random (or by using a Gaussian distribution) does not lead to predictions that cover a wide range of relative errors $k_M/\opt$. 
Thus, we employed a more sophisticated prediction generation.

For each instance, we start with a target prediction error of $v=0$.
For this target prediction error, we generate predictions with an error of $k_M \approx v$. (An error of $k_M = v$ is not always possible.)
We repeatedly generate such predictions while increasing $v$ until we cannot find any predictions with $k_M \approx v$.
To generate predictions with an error of $k_M \approx v$ for a given target value $v$, we start with tentative predicted values that equal the true values.
Then, we determine the set of elements $F$ whose predicted values can be placed such that the error $k_M$ for the tentative predicted values increases.
The procedure uniformly at random draws an $e \in F$ and places $\w_e$ such that the error increases by at least one.
We repeat this procedure until the targeted error (or, if not possible, an error close to the targeted error) is reached.

After we iteratively generate predictions with an increasing error, we equally divide the interval $[0,v_{\max}]$ into $25$ bins of equal size, where $v_{\max}$ is the maximum prediction error of the generated predictions, and, for each bin, select the $5$ predictions with the highest error within the bin. 
The resulting $125$ sets of predictions are then used for our experiments.

\subsection{Experimental results for the MST problem}
\label{subsub:exp_mst}

For the MST problem we generated instances based on the symmetric traveling salesman problem instances of the TSPLIB\footnote{\url{http://comopt.ifi.uni-heidelberg.de/software/TSPLIB95/tsp/}, accessed November 3, 2020.}.
We considered graphs of up to $90$ vertices and $4000$ edges.
For TSP instances with more than $90$ vertices, we instead selected and used connected sub-graphs with $90$ vertices.
The TSPLIB instances already contain the graph structure and the true edge weights. 

We generate the interval boundaries as described in~\cite{focke17mstexp}.
For each edge $e$ with true value $w_e$, an interval boundary is set close to $w_e$ with a probability of $\frac{1}{2}$.
That is, with a probability of $\frac{1}{2}$ we set either $L_e = w_e - \eps$ or $U_e = w_e + \eps$ for a small $\eps>0$.
This encourages the generation of instances where an intersecting interval contains the true value of $e$ and increases the chance of generating instances with mandatory elements.
The interval boundaries that are not set afterwards are drawn uniformly at random within a ratio of $d$ around the true value $w_e$ where $d$ is a parameter of the instance generation.

Our test instances were generated using different values for $d$, between $0.05$ and $1$.
While~\cite{focke17mstexp} observes that small choices for $d$ lead to more difficult instances in terms of the competitive ratio and, in particular, utilizes the parameter $d = 0.065$ to generate difficult instances, a small choice for $d$ also leads to smaller generated instances.
This is because a small $d$ increases the probability of small uncertainty intervals and therefore the probability that intervals on a cycle do not intersect.
Thus, instances effectively become smaller. 
Since we want to observe the performance of our algorithms for different choices of $\gamma$, we rely on instances with a wider range of relative errors and, thus, on bigger instances.

For each of the $102$ graphs we generated $5$ instances of the MST problem under uncertainty without predictions and for each such generated instance we in turn generated $100$ predictions.
While the minimum problem required a more sophisticated prediction generation to cover a wide range of prediction errors, we were able to cover a large enough range for the MST problem by selecting the predicted values uniformly at random.
By repeating the random prediction generation long enough, we ensure that a large enough relative error range is covered.

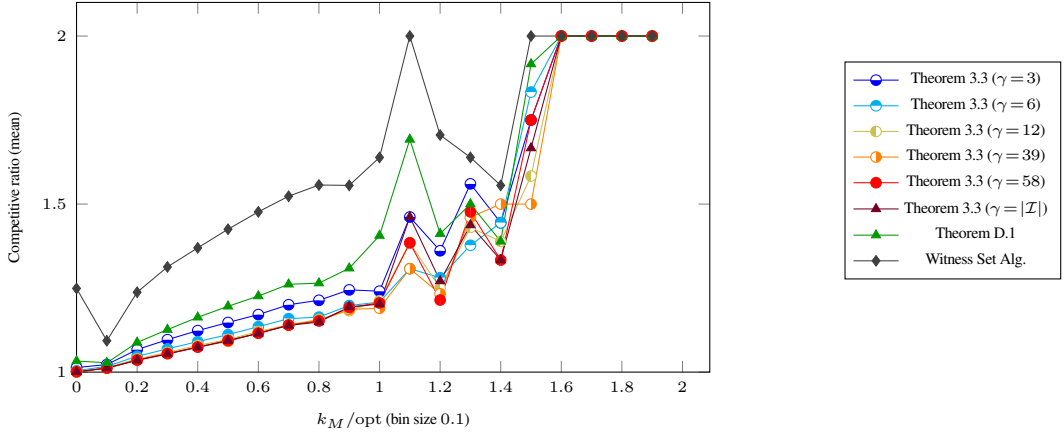
\begin{figure}[tb]
	\centering
	\begin{tikzpicture}
	    \begin{groupplot}[group style={group size= 1 by 3,horizontal sep = 1.25cm}]
	    	\nextgroupplot[
					,title style={font=\tiny}
					,xlabel={$k_M / \opt$ (bin size $0.1$)}
					,ylabel={Competitive ratio (mean)}
					,grid=minor
					,xmin = 0
					,ymin = 1
					,label style={font=\tiny},
					,tick label style={font=\tiny}  
					,legend entries={
									Theorem~\ref{theorem_mst2} ($\gamma = 3$),									
									Theorem~\ref{theorem_mst2} ($\gamma = 6$),										
									Theorem~\ref{theorem_mst2} ($\gamma = 12$),										
									Theorem~\ref{theorem_mst2} ($\gamma = 39$),
									Theorem~\ref{theorem_mst2} ($\gamma = 58$),										
									Theorem~\ref{theorem_mst2} ($\gamma = |\mathcal{I}|$),
									Theorem~\ref{theorem_mst1},
									Witness Set Alg.}
					,legend style={at={(axis cs:3.25,1.6)},anchor= east, font = \tiny}
					,table/col sep=comma,
					width=10cm,
					height=6.5cm]

				\addplot[blue,mark color=blue,mark=*,mark color=white,mark=halfcircle*,mark options={rotate=180}] table [x = Error_bin, y = MSTAlgorithmRecoveryB_003_cr] {exp_results/mst_error_sum.csv};	
				\addplot[cyan,mark color=white,mark=halfcircle*] table [x = Error_bin, y = MSTAlgorithmRecoveryB_006_cr] {exp_results/mst_error_sum.csv};	
				\addplot[yellow!75!black,mark color=white,mark=halfcircle*,mark options={rotate=90}] table [x = Error_bin, y = MSTAlgorithmRecoveryB_012_cr] {exp_results/mst_error_sum.csv};
				\addplot[orange,mark color=white,mark=halfcircle*,mark options={rotate=270}] table [x = Error_bin, y = MSTAlgorithmRecoveryB_039_cr] {exp_results/mst_error_sum.csv};	
				\addplot[red,mark color=red,mark=*,mark color=red,mark=halfcircle*,mark options={rotate=90}] table [x = Error_bin, y = MSTAlgorithmRecoveryB_058_cr] {exp_results/mst_error_sum.csv};	
				\addplot[purple!60!black,mark color=purple!60!black,mark=triangle*] table [x = Error_bin, y = MSTAlgorithmRecoveryB_195_cr] {exp_results/mst_error_sum.csv};	
				\addplot[green!60!black,mark color=green!60!black,mark=triangle*] table [x = Error_bin, y = MSTAlgorithmRecoveryA_cr] {exp_results/mst_error_sum.csv};	
	 			\addplot[darkgray,mark color=darkgray,mark=diamond*] table [x = Error_bin, y = CycleAlgorithm_cr] {exp_results/mst_error_sum.csv};
	    \end{groupplot}
		
	\end{tikzpicture}
	\caption{Experimental results for the MST problem under uncertainty. Instances and predictions were grouped into equal size bins according to their relative error.}
	\label{fig_mst_experiments}
\end{figure}

{ Figure~\ref{fig_mst_experiments} shows the results of the over $20,000$ simulations (instance and prediction pairs).
The figure compares the results of our prediction-based algorithms of Theorems~\ref{theorem_mst2} and~\ref{theorem_mst1} for different choices of the parameter $\gamma$ with the standard {\em witness set algorithm}. The latter sequentially resolves cycles by querying witness {sets} of size two and achieves the best possible competitive ratio of $2$ without predictions~\cite{erlebach08steiner_uncertainty}.
}

{Our prediction-based algorithms outperform the witness set algorithm for every relative error up to $1.4$.
For higher relative errors, the prediction-based algorithms match the performance of the witness set algorithm. 
	
Further, the parameter $\gamma$ reflects the robustness-performance tradeoff in the sense that the curves for different choices of $\gamma$ intersect and high values for $\gamma$ perform better for smaller relative errors, while smaller values for $\gamma$ perform better for high relative errors.}
The performance gap between the different values for $\gamma$ appears less significant for small relative errors, 
which suggests that selecting $\gamma$ not too close to the maximum value $|\mathcal{I}|$ might be beneficial.

In contrast to the results for the minimum problem (cf. Figure~\ref{fig_experiments}), the plots for the prediction-based algorithm are not monotonously increasing and instead contain jumps for higher relative errors. 
This is because the instances more strongly vary for the different error bins. 
In particular, most of the instances with predictions that lead to higher relative errors are small in the sense that $\opt$ is small. 
The variation in the instances between different error bins leads to jumps in the plots.
For relative errors of at least $1.6$ the instances are small enough such that the different choices for $\gamma$ essentially behave the same. 

\newpage
\small

\bibliographystyle{abbrv} 
\bibliography{paperv2}

\end{document}